\documentclass[10pt, journal, letterpaper]{IEEEtran}
  
\def\BibTeX{{\rm B\kern-.05em{\sc i\kern-.025em b}\kern-.08em
    T\kern-.1667em\lower.7ex\hbox{E}\kern-.125emX}}

\usepackage{hyperref}
\usepackage{soul}
\usepackage[dvipsnames]{xcolor}
\usepackage{cite}
\usepackage{setspace} % for \setstretch
\usepackage{subcaption}
  \usepackage[pdftex]{graphicx}
  \usepackage[underline=true,rounded corners=false]{pgf-umlsd}
\usepackage[cal=cm]{mathalfa}
\usepackage{url}
\usepackage{listings}
\usepackage{mathtools}
\usepackage[noend]{algpseudocode}
\usepackage{amssymb, amsmath}
\usepackage{amsthm}
\usepackage{physics}
\usepackage{nccmath}
\usepackage{dsfont}
\usepackage[Algorithm]{algorithm}
\usepackage{comment}
\usepackage{svg}
\usepackage{soul}
\usepackage{afterpage}
\usepackage{booktabs}
\usepackage{amsfonts}
\usepackage{soul}
\usepackage{cancel}
\usepackage{balance}
\usepackage{enumerate}
\usepackage{multirow}
\usetikzlibrary{positioning}
\newcommand{\keepcomment}{0} % 1 - Keep comments, 0 - Hide comments
\newcommand{\isextended}{1} % 1 - keep the appendix, 0 - cite the technical report

\algnewcommand{\IfThen}[2]{% \IfThenElse{<if>}{<then>}{<else>}
  \State \algorithmicif\ #1\ \algorithmicthen\ #2}

\ifnum\keepcomment=1
	\newcommand{\fabio}[1]{{\leavevmode\color{PineGreen}{[FAB: #1]}}}
	\newcommand{\gabriele}[1]{{\leavevmode\color{DarkOrchid}{[GAB: #1]}}}
	\newcommand{\giovanni}[1]{{\leavevmode\color{brown}{[GIO: #1]}}}
	\newcommand{\tareq}[1]{{\leavevmode\color{blue}{[TAR: #1]}}}
	\newcommand{\note}[1]{\noindent\textcolor{gray}{#1}}
	\newcommand{\delete}[1]{\leavevmode\sout{#1}}
    \newcommand{\gc}[1]{\textcolor{DarkOrchid}{#1}}

	\newcommand{\gcdelete}[1]{\textcolor{DarkOrchid}{\leavevmode\sout{#1}}}
	\newcommand{\gndelete}[1]{\textcolor{brown}{\leavevmode\sout{#1}}}
    \newcommand{\stkout}[1]{\ifmmode\text{\sout{\ensuremath{#1}}}\else\sout{#1}\fi}
    \newcommand\rev[3]{\textcolor{red}{\begin{scriptsize}{#1}\end{scriptsize}\stkout{#2}}\textcolor{blue}{#3}}
\else
    \newcommand{\fabio}[1]{\leavevmode\ignorespaces\unskip}
	\newcommand{\gabriele}[1]{\leavevmode\ignorespaces\unskip}
	\newcommand{\giovanni}[1]{\leavevmode\ignorespaces\unskip}
	\newcommand{\tareq}[1]{\leavevmode\ignorespaces\unskip}
	\newcommand{\fda}[1]{\leavevmode\ignorespaces\unskip}

    \newcommand{\gc}[1]{\ignorespaces#1\ignorespaces\unskip}
    
    \newcommand{\note}[1]{\leavevmode\ignorespaces\unskip}
	\newcommand{\delete}[1]{\leavevmode\ignorespaces\unskip}
	\newcommand{\gcdelete}[1]{\leavevmode\ignorespaces\unskip}
	\newcommand{\gndelete}[1]{\leavevmode\ignorespaces\unskip}
	\newcommand\rev[3]{\ignorespaces#3\ignorespaces\unskip}
\fi

\newtheorem{theorem}{Theorem}[section]

\newtheorem{proposition}{Proposition}[theorem]

\newtheorem{lemma}[theorem]{Lemma}

\DeclareMathOperator*{\argmax}{\arg\max}
\DeclareMathOperator*{\argmin}{\arg\min}

\algnewcommand{\IIf}[1]{\State\algorithmicif\ #1\ \algorithmicthen}
\algnewcommand{\ElseIIf}[1]{\State\algorithmicelse\ }
\algnewcommand{\EndIIf}{\algorithmicend\ \algorithmicif\ }
\algnewcommand{\EndIIff}{\algorithmicend\ \algorithmicif\ }
\newcommand{\advSet}[0]{\mathcal{A}}
\renewcommand{\vec}[1]{\pmb{{#1}}}

\newcommand{\OnlineGreedy}[0]{\textsc{OLAG}}
\newcommand{\SG}[0]{\textsc{SG}}

\newcommand{\consSet}{\mathcal{X}}
\newcommand{\alloc}[0]{x}
\newcommand{\allocVec}[0]{\vec{\alloc}}
\newcommand{\allocVecFrac}[0]{\vec{y}}
\newcommand{\allocFrac}[0]{y}
\newcommand{\convexhull}[1]{\mathrm{conv}\left(#1\right)}
\newcommand{\auxalloc}[0]{z}
\newcommand{\Auxalloc}[0]{Z}

\newcommand{\repo}[0]{\omega}
\newcommand{\auxload}[0]{\lambda}

\newcommand{\repoVec}[0]{\vec{\repo}}
\newcommand{\requestPath}{p}
\newcommand{\requestPathVec}{\vec{\requestPath}}
\newcommand{\requestBatch}[0]{r}
\newcommand{\requestBatchVec}[0]{\vec{\requestBatch}}
\newcommand{\requestWithPath}{\rho}
\newcommand{\request}{\rho}
\newcommand{\edges}{\mathcal{E}}
\newcommand{\vertices}{\mathcal{V}}
\newcommand{\weight}{w}
\newcommand{\taskcatalogCard}{N}
\newcommand{\taskcatalog}{\mathcal{\taskcatalogCard}}

\newcommand{\qualityset}{\mathcal{Q}}
\newcommand{\timehorizon}{T}
\newcommand{\requestSet}{\mathcal{R}}
\newcommand{\capacity}{b}
\newcommand{\loadmodel}{L}
\newcommand{\load}{l}
\newcommand{\loadVec}{\vec{l}}

\newcommand{\supp}[1]{\mathrm{supp} (#1)}
\newcommand{\smallCost}{\gamma}
\newcommand{\systemCost}{C}
\newcommand{\systemGain}{G}

\newcommand{\modelSet}{\mathcal{M}}
\newcommand{\reals}{\mathbb{R}}
\newcommand{\naturals}{\mathbb{N}}
\newcommand{\modelsNo}{K}

\newcommand{\subgrad}{g}
\newcommand{\subgradVec}{\vec{\subgrad}}

\DeclarePairedDelimiter{\floor}{\lfloor}{\rfloor}
\newcommand{\DepRound}{\textsc{DepRound}}
\newcommand{\AlgoName}{\textsc{INFIDA}}
\newcommand{\AlgoNameOffline}{$\textsc{INFIDA}_{\textsc{Offline}}$}

\newcommand{\AlgoFullName}{INFerence Intelligent Distributed Allocation}

\newcommand{\allocSet}{\mathcal{X}}
\newcommand{\allocSetFrac}{\mathcal{Y}}
\newcommand{\El}{\Lambda}
\newcommand{\loadSet}{\mathcal{L}}
\newcommand{\requestBatchSet}{\mathcal{B}}

\newcommand{\BigO}[1]{\mathcal{O} \left(#1\right)}

\newcommand{\Simplify}{\textsc{Simplify}}

\newcommand\vGroup[2]{\underset{#2}{\underbrace{#1} } }

\usepackage{array}
\newcolumntype{C}[1]{>{\raggedright\arraybackslash}p{#1}}

\newcommand{\newton}[1]{\textcolor{black}{#1}}

\newcommand{\set}[1]{\left\{#1\right\}}
\newcommand{\parentheses}[1]{\left(#1\right)}
\newcommand{\card}[1]{|#1|}

\makeatletter
\DeclareRobustCommand{\change}{%
  \@bsphack
  \leavevmode
  \normalcolor
  \@esphack
}
\DeclareRobustCommand{\stopchange}{%
  \@bsphack
  \normalcolor
  \@esphack
}
\makeatother

\allowdisplaybreaks

\begin{document}

\title{Towards Inference Delivery Networks: \\ Distributing Machine Learning \\with Optimality Guarantees}
\date{}

\author{ 
	\IEEEauthorblockN{
		Tareq {Si Salem}\IEEEauthorrefmark{1},
		Gabriele Castellano\IEEEauthorrefmark{1}\IEEEauthorrefmark{2},
		Giovanni Neglia\IEEEauthorrefmark{1},
		Fabio Pianese\IEEEauthorrefmark{2},
		Andrea Araldo\IEEEauthorrefmark{3}
		\\
		\IEEEauthorblockA{
			\IEEEauthorrefmark{1}Inria, Universit\'e C\^ote d'Azur, France, \{tareq.si-salem, gabriele.castellano, giovanni.neglia\}@inria.fr,
		}\\
		\IEEEauthorblockA{
			\IEEEauthorrefmark{2}Nokia Bell Labs, France, \{fabio.pianese, gabriele.castellano.ext\}@nokia.com,
		}\\
		\IEEEauthorblockA{
			\IEEEauthorrefmark{3}Samovar, T\'el\'ecom SudParis, Institut Polytechnique de Paris, France, andrea.araldo@telecom-sudparis.eu
		}
		%Email: \IEEEauthorrefmark{1}\{tareq.si-salem, gabriele.castellano, giovanni.neglia\}@inria.fr,\\
		%\IEEEauthorrefmark{2}fabio.pianese@nokia.com,
		%\IEEEauthorrefmark{3}andrea.araldo@telecom-sudparis.eu
	}
}

\maketitle
\bstctlcite{IEEEexample:BSTcontrol}

\ifnum\isextended=1
% page numbers
\thispagestyle{plain}
\pagestyle{plain}
\fi

\begin{abstract}
An increasing number of applications rely on complex inference tasks that are based on machine learning (ML). 
Currently, there are two options to run such tasks: either they are served directly by the end device (e.g., smartphones, IoT equipment, smart vehicles), or offloaded to a remote cloud.
Both options may be unsatisfactory for many applications: local models may have inadequate accuracy, while the cloud may fail to meet delay constraints.
In this paper, we present the novel idea of \emph{inference delivery networks} (IDNs), networks of computing nodes that coordinate to satisfy ML inference requests achieving the best trade-off between latency and accuracy. IDNs bridge the dichotomy between device and cloud execution by integrating inference delivery at the various tiers of the infrastructure continuum (access, edge, regional data center, cloud).
We propose a distributed dynamic policy for ML model allocation in an IDN by which each node dynamically updates its local set of inference models based on requests observed during the recent past plus limited information exchange with its neighboring nodes. Our policy offers strong performance guarantees in an adversarial setting and shows improvements over greedy heuristics with similar complexity in
realistic scenarios.
\end{abstract}

%\begin{IEEEkeywords}

%\end{IEEEkeywords}

% \fi
% \IEEEpeerreviewmaketitle

% \setlength{\textfloatsep}{0pt}

\section{Introduction}
\label{sec:introduction}

Machine learning (ML) models are often trained to perform inference, that is to elaborate predictions based on input data. ML model training is a computationally and I/O intensive operation and its streamlining is the object of much research effort. Although inference does not involve complex iterative algorithms and is therefore generally assumed to be easy, it also presents fundamental challenges that are likely to become dominant as ML adoption increases~\cite{stoica2017berkeley}. In a future where AI systems are ubiquitously deployed and need to make timely and safe decisions in unpredictable environments, inference requests will have to be served in real-time and the aggregate rate of predictions needed to support a pervasive ecosystem of sensing devices will become overwhelming. 

Today, two deployment options for ML models are common: inferences can be served by the end devices (smartphones, IoT equipment, smart vehicles, etc.), where only simple models can run, or by a remote cloud infrastructure, where powerful ``machine learning as a service'' (MLaaS) solutions 
%by the major cloud providers 
rely on sophisticated models and provide  inferences at extremely high throughput.

% Currently, there are two main deployment options for ML inferences: either they are served directly by the end device, e.g., smartphones, IoT equipment, smart vehicles, where only simple models can run, or by the remote cloud, where the big cloud players---Amazon, Microsoft, and Google---have all started pushing their ``machine learning as a service'' (MLaaS) solutions.
However, there exist applications for which both options may be unsuitable: local models may have inadequate accuracy, while the cloud may fail to meet delay constraints.
As an example, popular applications such as recommendation systems, voice assistants, and ad-targeting, need to serve predictions from ML models in less than 200~ms. Future wireless services, such as connected and autonomous cars, industrial robotics, mobile gaming, augmented/virtual reality, have even stricter latency requirements, often below 10~ms and in the order of 1~ms for what is known as the tactile Internet~\cite{simsek16}.
%With the advent of Edge Computing, computational resources (memory, storage, CPUs, GPUs) are also deployed at the edge of the networks (at base stations, access points, or ad-hoc servers), but those resources will still be very limited in comparison to the cloud and need to be wisely used. In conclusion, inference will require complex resource orchestration across users' devices, edge computing servers, and the cloud.
In enabling such strict latency requirements, the advent of Edge Computing plays a key role, as it deployes computational resources at the edge of the network (base stations, access points, ad-hoc servers). However, edge resources have limited capacity in comparison to the cloud and need to be wisely used. Therefore, integrating ML inference in the continuum between end devices and the cloud---passing through edge servers and regional micro data-centers---will require complex resource orchestration.

%We believe then that it is fundamental to study how to integrate ML inference in the continuum between the user device and the cloud---passing through edge servers and regional micro data-centers---to achieve the best accuracy/latency/resource-utilization trade-offs adapted to the requirements of the specific application. 
%In fact, inference accuracy and, in general, resource efficiency increase toward the cloud (larger ML models can be placed with resource economies thanks to scale effects) but latency and communication-resources usage increase as well. 
We believe that, to allocate resources properly, it will be crucial to study the trade-offs between accuracy, latency and resource-utilization, adapted to the requirements of the specific application. In fact, inference accuracy and, in general, resource efficiency increase toward the cloud, but so does communication latency.
In this paper, 
%{In our previous work~\cite{salem2021towards}} 
we present the novel idea of \emph{inference delivery networks} (IDN): networks of computing nodes that coordinate to satisfy inference requests achieving the best trade-off. 
An IDN may be deployed directly by the ML application provider, or by new IDN operators that offer their service to different ML applications, similarly to what happens for content delivery networks. 
The same inference task can be served by a set of heterogeneous models featuring diverse performance and resource requirements (e.g., different model architectures~\cite{howard17mobilenets}, multiple downsized versions of the same pre-trained model~\cite{deng20}, different configurations and execution setups).
Therefore, we study the novel problem of how to deploy the available ML models on the available IDN nodes, where a deployment strategy consists in two coupled decisions: \textit{(i)}~where to place models for serving a certain task and \textit{(ii)}~how to select their size/complexity among the available alternatives.
%The strategy must jointly \textit{(i)} meet models and nodes capacity and \textit{(ii)} optimize the latency/accuracy trade-off, while continuously adapting to changes in requests load and tasks popularity.
% ---- GABRIELE: PROPOSED ALTERNATIVE PIECE [END]

%In this paper, after defining a specific optimization problem and characterizing its complexity, we introduce \AlgoName{} (\AlgoFullName{}), a distributed dynamic allocation policy.\footnote{{A first version of \AlgoName{} was previously presented in~\cite{salem2021towards}.}}
In this paper, we first define a specific optimization problem for ML model allocation in IDNs.
%; {compared to our previous work \cite{salem2021towards}, we here improve the description of the system model with a new formalization that is both more intuitive and rigorous. 
We characterize the complexity of such problem and then introduce \AlgoName{} (\AlgoFullName{}), a distributed dynamic allocation policy.
Following this policy, each IDN node periodically updates its local allocation of inference models on the basis of the requests observed during the recent past and limited information exchange with its neighbors. 
The policy offers strong performance guarantees in an adversarial setting~\cite{hazan2016introduction}, that is a worst case scenario where the environment %\gcdelete{(requests and ML model capacities)} 
evolves in the most unfavorable way. 
Numerical experiments in realistic settings show that our policy outperforms heuristics with similar complexity. Our contributions are as follows:
%obtained adapting dynamic service caching policies.

%In summary, our contributions are as follows:
\begin{enumerate}[{(1)}]
    \item We present the novel idea of inference delivery networks {(IDNs)}.
    \item We frame the allocation of ML models in IDNs as an (NP-hard) optimization problem that captures the trade-off between latency and accuracy, \change and study how this problem diverges from settings considered in previous works %known settings 
    (Sec.~\ref{sec:design}).\stopchange
    \item We propose \AlgoName, a distributed and dynamic allocation algorithm for IDNs (Sec.~\ref{sec:algorithm}), and we show it provides strong guarantees in the adversarial setting, \change providing novel theoretical results in approximating budget-additive (submodular) set functions (Sec.~\ref{sec:guarantees}).\stopchange
    \item We evaluate \AlgoName{} in a realistic {simulation} scenario and compare its performance {both with an offline greedy heuristic and with its online variant under different topologies and} trade-off settings (Sec.~\ref{sec:validation}).
\end{enumerate}
%The paper is organized as follows. In Sec.~\ref{sec:related-work}) we discuss related work and other relevant background with respect to IDNs. We formalize the problem of allocating ML model in IDNs in Sec.~\ref{sec:design}. We introduce our \AlgoName{} online heuristic in (Sec.~\ref{sec:algorithm}) and discuss its theoretical guarantees in Sec.~\ref{sec:guarantees}. We show our experimental results in Sec.~\ref{sec:validation}.

%!TEX root = ../paper.tex
\section{Related Work}
\label{sec:related-work}

The problem of machine learning is often reduced to the training task, i.e., producing statistical models that can map input data to certain predictions. A considerable amount of existing works addresses the problem of model training: production systems such as Hadoop~\cite{babu2013massively} and Spark~\cite{zaharia2012resilient} provide scalable platforms for analyzing large amount of data on centralized systems, and even the problem of distributing the training task over the Internet has been largely addressed recently by many works on federated learning~\cite{fed1,fed2, konevcny2016federateda,konevcny2016federatedb, wu2020accelerating,neglia2019role}. However,  there is surprisingly less research on how to manage the deployment of ML models once they have been trained (inference provisioning).

%Inference provisioning is a well studied problem in the context of centralized systems.
Most of the existing solutions on inference provisioning (e.g., Tensorflow Serving~\cite{olston2017tensorflow}, Azure ML~\cite{chappell2015introducing}, and Cloud ML~\cite{enginegoogle}) address the scenario where inference queries are served by a data center. Recent works \cite{crankshaw2017clipper, mao2019learning, romero2019infaas, crankshaw2020inferline} propose improvements on  performance and usability of such cloud inference systems. Clipper~\cite{crankshaw2017clipper} provides a generalization of TensorFlow Serving~\cite{olston2017tensorflow} to enable the usage of different ML frameworks, such as Apache Spark MLLib~\cite{meng2016mllib}, Scikit-Learn~\cite{jia2014caffe}, and 
Caffe~\cite{pedregosa2011scikit}. The auhtors of~\cite{mao2019learning} propose a reinforcement learning scheduler to improve the system throughput. INFaaS~\cite{romero2019infaas} provides a real-time scheduling of incoming queries on available model variants, and scales deployed models based on load thresholds. Last, InferLine~\cite{crankshaw2020inferline} extends Clipper to minimize the end-to-end latency of a processing pipeline, both periodically adjusting the models allocation and constantly monitoring and handling unexpected query spikes; the solution can be applied to any inference serving system that features a centralized queue of queries.
All these solutions address the problem of inference provisioning in the scenario where the requests are served within a data center and are not suitable for a geographically distributed infrastructure where resources are grouped in small clusters and network latency is crucial (e.g., Edge Computing). For instance, none of the previous works consider the network delay between different compute nodes, it being negligible in a data center.

{For what concerns inference provisioning in constrained environments, fewer works exist.
%A related set of works attempts 
Some solutions attempt} to adapt inference to the capabilities of mobile hardware platforms through the principle of model splitting, a technique that distributes a ML model by partitioning its execution across multiple discrete computing units. Model splitting was applied to accommodate the hardware constraints in multi-processor mobile devices~\cite{deepX}, to share a workload among mobile devices attached at the same network edge~\cite{Fernando2019},
%,Zhao18}, 
and to partially offload inferences to a remote cloud infrastructure~\cite{kang2017neurosurgeon}, possibly coupled with early exit strategies~\cite{Branchynet} and conditional hierarchical distributed deployment~\cite{teerapittayanon2017distributed}. Model splitting is orthogonal to our concerns and could be accounted for in an enhanced IDN scheme.

There has been some work on ML model placement at the edge in the framework of what is called ``AI on Edge''~\cite{Deng2020}, but it considers a single intermediate tier between the edge device and the cloud, while we study general networks with nodes in the entire cloud-to-the-edge continuum. Our dynamic placement \AlgoName{} algorithm could be applied also in this more particular setting, for example in the MODI platform~\cite{Ogden2018}. The work closest to ours is~\cite{LaiJiao2020}, which proposes an online learning policy, with the premise of load balancing over a pool of edge devices while maximizing the overall accuracy. \AlgoName{} has more general applicability, as we do not make any assumption on the network topology and also employ a more flexible cost model that takes into account network delay.
{Another related work in this framework is} VideoEdge~\cite{hung2018videoedge}, which studies how to split the analytics pipeline across different computational clusters to maximize the average inference accuracy. Beside the focus on the  specific video application, the paper does not propose any dynamic allocation placement algorithm.

Even if the problem of inference provisioning is currently overlooked in the context of distributed systems, there exists a vast literature on {the problem of} content placement~\cite{Qiu2015} where objects can be stored (cached) into different nodes in order to reduce the operational cost of content delivery. \newton{Content placement has been extended to the case of \emph{service caching} (or placement), where an entire service can be offloaded onto nodes co-located with base-stations or mobile-micro clouds, engaging not only storage but also computational resources and energy~\cite{xu18,Leung2017}.  The similarities between this problem and inference provisioning inspired us in the design of Inference Delivery Networks. However, the two problems feature  crucial differences. First, in a content delivery network a request for a given item may only be served by a server storing that specific item. Whereas, in IDNs, several models can provide an answer but the accuracy of the answer can be different~\cite{Ben-Ameur2021}. 
Second, 
%First, 
a key property of content placement is that the service cost always increases together with the path length, i.e., the distance between the request source and the node serving the request. 
This is not the case for inference delivery networks, as 1)~upstream models may be more accurate and 2)~the same model at may feature different processing delays based on the serving node properties. 
%This is not the case for inference delivery networks, as models may feature different processing delays based on the serving node properties.
This leads to a more complex cost function, as the first node receiving the request may not be the optimal one to serve it. 
Note that this difference was crucial in the design of our algorithm (see Fig.~\ref{fig:synch}). %Second, services in content placement  provide responses on a request-by-request basis. Whereas, in IDNs, several models can provide an answer but the accuracy of the answer can be different~\cite{Ben-Ameur2021}. 
Finally, multiple requests can simultaneously be processed by a given model, leading to additional considerations about requests load and serving capacities.}

A similar trade-off between resource usage and perceived quality typically emerges in the context of video caching~\cite{JoongheonKim2018,ye2017quality,zhan2017content,araldo2016representation,SongGuo2020,poularakis2014video,poularakis2016caching}, where the same video can be cached into multiple network nodes at different qualities (or resolutions): the operator optimizes the user experience by jointly deciding the placement of videos and their quality.
These works either maximize the video quality perceived by the user~\cite{JoongheonKim2018,araldo2016representation,SongGuo2020}, minimize the download time~\cite{zhan2017content,poularakis2016caching} and the backhaul traffic~\cite{ye2017quality}, or minimize a combined cost~\cite{poularakis2014video}. Although some of the models in these papers may be adapted to inference provisioning in IDNs, these works in general study static optimization problems under a known request process and consider simple network topologies: a single cache~\cite{ye2017quality,zhan2017content}, a pool of parallel caches~\cite{JoongheonKim2018,poularakis2016caching}, bipartite networks~\cite{poularakis2014video,SongGuo2020}.
The only exception is~\cite{araldo2016representation}, which considers an arbitrary topology and provides some online heuristics, but ignores the service latency, which is of paramount importance when placing interactive ML models (e.g., for applications like augmented reality or autonomous driving).
Instead, we propose a dynamic policy that jointly optimizes inference quality and latency and provides
%is capable to adapt to varying requests in every time slot. We also prove 
strong performance guarantees without requiring any knowledge about the request process
%, which are meant to hold 
 thanks to our adversarial setting (Sec.~\ref{sec:guarantees}).

\newton{ Adversarial analysis is typically studied through the lens of online convex optimization (OCO)~\cite{hazan2016introduction}. OCO models can be tackled with well-understood learning algorithms~\cite{hazan2016introduction,ShalevOnlineLearning, mcmahan2017survey}. However, the problem of optimizing the allocation of ML models in IDNs diverges from the template of OCO. In particular, the decision set and the cost functions  are non-convex because the allocation decisions are not continuous; moreover, as we show in Appendix \ref{appendix:np_hard}, computing the optimal allocation is NP-hard contrarily to the OCO setting. In our work, we generalize findings from~\cite{stratis, NEURIPS2021_2387337b,shanmugam2013femtocaching}, and provide novel results to approximate budget-additive (submodular) set functions~\cite{budget_additive}, which are of independent interest beyond this work (e.g., for online advertising~\cite{TCS-057,sub_o2}, market equilibrium~\cite{sub-m1,sub_m2}). }

Finally, ML model allocation in an IDN can also be considered as a particular instance of \emph{similarity caching}~\cite{garetto20infocom}, a general model where items and requests can be thought as embedded in a metric space: edge nodes can store a set of items, and the distance between a request and an item determines the quality of the matching between the two. 
%Similarity caching was applied to a number of applications including content-based image retrieval~\cite{falchi2008metric}, contextual advertising~\cite{pandey2009nearest}, object recognition~\cite{drolia2017cachier,drolia2017precog,guo2018foggycache,venugopal2018shadow}, and recommender systems~\cite{Sermpezis18,costantini20}.
Similarity caching was applied to a number of applications including content-based image retrieval~\cite{falchi2008metric}, 
contextual advertising~\cite{pandey2009nearest}, 
 object recognition~\cite{drolia2017cachier},
 and recommender systems~\cite{Sermpezis18}.
To the best of our knowledge, the literature on similarity caching has restricted itself to \textit{(i)}~a single cache (with the exception of~\cite{Sermpezis18, zhou20,garetto21arxiv}), and \textit{(ii)}~homogeneous items with identical resource requirements. %, and \textit{(iii)} a serving model without overloading. %Moreover, in similarity caching a single item can serve any number of requests per second, while in IDNs models have a bounded serving capacity.
A consequence is that in our setting similarity caching policies would only allocate models based on their accuracy, ignoring the trade-offs imposed by their resource requirements.
{Moreover, the literature on similarity caching ignores system throughput constraints,
%and assumes that the allocation of a single item is enough for serving an arbitrarily large number of requests, 
while we explicitly take into account that 
%in IDNs 
each model can only serve a bounded number of requests per second, according to its capacity.}

%!TEX root = ../paper.tex

\section{Inference System Design}
\label{sec:design}
% (illustrated in Fig. \ref{fig:system})
We consider a network of compute nodes, each capable of hosting some pre-trained ML models depending on its capabilities. Such ML models are used to serve inference requests for different classification or regression \emph{tasks}.\footnote{
    We are using the term task according to its meaning in the ML community, e.g., a task could be to detect objects in an image, to predict its future position, to recognize vocal commands.
} 
{As shown in Fig.~\ref{fig:system},} requests are generated by end devices and routed over given serving paths (e.g., from edge to cloud nodes). The goal of the system is to optimize the allocation of ML models across the network so that the aggregate serving cost is minimized. Our system model is detailed below, \rev{ts}{}{and the notation
used across the paper is summarized in Table \ref{tab:notations_definitions}.}

\begin{figure}[t!]
	\centering
	\includegraphics[clip= true, width= 0.8\columnwidth, trim= 0.0in 0.0in 0in 0.0in]{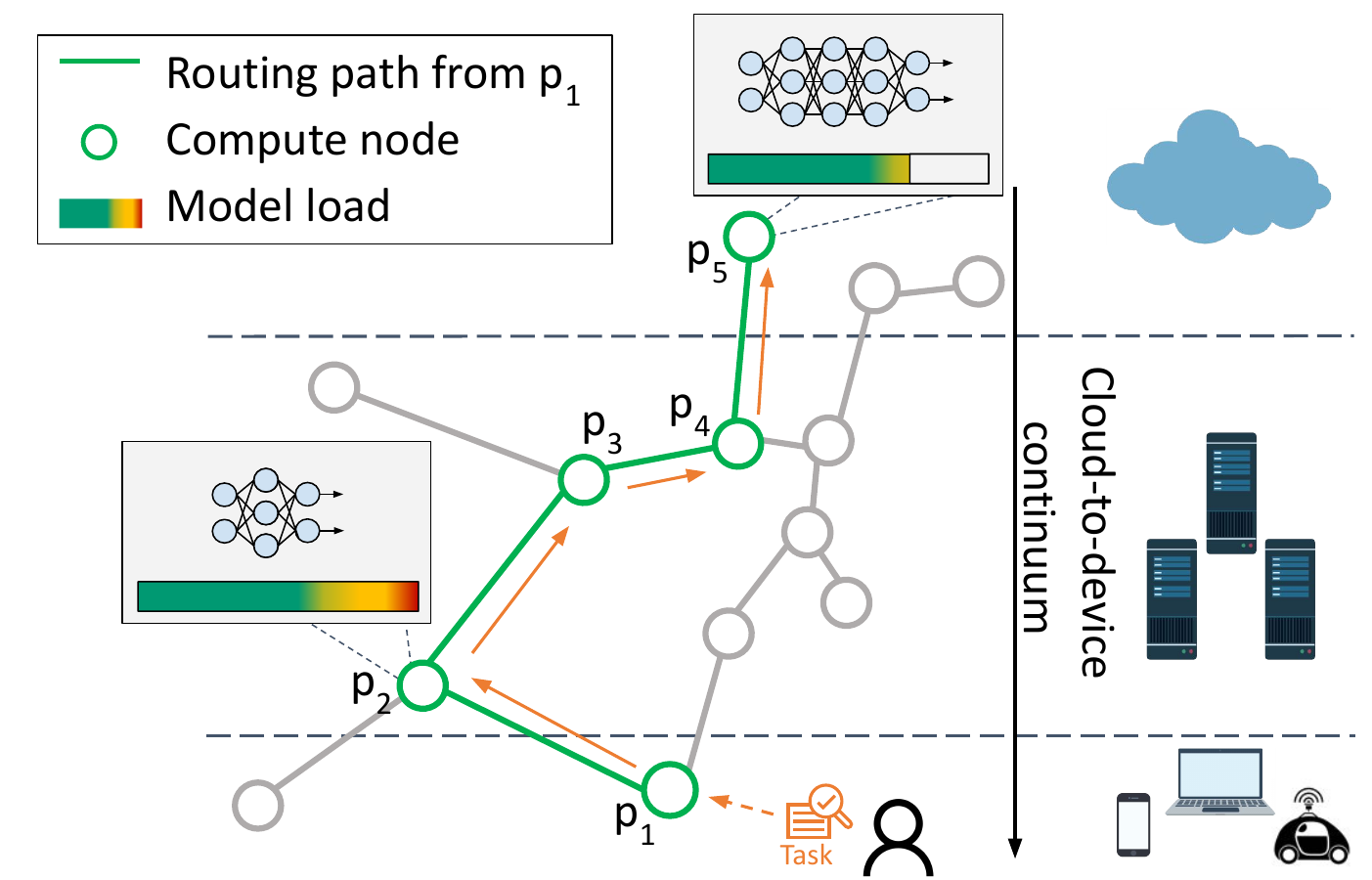}
	\caption{{System overview: a network of compute nodes serves inference requests along predefined routing paths. A repository node at the end of each path ensures that requests are satisfied even when there are no suitable models on intermediate nodes. }}
	\label{fig:system}
	% Fig source at: https://docs.google.com/drawings/d/1J4W7CDZdyTmcWLo1YsMgEdp4eZBqqZxDlGUA0LKB-r8/edit
\end{figure}

%%%%%%%%%%%%%%%%%%%%%%%%%%%%%%%%%%%%%%%%%%%%
\subsection{Compute Nodes and Models}
\label{subsec:nodes-models}
%%%%%%%%%%%%%%%%%%%%%%%%%%%%%%%%%%%%%%%%%%%%
\begin{figure}[t!]
    \centering
    \subcaptionbox{Keras pre-trained models}{\includegraphics[width = .5\linewidth]{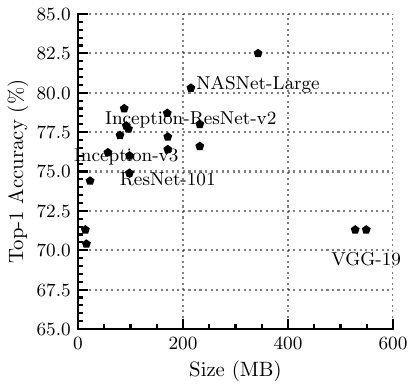}}
    \subcaptionbox{Pytorch pre-trained models}{\includegraphics[width = .485\linewidth]{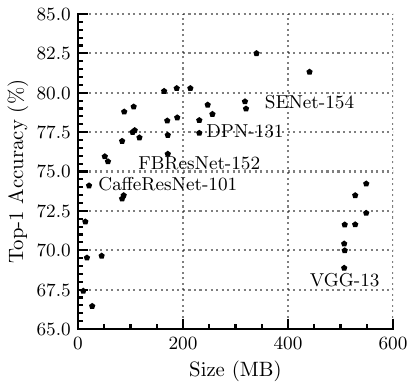}}
    \caption{Example of pre-trained model catalog for the image classification task. Data from~\cite{blalock2020state}.}
    \label{fig:model_catalog}
\end{figure}
\begin{table}[htbp]

  \caption{{Notation Summary. \label{tab:notations_definitions}  } }
  %\vspace{-0.3em}
  \begin{center}
        \begin{footnotesize}
  \begin{tabular}{|C{.09\textwidth}C{.34\textwidth}|}
\hline
&\textbf{Inference Delivery Networks} \\
    \hline
    $G(\vertices, \edges)$ & Undirected weighted graph, with nodes $\vertices$ and edges $\edges$\\
    $w_{u,v}$ & Weight of edge $(u,v) \in \edges$ \\ 
    $\taskcatalog$ / $\modelSet$ / $\modelSet_i$ &  Tasks / models catalog / models catalog for task $i$ \\
    % $\modelSet$ & Models catalog \\
    $s^v_m$ & Size of model $m \in \modelSet$ on node $v\in\vertices$\\
    $a_m$ & Prediction accuracy of model $m$ \\
    $d^v_m$ & Average inference delay of model $m$ at node $v$\\
    $\capacity^v$ & Allocation budget constraint at node $v$\\
 
    $\alloc^v_m$ & 0-1 indicator variable set to 1 if model m at node $v$ is allocated \\
    $\repo^v_m$ & 0-1 indicator constant set to 1 if model m at node $v$ is a permanent repository model \\
       
 $\repoVec$ / {$\allocVec$}
    & Minimal allocation vector (Sec.~\ref{subsec:load}) / allocation vector  (Sec.~\ref{subsec:nodes-models})
    \\
    $\allocVec^v$ / $\allocVec$ & Allocation vector at node $v$ / global allocation vector\\
    $\vec s^v$ & Vector of model sizes  at node $v$\\
    $\requestPathVec$ & Routing path $\{\requestPath_1, \dots, \requestPath_J\}$ of connected nodes $\requestPath_j \in \vertices$\\
  $\nu(\requestPathVec)$ & Repository node associated to the path $\requestPathVec$\\
    $C^{\requestPath_j}_{\requestPathVec, m}$ & Cost of serving at node $\requestPath_j$ along path $\requestPathVec$ using model $m$\\
    
    $\request$ & Request type $\request = (i,\requestPathVec)$, i.e., the requested task $i$ and the request's routing path $\requestPathVec$\\
    $\requestSet / \requestSet_i$ & Set of all the possible request types / request types for task $i$\\
    $R$ & Total number of request types\\
    $r^t_{\request}$ & Number of times $\request$ is requested during time slot $t$\\
    $\mathrm{load}^{t, v}_m(\rho')$ & Number of type-$\rho$ requests served by model $m$ at node $v$ during the $t$-slot \\
    $\loadmodel^v_{m}$ & Maximum capacity of model $m$ on node $v$\\ 
    $\load^{t,v}_{\request,m}$ & Potential available capacity of $m$ on node $v$ for request type $\request$ at time $t$
    \\
    $\smallCost^k_{\request}$ & $k$-th smallest cost for request type $\request$ along its path{ (Sec.~\ref{subsec:serving}).} 
    \\
    $\vec r_t$ / $\vec l_t$ & Request batch vector / potential available capacities vector\\ 
    $\auxload^{k}_{\request}{(\loadVec_t)}$ / $\auxalloc^{k}_{\request}{(\loadVec_t, \allocVec)}$  & Potential / effective available capacity of the model serving $\request$ with cost $\smallCost^k_{\request}${ (Sec.~\ref{subsec:serving}).}
    \\
    $\Auxalloc^k_\request(\requestBatchVec_t,\loadVec_t, \allocVec)$ & Number of requests of type $\request$ that can be  served by the $k$-th smallest cost models found along path $\requestPathVec$ under allocation $\allocVec$ at time slot $t$~\eqref{eq:sum_of_auxvars}\\
       $\kappa_{\request} (v, m)$ & Rank of model $m$ allocated to node $v$, among all the models that can potentially serve requests of type $\request$
    \\
  {$\modelsNo_{\request}$} &
    {Maximum number of models that request $\rho$ may encounter along its serving path}
    \\
    $\allocSet^v$ / $\allocSet$  & Set of valid integral allocations at node $v$ / at all nodes \\
     $\systemCost$ / $\systemGain $ / $\systemGain_T$ & The overall system cost \eqref{eq:cost-expression} / allocation gain~\eqref{eq:gain} / static allocation gain~\eqref{eq:static_gain}
    \\
    \hline
&\textbf{\AlgoName{}} \\
\hline
    $\Phi^v$ & Weighted negative entropy map $\Phi^v: \mathcal{D}^v \to \reals$ given as $\Phi^v(\allocVecFrac^v) = \sum_{m\in\modelSet}s^v_{m}\allocFrac^v_{m} \log(\allocFrac^v_{m})$\\
    $\Phi$ & Global mirror map $\Phi(\allocVecFrac) = \sum_{v\in\vertices}\Phi^v(\allocVecFrac^v) $\\
    $\allocFrac^v_m$ & Fractional allocation (allocation probability) of model $m$ is at node $v$ \\
    $\allocVecFrac^v$ / $\allocVecFrac$ & Fractional allocation vector at node $v$ / global fractional allocation vector\\
 $\allocSetFrac^v$ / $\allocSetFrac$ & Set of valid fractional allocations at node $v$ / at all nodes\\
    $\Vec{g}_t$ & Subgradient vector of $\systemGain$ over $\allocSetFrac$ at point $\vec y_t$\\
    $g^v_{t,m}$ & Component $(v,m)$ of the subgradient vector $\vec g_t$\\
   $\mathcal{P}_{\allocSetFrac^v \cap \mathcal{D}^v}^{\Phi^v} (\,\cdot\,)$ & Projection operator onto $\allocSetFrac^v \cap \mathcal{D}^v$\\ 
   $\eta$ & learning rate\\
   $B$ & Refresh period \\
   $T$ & Time horizon \\
   $t$ & Time slot $t \in [T]$ \\ 
   $A$ & Regret constant \\
   $L_{\max}$ & Upper bound on model capacities \\
   $\Delta_C$ & Upper bound on maximum serving cost difference \\
   $\psi$ & Regret discount factor equal to $1-\frac{1}{e}$ \\
    \hline
\end{tabular}
\end{footnotesize}
\end{center}
\end{table}
We represent the inference delivery network (IDN) as a weighted graph $G(\vertices, \edges)$, where $\vertices$ is the set of compute nodes, and $\edges$ represents their interconnections. Each node $v\in \vertices$ is capable of serving inference tasks that are requested from anywhere in the network (e.g., from end-users, vehicles, IoT devices). We denote by $\taskcatalog = \{1,2, \dots, |\taskcatalog|\}$ the set of tasks the system can serve (e.g., object detection, speech recognition, classification), and assume that each task $i\in \taskcatalog$ can be served with different quality levels 
%$\qualityset_i \in \{1,2, \dots, |\qualityset_i| \}$ 
(e.g., different accuracy as illustrated in Fig.~\ref{fig:model_catalog}) and different resources' requirements by a set of suitable models $\modelSet_i$.
% {=} 
% \{m^i_1,m^i_2, \dots, m^i_{|\modelSet_i|}\} 
%\subset \mathbb{N}$. 
Each task is served by a separate set of models, i.e., $\mathcal M_i {\cap} \mathcal{M}_{i'}{=}\emptyset, \forall i,i'\in\mathcal{N},i \neq i'$. Catalog $\modelSet_i$ may encompass, for instance, independently trained models or shrunk versions of a high quality model generated through distillation~\cite{Hinton2015,Ravi2018}. We denote by $\mathcal{M} {=} \cup_{i \in \taskcatalog} \mathcal M_i =\{1,2,\dots, |\modelSet|\}$ the catalog of all the available models.
% {=} \cup_{i \in \taskcatalog} \mathcal M_i=\{1,2,\dots, |\modelSet|\} \in \naturals$ the catalog of all the available models. %and by $m = (i,q) \in \modelSet$ the model serving task $i \in \taskcatalog$ with quality $q \in \qualityset_i$. 
% Note that  %\cite{Hinton2015,Han2015a,Ravi2017,Ravi2018}.
 
%\gndelete{, or even multiple setups (e.g., processing unit, machine learning framework, batch size) of the same model architecture.}
%The different models for task $i$ may have been trained independently or generated through distillation from a single high quality model \cite{DISTILLATION}.
%I THINK THE CAPACITY (MAX THROUGHPUT OF EACH MODEL) SHOULD ALREADY BE INTRODUCED HERE. %,each providing different performance. 
Finally, every model of the catalog may provide a different throughput (i.e., number of requests it can serve in a given time period), and therefore, support a different load (we formalize this in Sec.~\ref{subsec:load}).

For each compute node $v \in \vertices$, we denote by
\begin{equation}
\label{eq:decision-variable}
    x^v_{m} \in \{0,1\}, \text{ for } m \in \modelSet,
\end{equation}
the decision variable that indicates if model $m \in \modelSet$ is deployed on node $v$.\footnote{
    %Note that we can allow each node to host multiple copies of the same model to be able to satisfy a larger number of requests. This can be captured by considering two copies of the same model as two distinct models with identical performance and requirements.
    Our formulation allows each node to host multiple copies of the same model to satisfy a larger number of request. For example, two copies of the same model can be represented as two distinct models with identical performance and requirements.
%In practice, we can use multiple decision variables for the same model to support the allocation of additional replicas (horizontal scaling).
} % end-footnote
%
% allocation vector X
%Therefore, $\allocVec = [\alloc^v_{i,q}]_{v \in \vertices, i \in \taskcatalog, q \in \qualityset_i}$ denotes the global allocation decision.
Therefore, $\allocVec^v = [\alloc^v_{m}]_{m \in \modelSet}$ is the allocation vector on node $v$, and $\allocVec = [\allocVec^v]_{v \in \vertices}$ denotes the global allocation decision.
\iffalse
The global allocation decision variable is then
\begin{equation}
    \\allocVec = [\alloc^v_{m}]_{(i,v, q) \in \bigcup_{i \in \taskcatalog} \{i\} \times  \vertices \times   \qualityset^v_{i} }. 
\end{equation}
\fi

% node capacity
We assume that the allocation of ML models at each node is constrained by a single resource dimension, potentially different at each node. A node could be, for instance, 
%the number of GPUs or GPU quotas\footnote{\url{https://cloud.google.com/compute/quotas}} available. 
severely limited by the amount of available GPU memory, another by the maximum throughput in terms of instructions per second.
The limiting resource determines  the \emph{allocation budget} $\capacity^v \in \mathbb{R}_+$ at node $v \in \vertices$. We also denote with $s_{m}^v \in \mathbb{R}_+$ the size of model $m \in \modelSet$, i.e., the consumed amount of the limiting resource at node $v$.\footnote{Note that, even when the limiting resource is the same, say computing, the budget consumed by a model may be different across nodes, as they may have different hardware (e.g., GPUs, CPUs, or TPUs).} 
%{Resources can be defined separately at different nodes, i.e., some computing nodes may have different kind of resources (e.g., GPUs, CPUs, or TPUs)}, therefore the budget consumed by a model can potentially be different across nodes.} 
Therefore, budget constraints are expressed as
\begin{equation}
\label{eq:capacity-constraint}
    \sum_{m \in \modelSet} x^v_{m} s_{m}^v \leq \capacity^v, \forall v \in V. 
\end{equation}

% repository nodes
To every task $i \in \taskcatalog$, we associate a fixed set of \textit{repository nodes} that always run one model capable of serving all the requests for task $i$ (e.g., high-performance models deployed in large data centers). We call these models \textit{repository models} and they are statically allocated.
Repository models ensure requests are satisfied even when the rest of the network is not hosting any additional model. % (there is no intermediate allocation). 

We discern repository models through constants $\repo_{m}^{v} \in \{0,1\}$, each indicating if model $m$ is permanently deployed on node $v$. {We assume that values $\repo_{m}^{v}$ are given as input. } We call the vector $\repoVec = [\repo^v_m]_{(v,m) \in \vertices \times \modelSet}$  the \emph{minimal allocation}.
% \iffalse
% The global repository allocation is given as
% \begin{align}
% \label{eq:repository-capacity}
% 	\repoVec = [\repo^v_{m}]_{(i,v,q) \in \bigcup_{i \in \taskcatalog} \{i\}\times \vertices  \times \qualityset^v_i}.
% \end{align}
% \fi
% We assume that the repository allocation is feasible, i.e., 
% \begin{align}
%     \sum_{i \in \taskcatalog} \sum_{q \in \qualityset^v_{i}}  \repo^v_{i,q} s_{i,q} \leq \capacity^v.
% \end{align}
%Note that if $\repo_{i,q}^{v}=1$, then $x_{i,q}^v=1$.
% $\repoVec$ acts as a binary mask vector for the permanently stored models, and these allocations are set to 1 and they are never changed, i.e.,  for all $(i,q,v) \in \bigcup_{i \in \taskcatalog} \{i\} \times \qualityset_i \times \vertices$ if  $\alpha^v_{i,q} = 1$ we add the constraint $\alloc^v_{i,q} = 1$.
%\begin{align}
%    \alloc^v_{i,q} = 1.
%\end{align}
Note that the presence of repositories introduce the following constraints to the allocation vector:
\begin{align}
    \alloc^v_m \geq \repo^v_m, \forall v \in \vertices, \forall m \in \modelSet.
    \label{eq:repo_ineq}
\end{align}
%Note that the above constraints deactivates the overlapping allocation variables (they are always set to 1), and 
%the allocation vector $\repoVec = [\repo^v_m]_{(v,m) \in \vertices \times \modelSet}$ represents the minimal allocation. 

%Without loss of generality, we assume that a repository node for a given task stores a single model with large capacity (it represents a pool of highest quality models), i.e., for every node $v \in \vertices$ and task $i \in \taskcatalog$ we have
%\begin{align}
%    \sum_{m \in \modelSet_i} \omega^v_m \leq 1.
%\end{align}

The set of possible allocations at node $v\in\vertices$ is determined by the integrality constraints~\eqref{eq:decision-variable}, budget constraints~\eqref{eq:capacity-constraint}, and repository constraints~\eqref{eq:repo_ineq}, i.e., 
\begin{align}
    \allocSet^v \triangleq \left\{ \allocVec^v \in \{0,1\}^\modelSet: \allocVec^v \text{ satisfies Eqs.~\eqref{eq:decision-variable}--\eqref{eq:repo_ineq}}\right \}.
\end{align}
The set of possible global allocations is given as $\allocSet
\triangleq \bigtimes_{v\in\vertices} \allocSet^v  $.

% i.e., the one where models are only deployed in their respective repositories.
% {Without loss of generality, we assume that a repository node for a given task stores a single model with large capacity (it represents a pool of high quality models).} 

%CAN A NODE BE BOTH A REPOSITORY NODE AND A CACHE NODE? IF SO THEN WE CANNOT WRITE SEPARATELY EQUATION (3) AND (2). PERHAPS,  WE CAN SAY THAT FOR THE REPOSITORY $x_{i,q}^v =1$ if $v\in \hat V^i$ AND THESE VARIABLES ARE NOT MODIFIED BY THE ALGORITHM.  

%%%%%%%%%%%%%%%%%%%%%%%%%%%%%%%
\subsection{Inference Requests}
\label{subsec:requests}
%%%%%%%%%%%%%%%%%%%%%%%%%%%%%%%

% The system serves inference requests over the network $G(\vertices, \edges)$. 
% A request is determined by the requested task $i \in \taskcatalog$ and node $v \in \vertices$ where the request is generated (e.g., base station). 
We assume that every node has a predefined routing path towards a suitable repository node for each task $i \in \taskcatalog$. 
% {The routing is therefore predetermined, and our decisions only concern placement of models (i.e., variables $\alloc^v_m$)}
\newton{ Therefore, for a given request for task $i$, the routing path is a set of network nodes towards a repository node able to serve task $i$. Since we assume repository nodes are predefined, the routing path does not depend on the placement decisions (i.e., on the variables $\alloc^v_m$). Hence, a request always follows its predetermined path, but intermediate nodes that host suitable models can serve it directly instead of forwarding it all the way to the repository node. In such cases, the request would traverse just a portion of the path.}
% routing path
A routing path $\requestPathVec$ of length $|\requestPathVec| = J$ is a sequence $\{\requestPath_1, \requestPath_2, \dots, \requestPath_J\}$ of nodes $\requestPath_j \in \vertices$ such that edge $(\requestPath_j, \requestPath_{j+1}) \in \edges$ for every $j \in \{1, 2, \dots, J{-}1\}$. As in \cite{stratis}, we assume that paths are simple, i.e., they do not contain repeated nodes. A request is therefore characterized by the pair $(i, \requestPathVec)$, where $i$ is the task requested and $\requestPathVec$ is the routing path to be traversed. We call the pair $(i, \requestPathVec)$ the \textit{request type}.
% put back if there is space
%\footnote{The request could be equivalently denoted as $(i,v)$, where $v$ is the source node at which the request is generated, but the notation $(i,\requestPathVec)$ is more convenient, as it will be clear in what follows.} 
We denote by $\requestSet$ the set of all possible request types, and by $\requestSet_i$ all possible request types for tasks $i$. When a request {for task $i$} is propagated from node $\requestPath_1$ toward the associated repository node $\nu(\requestPathVec) \triangleq \requestPath_J$, any intermediate node along the path that hosts a suitable model {$m \in \modelSet_i$} can serve it. The actual serving strategy is described in Sec.~\ref{subsec:serving}.

% time slots and batch moved down

%%%%%%%%%%%%%%%%%%%%%%%
\subsection{Cost Model}
\label{subsec:cost}
%%%%%%%%%%%%%%%%%%%%%%%

When serving a request {of type} $\request{=}(i,\requestPathVec) \in \requestSet$ on node $p_j$ using model $m$, the system experiences an \textit{inference cost} that depends on the quality of the model (i.e., on inference inaccuracy) and the inference time.\footnote{\newton{Note that deployed models may need to be re-trained from time to time, and we do not consider the corresponding costs. Moreover, to streamline the presentation, we assume the inference costs to be static over time; nonetheless, one could easily extend our model to the case in which these costs are time-varying.
%Note that models may be periodically re-trained, but we do not  consider the corresponding costs.
%but we do not consider the cost of periodically re-training the different models in the catalog is not captured in the design of our system.
}}
% this wrong, the serving cost does not depend from the network. I renamed it in inference cost
%and on the latency experienced by the user (delay between when the request is issued and when the inference is retrieved) 
%as well as a
Additionally, the system experiences a \textit{network cost}, %which corresponds to the resources consumed along the path. %The service cost depends on the quality of the model, its inference time, the characteristics of node $p_k$. 
%The network cost depends 
due to using the path between $p_1$ and $p_j$. Similarly to previous work~\cite{Mengshoel2017}, we can write the total cost of serving a request as 
\begin{equation}
    C^{\requestPath_j}_{\requestPathVec, m} =    f((p_1, \dots, p_j),m).
    %C^{\requestPath_k}_{\request, m} = f(\request,p_k,m).
\end{equation}
%Our theoretical results hold under this very general cost model, as far as $f(\requestPathVec,m)$ is non-decreasing in $\requestPath$, that is $f(\requestPathVec,m) \le f(\requestPathVec',m)$ if $\requestPathVec$ is a subpath of $\requestPathVec'$. In what follows---for the sake of concreteness---we refer to the following simpler model:
While our theoretical results hold under this very general cost model, in what follows---for the sake of concreteness---we refer to the following simpler model:
%\todoi{It could be changed to an explanation how $f(\cdot)$ looks like in real applications, otherwise  we don't need to have the increasing assumption on $f(\cdot)$.}
%given by the combination of the processing delay introduced by $\model_{i,q}$ and the prediction inaccuracy. Additionally, the system also experiences a \textit{network cost} due to the round-trip latency between $\requestPath_1$ and $v$.
%To capture the network cost, we associate a weight $\weight_{i,j} \in \reals_+$ with each edge $(i,j) \in \edges$, representing the cost of forwarding and serving a request along this edge. The cost of serving request $\request$ at the $k$-th node of the path using model $\model_{i,q}$ is therefore given by:
\begin{equation}
    \label{eq:serving_cost2}    C^{\requestPath_j}_{\requestPathVec, m} =\sum_{j' = 1}^{j-1} \weight_{\requestPath_{j'},\requestPath_{j'+1}} + d^{\requestPath_j}_m + \alpha(1{-}a_m),
\end{equation}
%\todoi{We may need for $m \in \modelSet \setminus \modelSet_i$ to have $b^{p_k}_m = \infty$, otherwise it is confusing.}
where $a_m$ and $d^{\requestPath_j}_m$  are respectively the prediction accuracy (in a scale from 0 to 1) and the {average} inference delay of model $m$ on node $\requestPath_j$. Indeed, the same model may provide different inference delays, depending on the hardware capabilities of the node on which it is deployed, e.g., the type of GPU or TPU~\cite{Brooks2019}. Parameter $\weight_{v,v'} \in \reals_+$ is the (round-trip) latency of edge $(v,v') \in \edges$. Parameter $\alpha$ weights the importance of accuracy w.r.t. the overall latency and can be set depending on the application.
%where the sum of the edge weights represents the round-trip cost, while $a^{\requestPath_k}_{m} \in \reals_+$ models the service cost. We compute $a^{\requestPath_k}_{m}$ as a linear combination between inference time and prediction inaccuracy, using a coefficient $\alpha$ to weight the contribution of the prediction inaccuracy. In general, the service cost can be computed using any arbitrary function. 
%This service cost may also depend on the compute node $p_k$, because of the heterogeneity of the computing capabilities.
Note that seeking cost minimization along a serving path usually leads to a trade-off: while the network cost always increases with $j$, in a typical network the service cost $d^{\requestPath_j}_m + \alpha(1{-}a_m)$ tends to decrease, as farther nodes (e.g., data centers) are better equipped and can run more accurate models { (Fig.~\ref{fig:model_catalog})}. 
\newton{We remark that models' sizes determine which allocations are feasible, but do not affect directly the service costs. As a consequence, even  if we assume different limiting resources on different nodes (e.g., GPU, memory), we do not need to convert amounts of different resources to a common unit (e.g., a monetary cost).}
%\newton{We remark that even if we assume different limiting resources on different nodes (e.g., GPU, memory), this does not play any role in the cost function, since it does not depend on the model size. Hence the system does not need to translate different resource constraints to a common reference frame to minimize the overall cost.}
%%%%%%%%%%%%%%%%%%%%%%%
\subsection{Request Load and Serving Capacity}
\label{subsec:load}
%%%%%%%%%%%%%%%%%%%%%%%

Let us assume that time is split in slots of equal duration. We consider a time horizon equal to $T$ slots. 
%We denote by $\requestBatch^t_\request \in \naturals$ the number of times request $\request \in \requestSet$ is requested during time slot $t$. The whole batch of requests at time slot $t$ is characterized by the request batch vector
At the beginning of a slot $t$, the system receives a batch of requests $\requestBatchVec_t = [r^t_{\request}]_{\request \in \requestSet}$, where $\requestBatch^t_\request \in \naturals \cup \{0\}$ denotes the number of requests of type $\request \in \requestSet$.
\iffalse
\begin{align}
    \requestBatchVec_t = [r^t_{\request}]_{\request \in \requestSet}.
\end{align}
\fi

%\gn{We do not assumptions on the distribution }
Model $m \in \modelSet$ has maximum capacity $\loadmodel^v_{m} \in \naturals$ when deployed at node $v \in \vertices$, i.e., it can serve at most $\loadmodel^v_{m}$ requests during one time slot $t\in [\timehorizon]$, in absence of other requests for other models. We do not make specific assumptions on the time required to serve a request.

We denote by $\load^{t,v}_{\request,m} \in \naturals \cup \{0\}$ the  \emph{potential available capacity}, defined as the maximum number of type-$\request$ requests node $v$ can serve at time $t$ through model $m$, under the current request load~$\requestBatchVec_t$ and allocation vector~$\allocVec_t^v$. Formally, let $\mathrm{load}^{t, v}_m(\rho)$  denote the number of type-$\rho$ requests served by model $m$ at node $v$ during the $t$-slot, then
\begin{align}
l^{t,v}_{\rho, m}\triangleq \min\left\{L^v_m - \sum_{\rho' \in \requestSet \setminus \{\rho\} }\mathrm{load}^{t, v}_m(\rho'),\, r^{t}_{\request}\right\}. 
\label{eq:capacity_min_def}
\end{align}
% $\requestBatch^{t,v}_\rho$ denote the number of type-$\rho$ requests routed through node $v$ and  $\mathrm{load}^{t, v}_m(\rho)$ 

% The potential available capacity is obviously limited by the maximum capacity of the model and by the total number of type-$\rho$ requests routed through node $v$ (denoted by $\requestBatch^{t,v}_\rho$), i.e., 
% % \begin{align}
%     \load^{t,v}_{\request,m} & \le \min\{\loadmodel^v_{m},\requestBatch^{t,v}_\rho\},
%     \label{eq:load_lim}
% \end{align}
The potential available capacity depends on the request arrival order and the scheduling discipline at node $v$.
%, and the distribution of requests for other tasks across the different nodes. 
For instance, suppose that in time slot $t$, requests of two types $\request=(i,\requestPathVec)$ and $\request' = (i, \requestPathVec')$ arrive at node $v$. The arrival order and the node scheduling discipline may determine that many requests of type $\request'$ be served, which would leave a small $\load^{t,v}_{\request,m}$ available for requests of type $\request$. Or the opposite may happen.
% eq:capacity_min_def
It is useful to define the potential available capacity also for models that are not currently deployed at the node, as $\load^{t,v}_{\request,m} \triangleq \min\{\loadmodel^v_{m},\requestBatch^{t}_\rho\}$. The \emph{effective available capacity} is then equal to $\load^{t,v}_{\request,m} x_m^v$.

%During the slot $t$, if model $m$ is deployed at node $v$ at most  requests of type $\rho$ can be served. We call $\load^{v,t}_{\rho,m}$ the \emph{available capacity}.
%{Note that $\load^{t,v}_{\request,m}$ depends on the request arrival order, the scheduling discipline at node $v$, and the distribution of requests for other tasks across the different nodes. For instance, suppose that in time slot $t$, requests of two types $\request$ and $\request'$ arrive at node $v$. The arrival order and the node scheduling discipline may determine that many requests of type $\request'$ be served, which would leave a small $\load^{t,v}_{\request,m}$ available for requests of type $\request$. Or the opposite may happen. }
%{In any case,} $\load^{t,v}_{\request,m} \le \min\{\loadmodel^v_{m}, r^t_\rho\}$ {because the model cannot serve beyond what is requested at time $t$.} 

Our analysis in Sec.~\ref{sec:guarantees} considers a ``pessimistic'' scenario where 
an adversary selects both requests and available capacities for all models but the repository ones.
%both requests and available capacities are selected by an adversary. 
This approach relieves us from the need to model system detailed operations, while our proposed algorithm (Sec.~\ref{sec:algorithm}) benefits from strong guarantees in the adversarial setting. In what follows, we can then consider that values $\load^{t,v}_{\request,m}$ are exogeneously determined.
The vector of potential available capacities at time $t \in [\timehorizon]$ is denoted by
\begin{align}
   \loadVec_t = [\load^{t,v}_{\request,m}]_{(\request, m, v ) \in \bigcup_{i \in \taskcatalog} \requestSet_i \times \modelSet_i \times \vertices}.
\end{align}

As we mentioned in Sec.~\ref{subsec:nodes-models},  any request of type $\rho = (i, \requestPathVec) \in \requestSet$ can always be served by the associated repository model at node $\nu(\requestPathVec)$. 
% We assume that all repository models have capacity $L_{\max}$. 
This requirement can be expressed as follows:
\begin{align}
    % \load^{t,v}_{\request,m} = \min\{r^{t,v}_\rho, L^v_m\} = r^{t,v}_\rho 
    % r^{t, v}_{\rho}\leq \load^{t,v}_{\request,m} \leq r^t_\rho
    % \qquad\mathrm{if}\qquad \repo^{v}_{m}=1
    \sum_{\rho \in \requestSet_i} r^t_\rho \leq \sum_{m \in \modelSet_i} \omega^{\nu(\requestPathVec)}_m L^{\nu(\requestPathVec)}_{m} \gc{, \forall i \in \taskcatalog}.
    \label{eq:repo_feasibility_constraint}
\end{align}
%we assume the repository model can always process a number of requests equal to their maximum  capacity (then $L^v_m \gg r^t_\rho$  and  $\load^{t,v}_{\request,m} = \min\{r^t_\rho, L^v_m\} = r^t_\rho$ if $\repo^{v}_{m}=1$) and require that for every $t \in [T]$ to have
% \begin{align}
%     %\sum_{\requestPathVec \colon \request\in \requestSet}\requestBatch_{t,\request} \leq \sum_{m \in \modelSet_i} \repo^{v}_{m} {\tilde{\loadmodel}^{v}_{m}(\loadVec_t)}, \forall v \in \bar{\vertices}_i.
%   \sum_{m \in \modelSet_i} \repo^{v}_{m} \load^{t,v}_{\request, m}  = \requestBatch^{t}_{\request} , \text{ for all } v \in \requestPathVec \colon \sum_{m \in \modelSet_i}\repo^v_m = 1.%\bar{\vertices}_i.
%   \label{eq:repo_feasibility_constraint}
% \end{align}
 
 Thus, at any time $t \in [T]$ the adversary can select a request batch $\vec r_t$ and potential available capacity $\vec l_t$ from the set
 
 \begin{align}\nonumber
 \mathcal{A} \triangleq \bigg\{&
 (\vec r, \vec l) \in (\naturals\cup\{0\})^\requestSet \times 
 (\naturals \cup \{0\})^{\bigcup_{i \in \taskcatalog} \requestSet_i \times \modelSet_i \times \vertices} :\\
 &\sum_{\rho \in \requestSet_i} r^t_\rho \leq \sum_{m \in \modelSet_i} \omega^{\nu(\requestPathVec)}_m L^{\nu(\requestPathVec)}_{m}, l^v_{\rho, m} \leq \min\{L^v_m, r_\rho\},\nonumber \\
 &\forall i \in \mathcal{N}, v \in \vertices,m \in \modelSet, \rho \in \requestSet
 \bigg\}.    
 \end{align}
%  {eq:load_lim}
 Note the constraint on potential available capacities is looser than the definition in~\eqref{eq:capacity_min_def} corresponding to a more powerful adversary. 
 
%  \begin{align}
%      \requestBatchSet \triangleq \left\{ \vec r \in (\naturals\cup\{0\})^\requestSet: \vec r \text{ satisfies Eqs.~\eqref{eq:load_lim},~\eqref{eq:repo_feasibility_constraint}}\right\}
% \end{align} 
%  and 
%  \begin{align}
%      \loadSet \triangleq  \left\{ \vec l \in (\naturals \cup \{0\})^{\bigcup_{i \in \taskcatalog} \requestSet_i \times \modelSet_i \times \vertices}: \vec l \text{ satisfies Eqs.~\eqref{eq:load_lim},~\eqref{eq:repo_feasibility_constraint}} \right\}.
%  \end{align}
 
%  hile $\requestBatchSet$ and $\loadSet$ are respectively the set of all possible request batches and potential available capacities, i.e., $\requestBatchVec_t \in \requestBatchSet$ and $\loadVec_t \in \loadSet, \forall t \in [\timehorizon]$.

%\todoi{aa: An alternative way to say the same thing could be that we assume that $L_{i,q}^v=\infty$ for any node $v$ being repository for the model $M_{i,q}$. Did I understand correctly?}
%with $v$ being  node $v$ (i.e., $v \in \vertices \colon \sum_{q \in \qualityset_i} {\repo^v_{i,q}} \geq 1$.

%\rev{ts}n {Observe that in this paper we only take into account Quality of Service for machine learning application users, which we represent as a combination of delay and accuracy. We do not consider instead the cost for the network operator, which would tend to minimize bandwidth and computational resource usage. We think such a problem is complementary to ours and deserves investigation in our future work.}{Perhaps we don't need this paragraph anymore, since we perform experiments with the update cost.}

%%%%%%%%%%%%%%%%%%%%%%%%%%
\subsection{Serving Model}
\label{subsec:serving}
%%%%%%%%%%%%%%%%%%%%%%%%%%
\begin{figure}[t!]
    \centering
    \includegraphics[width = .7\linewidth]{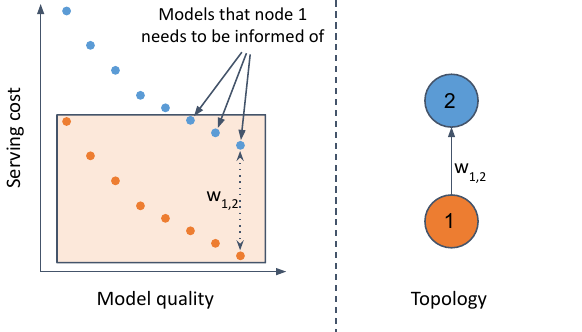}
    \caption{Necessity of partial synchronization in IDN among close-by computing nodes under the cost model in Eq.~\eqref{eq:serving_cost2}.}
    \label{fig:synch}
\end{figure}
Given request $\request{=}(i,\requestPathVec){\in} \requestSet$, let  $\modelsNo_{\request} = |\requestPathVec| |\modelSet_i|$ denote the maximum number of models that may encounter along its serving path~$\requestPathVec$. We order the corresponding costs $\{C^{\requestPath_{j}}_{\requestPathVec, m} , \forall m \in \modelSet_i, \forall \requestPath_{j} \in \requestPathVec \}$ in increasing order and we denote by $\kappa_{\request} (v, m)$ the rank of model $m\in\modelSet_i$ allocated at node $v$ within the order defined above.\footnote{
    Note that we do not consider only  the models deployed in the network, but all the possible node-model pairs.
    %considering all nodes in $\requestPathVec$ and all possible models from $\modelSet_i$.
}
%, i.e., model $m$ at node $v\in\requestPathVec$ has the $\kappa_{\request} (v, m)$-th smallest cost. 
If $v\notin \requestPathVec$ we have $\kappa_{\request} (v, m) = \infty$.

If $\kappa_{\request} (v, m)=k$, then model $m$ at node $v$ has the \mbox{$k$-th} smallest cost to serve request $\rho$. We denote the model service cost, its potential available capacity, and its effective capacity as  $\smallCost^k_{\request}$, $\auxload^{k}_{  \request}(\loadVec_t)$, and $\auxalloc^{k}_{  \request}(\loadVec_t, \allocVec)$, respectively:
% , where %$\modelsNo_{\request} = \sum^{|\requestPathVec|}_{k=1} |\qualityset^{\requestPath_k}_i|$.
% $\modelsNo_{\request} = |\requestPathVec| |\modelSet_i|$.
% For every request $\request \in \requestSet$ and $k \in [\modelsNo_{\request}]$, we denote by $\smallCost^k_{\request}$ the $k$-th smallest cost along the request path. 
%%% Modified on 12/03
%Similarly, we denote by $\auxalloc^k_{ \request}(\loadVec_t, \allocVec)$$\auxalloc^k_{ \request}(\loadVec_t, \allocVec)$ the effective available capacity at time $t$ of the model incurring the \mbox{$k$-th} cost. More precisely, if the $k$-th smallest cost is obtained by model $m$ at node $v \in \requestPathVec$ (i.e., $\smallCost^k_{\request} =    C^{v}_{ \requestPathVec, m}$), then 
%%% New version on 12/03
%Similarly, we define two map functions for the potential and effective available capacity as follows. Assume the $k$-th smallest cost is obtained by model $m$ at node $v \in \requestPathVec$ (i.e., $\smallCost^k_{\request} =    C^{v}_{ \requestPathVec, m}$), then 
\begin{align}
% \label{eq:gamma}
    \smallCost^k_{\request} =  C^{v}_{ \requestPathVec, m},\quad%\\
% \label{eq:lambda}
    \auxload^{k}_{  \request}(\loadVec_t) =  \load^{t,v}_{\request,m},\quad%\\
% \label{eq:z}
    \auxalloc^{k}_{  \request}(\loadVec_t, \allocVec) = \alloc^{v}_{m} \load^{t,v}_{\request,m}.
    \label{eq:several-definitions}
\end{align}

% \begin{align}
% \label{eq:z}
% \auxalloc^{k}_{  \request}(\loadVec_t, \allocVec) = \alloc^{\requestPath_{k'}}_{m} \tilde{\loadmodel}^{\requestPath_{k'}}_{\request,m}(\loadVec_t).
% \end{align}

We assume the IDN serves requests as follows.
%At every time slot $t$, 
%for any $(i, \requestPathVec) \in \supp {\requestBatchVec_t}$ each request is forwarded along $\requestPathVec$ 
%any request $\request \in \supp {\requestBatchVec_t}$, i.e., that has been requested at least once, 
Each request is forwarded along its serving path and served when it encounters a model with the smallest serving cost among those that are not yet saturated, i.e., that may still serve requests. 
%(i.e., $\min\{k \colon \auxalloc^{k}_{  \request} > 0\}$).

Since models do not necessarily provide increasing costs along the path, this serving strategy requires that a node %along the path of a request for task $i$ and storing a model $m \;{\in}\; \modelSet_i$
{that runs a model $m \;{\in}\; \modelSet_i$ and receives a request for task $i$,}
knows whether there are better alternatives for serving task $i$ upstream or not. In the first case, it will forward the request along the path, otherwise it will serve it locally.
We argue that, in a real system, this partial knowledge can be achieved with a limited number of control messages.
%does not correspond to each node having global knowledge about the system is much more limited can be achieved without the necessity of  synchronization by exchanging a limited number of control messages. 
In fact, if node $v=\requestPath_h$ hosts the model with the $k$-th cost for request  $(i,\requestPathVec)$, it only needs information about those models that \textit{(i)} are located upstream on the serving path (i.e., on nodes $\requestPath_l \in \requestPathVec$ with $l > h$), and \textit{(ii)} provide a cost smaller than $\smallCost^k_{\request}$. Since the cost increases with the network latency (as illustrated in Fig.~\ref{fig:synch}), the number of models satisfying these criteria is
%limited and practically smaller than $|\requestPathVec|-h$
small in practice.\footnote{
In realistic settings (Sec.~\ref{sec:validation}), we experienced that each deployed model has at most 6 
%\textcolor{red}{ts: depends on alpha}\gabriele{we can report the value for alpha=1, as it is the one that deploys models at all the tiers.}\gabriele{Actually, for bigger alpha this number even decreases.}
better alternatives on upstream nodes (worst case with $\alpha{=}1$).
} % - end footnote
A node needs to propagate downstream a control message with the information about the 
requests it can serve and the corresponding costs. Nodes forwarding the control message progressively remove the information about the tasks they can serve with a smaller cost, until the control message payload is empty and the message can be dropped. %These messages should be sent whenever the availability to serve additional requests changes.
Every node $v\in\vertices$ generates this control message whenever the available capacity of any of the local models in $v$ changes.
% \todoi{aa: Based on the last sentence, a node $v$ generates control messages not at every request, but only when some models in $v$ change. What happens in the following situation? (i) there is a ``mediocre model'' in node $v_\text{up}$ and a ``good'' model in node $v_\text{down}$. Therefore, $v_\text{down}$ ``censors'' the message from node $v_\text{up}$. (ii)~Then, in node $v_\text{down}$, the good model is replaced by a ``very bad model''. (iii)~In this case, node $v_\text{up}$ does not generate any message, since its local models have not changed. But this is unfortunate: indeed, the mediocre model in node $v_\text{up}$ would be worth using, but node $v_\text{up}$ has no way to manifest itself. How do we solve this problem? \gabriele{$v_\text{down}$ received the information about the existence of a mediocre model on $v_\text{up}$, therefore once it receives requests for that task it will know that there is a better model upstream compared to the local one and it will forward the requests.}}
%\gcdelete{ beyond given thresholds}. 
% the path withabout which 
% it is enough that every node propagates a threshold-based message downstream to notify any sub-optimal models whether it can currently process additional requests or not.

According to the presented serving strategy, the requests load is split among the currently available models giving priority to those that provide the smallest serving costs up to their saturation. In particular, model $m$ with the $k$-th smallest cost will serve some requests of type $\request$ only if the less costly models have not been able to satisfy all of such requests (i.e., if $\sum^{k-1}_{k'=1} \auxalloc^{k'}_{\request}(\loadVec_t, \allocVec) < \requestBatch^{t}_{\request}$). If this is the case, model $m$ will serve with cost $\smallCost_\request^k$ at most  $\auxalloc^k_{\request}(\loadVec_t,\allocVec)$ requests (its effective available capacity) out of the $ \requestBatch^{t}_{\request} - \sum^{k-1}_{k'=1} \auxalloc^{k'}_{\request}(\loadVec_t, \allocVec)$ requests still to be satisfied. The aggregate cost incurred by the system at time slot $t$ is then given by
\begin{align}
\nonumber
	   \systemCost(\requestBatchVec_t, \loadVec_t, \allocVec)\;
	   {=}\hspace{-0.1em}
	   \sum_{\request \in \requestSet} \sum_{k=1}^{\modelsNo_{\request}}
	   \smallCost^k_{\request}
	   &\cdot
	   %\vGroup
	   %{
	   \min\biggl\{\requestBatch^t_{\request}\;\hspace{-0.2em}{-}\hspace{-0.2em}\sum^{k-1}_{k'=1}\hspace{-0.15em}\auxalloc^{k'}_{\request}(\loadVec_t, \allocVec),\hspace{0.1em} 
	   \auxalloc^{k}_{\request}(\loadVec_t,\allocVec)
	   \biggr\} \\
	   &\cdot
	   \mathds{1}_{\left\{\sum^{k-1}_{k'=1} \auxalloc^{k'}_{\request}(\loadVec_t, \allocVec) < \requestBatch^{t}_{\request}\right\}}
	   %}
	   %{\text{Number of requests served by the $k$-th best model}}
	   .
	    \label{eq:cost-expression}
\end{align}
% \todoi{I added the text ``Number of requests served \dots'' in the formula. I also added the $\cdot$, which makes the formula more readable.}
{Note that we introduce the $\min\{\,\cdot\,,\,\cdot\, \}$ operator, since the number of requests served by the $k$-th best model cannot exceed its effective capacity $\auxalloc^{k}_{\request}(\loadVec_t, \allocVec)$. We add the indicator function $\mathds{1}_{\{\,\cdot\,\}}$ to indicate that the $k$-th best model does not serve any requests, in case better models (ranked from $1$ to $k-1$) are able to satisfy all of them.}

% The intuition behind this formula is that, along each serving path, the load of requests is split among the currently available models giving priority to those that provide the smallest costs up to their saturation.
% If the request load exceeds the current allocation on intermediate nodes, the remaining portion is eventually served on the last node of the correspondent path (i.e., the repository node).

%%%%%%%%%%%%%%%%%%%%%%%%%%%%%%%%%%%%%%%%%
\subsection{Allocation Gain and Static Optimal Allocations}
\label{subsec:gain}
%%%%%%%%%%%%%%%%%%%%%%%%%%%%%%%%%%%%%%%%%

We are interested in model allocations {$\allocVec$} that minimize the aggregate cost~\eqref{eq:cost-expression}, or, equivalently, that maximize the \emph{allocation gain} defined as 
\begin{equation}
     \systemGain(\requestBatchVec_t,\loadVec_t, \allocVec) = \systemCost(\requestBatchVec_t, \loadVec_t, \repoVec) -   \systemCost(\requestBatchVec_t ,\loadVec_t, \allocVec).
     \label{eq:gain}
\end{equation}

The first term $\systemCost(\requestBatchVec_t, \loadVec_t, \repoVec)$ on the right hand side is the service cost when only repository models are present in the network. 
%is a constant that does not 
%When there is no allocation on intermediate nodes, the cost is given by $\systemCost(\requestBatchVec_t, \loadVec_t, \repoVec)$.
%
Since intermediate nodes can help serving the requests at a reduced cost, 
%$\systemCost_0(\requestBatchVec_t)$ 
$\systemCost(\requestBatchVec_t, \loadVec_t, \repoVec)$ is an upper bound on the aggregate serving cost, and the allocation gain captures the cost reduction achieved by model allocation $\allocVec$.

The static model allocation problem can then be formulated as finding the model allocation $\allocVec^*$ that maximizes the time-averaged allocation gain  over the time horizon $T$, i.e., 
\begin{align}
   \allocVec_*  = \underset{\vec x \in \allocSet}{\argmax}\left( \systemGain_T (\allocVec) \triangleq\frac{1}{T} \sum^T_{t=1} \systemGain(\requestBatchVec_t, \loadVec_t, \allocVec)\right).
   \label{eq:static_gain}
\end{align}

% can then be formulated as finding the model allocation $\allocVec^*$ that maximizes the sum of the allocation gains over the time horizon $[1,T]$, i.e.,~$\allocVec^* = \argmax \sum_t \systemGain(\requestBatchVec_t,\loadVec_t, \allocVec)$ subject to the integrality constraints~\eqref{eq:decision-variable}, budget constraints~\eqref{eq:capacity-constraint}, and repository constraints~\eqref{eq:repo_ineq}. %\footnote{Note that we are still considering the load $\loadVec_t$ to be hexogeneously determined.}
This is a submodular maximization problem under multiple knapsack constraints.
In our context, this intuitively means that the problem is characterized by a diminishing return property: adding a model $m$ to any node $v$ gives us a marginal gain that depends on the current allocation: the more the models already deployed in the current allocation, the less the marginal gain we get by the new $m$. We 
% will
prove submodularity in Lemma~\ref{lemma:submodularT} in Appendix~\ref{appendix:submodularity}.  \newton{In Appendix~\ref{appendix:np_hard} Theorem~\ref{theorem:np_hard}, we prove that this problem is NP-hard even under cardinality constraints (i.e., the models have equal size) and a two nodes scenario. We demonstrate the
hardness of the problem by a reduction of the similarity caching problem~\cite{ascent, garetto20infocom}, which is NP-hard; a result that follows from a reduction of the dominating set problem.}
%, which is NP-hard.\footnote{{In Appendix~\ref{appendix:submodularity} Lemma~\ref{lemma:submodularT}, we prove that $G_T (\allocVec)$ is submodular.}}
It is known that submodular maximization problems
% are NP-hard and
cannot be approximated with a ratio better than $(1 - 1/e)$ even under simpler cardinality constraints~\cite{nemhauser1978best}. Under the multi-knapsack constraint, it is possible to solve the offline problem achieving a $(1 - 1/e - \epsilon)$-approximation through a recent algorithm proposed in \cite{fairstein20201}.

\gc{Let us consider a model allocation $\allocVec$. Within time slot $t$, the $k$ smallest cost models along a path $\requestPathVec$ that are suitable for request type $\request = (i, \requestPathVec)$ can serve up to $\Auxalloc^k_\request(\requestBatchVec_t,\loadVec_t, \allocVec)$ requests, where $\Auxalloc^k_\request(\requestBatchVec_t,\loadVec_t, \allocVec)$ is defined as}
%\rev{ts}{}{The number of requests of type $\request = (i, \requestPathVec)$ that can be potentially served by the $k$ smallest cost models found along path $\requestPathVec$ under allocation $\allocVec$ at time slot $t$ is defined as}
\begin{align}
\Auxalloc^k_\request(\requestBatchVec_t,\loadVec_t, \allocVec)\triangleq
\min\left\{
    \requestBatch^t_\request,  \sum^{k}_{k'=1} {\auxalloc}^{k'}_{\request}(\loadVec_t, \allocVec)
\right\}.
\label{eq:sum_of_auxvars}\end{align}
The $\min\{\, \cdot\,, \,\cdot\,\}$ operator denotes that we can never serve more than the number of requests $\requestBatch^t_\request$ issued by users. %\rev{ts}{Observe that if we increase the number $k$ of considered models, also $\Auxalloc^k_\request(\requestBatchVec_t,\loadVec_t, \allocVec)$ increases. Moreover, since the models activated in $\allocVec$ always include the models activated in $\repoVec$, we have that $\Auxalloc^k_\request(\requestBatchVec_t,\loadVec_t, \allocVec)\ge \Auxalloc^k_\request(\requestBatchVec_t,\loadVec_t, \repoVec)$. Since $\repoVec$ is an input parameter and does not depend on our decisions, we emphasize that $\Auxalloc^k_\request(\requestBatchVec_t,\loadVec_t, \repoVec) \in \{0, \requestBatch^t_\request\}$ is a constant. }{}
\gc{Observe that, being the minimal allocation $\repoVec$ an input parameter not dependent on our decisions, $\Auxalloc^k_\request(\requestBatchVec_t,\loadVec_t, \repoVec)$ is a constant. Additionally, since the models allocated in $\allocVec$ always include those allocated in $\repoVec$, we have $\Auxalloc^k_\request(\requestBatchVec_t,\loadVec_t, \allocVec)\ge \Auxalloc^k_\request(\requestBatchVec_t,\loadVec_t, \repoVec)$.}

\gc{Using \eqref{eq:sum_of_auxvars}, we provide  the following alternative formulation of the allocation gain.}
%We conclude by providing a useful alternative formulation of the allocation gain
%{The allocation gain has therefore the following alternative formulation}
%\ifnum\isextended=1
%(the proof is in Appendix \ref{appendix:gain_expression}):
%\else
%(the proof is in \cite{sisalem21techrep}):
%\fi\newline
\begin{lemma}
\label{lemma:gain}
The allocation gain \eqref{eq:gain} has the following equivalent expression:

% \begin{align}
% \rev{aa}{
%     \systemGain(\requestBatchVec_t,\loadVec_t, \allocVec) =
%     \sum_{\request \in \requestSet}
%     \sum_{k=1}^{\modelsNo_{\request}-1} 
%     \left(\smallCost^{k+1}_{\request} - \smallCost^{k}_{\request}\right)
%     \cdot 
%     \min\left\{
%         \requestBatch^t_{\request} {-} 
%         \beta^{k}_{ \request} (\requestBatchVec_t,\loadVec_t, \repoVec),
%         \sum^{k}_{k'=1}\auxalloc^{k'}_{\request}(\loadVec_t, \allocVec) {-}
%         \beta^{k}_{ \request} (\requestBatchVec_t,\loadVec_t, \repoVec)
%     \right\},
%     \label{eq:gain-compact}
% }{}
% \end{align}

\begin{align}
    \systemGain(\requestBatchVec_t,\loadVec_t, \allocVec) =
    &\sum_{\request \in \requestSet}
    \sum_{k=1}^{\modelsNo_{\request}-1} 
    \vGroup{
    \left(\smallCost^{k+1}_{\request} - \smallCost^{k}_{\request}\right)
    }{\text{cost saving}}\nonumber\\
    &
    \vGroup{
    \left(
    \Auxalloc^k_\request(\requestBatchVec_t,\loadVec_t, \allocVec)
    - \Auxalloc^k_\request(\requestBatchVec_t,\loadVec_t, \repoVec)
    \right).
    }{\text{additional requests}}
    \label{eq:gain-compact}
\end{align}

% \rev{aa}{where $\beta^{k}_{ \request} (\requestBatchVec_t,\loadVec_t, \repoVec) := 
% \min\left\{
%     \requestBatch_{t,\request},  \sum^{k}_{k'=1} {\auxalloc}^{k'}_{\request}(\loadVec_t, \repoVec)
% \right\}$.}{}
%A proof can be found in \cite{??}.
\end{lemma}

% The proof of Lemma~\ref{lemma:gain} is in
% \ifnum\isextended=1
% Appendix \ref{appendix:gain_expression}.
% \else
% \cite{sisalem21techrep}.
% \fi
We prove this lemma in Appendix~\ref{appendix:gain_expression}.
% A proof of Lemma~\ref{lemma:gain} can be found in~\cite{sisalem21techrep}.
% Our objective is to maximize the gain by means of a fully distributed algorithm, which can operate without requiring synchronization between compute nodes nor prior knowledge on the requests pattern.
{This result tells us that the gain of a certain allocation~$\allocVec$ can be expressed as a sum of several components. In particular, for each request type $\request$, the $k$-th smallest cost model along the path contributes to the gain with a component \textit{(i)}~proportional to its cost saving $\smallCost^{k+1}_{\request} - \smallCost^{k}_{\request}$ with respect to the ($k+1)$-th smallest cost model and \textit{(ii)}~proportional to the amount of additional requests that the $k$-th smallest cost models in allocation $\allocVec$ can serve with respect to the minimal allocation $\repoVec$.} 

%!TEX root = ../paper.tex
\section{\AlgoName{} Algorithm}
\label{sec:algorithm}
%(see Section~\ref{sec:guarantees} for a reduction from the budgeted maximum coverage problem \cite{budgetmaxcoverage}). 
%\todoi{add full name of the algo?}
In this section, we propose \AlgoName{}, an online algorithm that can operate in a distributed fashion without requiring global knowledge of the allocation state and requests arrival.
% In Sec.~\ref{sec:guarantees}, we show that \AlgoName{} converges to an allocation within a $(1 - 1/e)$-approximation from the optimum.
In Sec.~\ref{sec:guarantees}, we show that \AlgoName{} generates dynamically allocations experiencing average costs that converge to a $(1 - 1/e - \epsilon)$-approximation of the optimum, which matches the best approximation ratio achievable in polynomial time even in this online setting.

%%%%%%%%%%%%%%%%%%%%%%%%%%%%%%%%%
\subsection{Algorithm Overview}
\label{subsec:overview}
%%%%%%%%%%%%%%%%%%%%%%%%%%%%%%%%%

\begin{algorithm}[t!]
	\begin{algorithmic}[1]
	\begin{footnotesize}
	\Procedure {INFIDA}{$\allocVecFrac^v_1 {=} \underset{\allocVecFrac^v  \in \allocSetFrac^v \cap \mathcal{D}^v}{\arg\min} \,\Phi^v(\allocVecFrac^v)$, $\allocVec^v_1 {=} \DepRound(\allocVecFrac^v_1)$, $\eta {\in} \mathbb{R}_+$} 	\For{$t = 1,2,\dots,T$}
    	\State 
    	Compute  $\subgradVec^v_t \in \partial_{\allocVecFrac^v} \systemGain(\requestBatchVec_t, \loadVec_t, \allocVecFrac_t)$ through \eqref{eq:subgradient_expression}.
    	\State
    	$\hat{\allocVecFrac}^v_t \gets \nabla \Phi^v (\allocVecFrac^v_t)$\Comment{{Map state to the dual space}}
    	\State
    	$\hat{\vec{h}}^v_{t+1}   \gets \hat{\allocVecFrac}^v_t + \eta \subgradVec^v_t$\Comment{{Take gradient step in the dual space}}
    	\State
    	$\vec{h}^v_{t+1}   \gets \left(\nabla \Phi^v\right)^{-1}(\hat{\vec{h}}^v_{t+1})$\Comment{{Map dual state back to the primal space}}
    	\label{ln:update}
    	%\textcolor{gray}{\Statex (3-5):~$\vec{h}_{t+1}   \gets \left(\nabla \Phi^v\right)^{-1}(\nabla \Phi^v (\allocVecFrac_t)+ \eta \vec{g}_t)$} \Comment{ Given directly by Eq.~\eqref{eq:update_rule}}
    	\State
    	$\allocVecFrac^v_{t+1}\gets \mathcal{P}_{\allocSetFrac^v \cap \mathcal{D}^v}^{\Phi^v}(\vec{h}^v_{t+1})$\Comment{{Project new state onto the feasible region using Algorithm~\ref{alg:bregman_divergence_projection}}}
    	
    	 %\If{\emph{hard-constraint} \textbf{and} $\roundInterval|t$}
    	 
    	 \State $\allocVec^v_{t+1} \gets \DepRound(\allocVecFrac^v_{t+1})$ \Comment{Sample a discrete allocation} %\Comment{$\norm{\allocVec_{t+1}}_1 = h$}
    	 %\ElsIf{$\lnot$\emph{hard-constraint}}
    	 
    	 %\State $\allocVec_{t+1} \gets \Round(\allocVec_t, \allocVecFrac_t,\allocVecFrac_{t+1})$
    	 
	\EndFor
	\EndProcedure
	\end{footnotesize}
	\end{algorithmic}
	\caption{\AlgoName{} distributed allocation on node $v$}
	\label{algo:idn}
\end{algorithm}

On every node $v\in\vertices$, \AlgoName{} updates the allocation $\allocVec^v\in\allocSet^v{\subset} \{0,1\}^{|\modelSet|}${, by} operating on a correspondent fractional state $\allocVecFrac^v\in\allocSetFrac^v{\subset} [0,1]^{|\modelSet|}$, and the fractional allocations satisfy the budget constraint in Eq.~\eqref{eq:capacity-constraint}. Note that, if $\norm{\vec s^v}_1 < b^v$ for  a node $v \in \vertices$, we can always consider fractional allocations that consume entirely the allowed budget; otherwise, all the allocations are set to 1 (node $v$ can store the whole catalog of models). Formally, if $\norm{\vec s^v}_1 \geq b^v$ then
\begin{align}
    \allocSetFrac^v \triangleq
    \left\{ \allocVecFrac^v \in [0,1]^{\modelSet}:
   \sum_{m \in \modelSet} y^v_{m} s_{m}^v =\capacity^v, \forall v \in V\right\};
   \label{eq:allocSetFrac}
\end{align}
otherwise, for the corner case $\norm{\vec s^v}_1 < b^v$, we have $\allocSetFrac^v \triangleq \left\{[1]^\modelSet\right\}$.

 Each variable $\allocFrac^v_{m}$ can be interpreted as the probability of hosting model $m$ on node $v$, i.e., $\allocFrac^v_{m} = \mathbb{P}[\alloc^v_{m}=1] = \mathds{E}[\alloc^v_{m}].$
% \begin{equation}
% \nonumber
%     \allocFrac^v_{m} = \mathbb{P}[\alloc^v_{m}=1] = \mathds{E}[\alloc^v_{m}].
% \end{equation}
%where $\nu$ is the joint probability distribution of variables $\alloc^v_{i,q}$.

We define $\systemGain(\requestBatchVec_t,\loadVec_t, \allocVecFrac)$ as in %\eqref{eq:gain-compact} or, equivalently, as~
\eqref{eq:gain}, replacing $\allocVec$ with $\allocVecFrac$. %In particular, based on~\eqref{eq:several-definitions}, $\auxalloc^{k}_{  \request}(\loadVec_t, \allocVec) = \allocFrac^{v}_{m} \load^{t,v}_{\request,m}$.
Note that $\systemGain(\requestBatchVec_t,\loadVec_t, \allocVecFrac)$ is a concave function of variable $\allocVecFrac \in \allocSetFrac = \bigtimes_{v\in\vertices}\allocSetFrac^v$ (see Lemma~\ref{lemma:concavity} in Appendix~\ref{apx:regret_constants} ).%\rev{aa}{ Indeed, \eqref{eq:gain-compact} shows that $\systemGain(\requestBatchVec_t,\loadVec_t, \allocVecFrac)$ is a linear combination, with positive coefficients, of concave functions (the minimum of affine functions in $\allocVecFrac$)}{}.
 
% since the min operator \rev{aa}{that takes as arguments a constant and a linear function) with positive coefficients.}{is concave, the first argument of the min is constant with respect to $\mathbf{y}$ and the second is linear (see~\eqref{eq:z}).}

Within a time slot $t$, node $v$ collects measurements from messages that have been routed through it (Sec.~\ref{subsec:subgradient}). At the end of every time slot, the node \textit{(i)} computes its new fractional state $\allocVecFrac^v$, and \textit{(ii)} updates its local allocation $\allocVec^v$ via randomized rounding (Sec.~\ref{subsec:rounding}). \AlgoName{} is summarized in Algorithm~\ref{algo:idn} and detailed below.

% , a subgradient  of $\systemGain(\requestBatchVec_t, \loadVec_t, \allocVecFrac_t)$ with respect to $\allocVecFrac^v$, i.e., $\subgradVec^v_{t} \in \partial_{\allocVecFrac^v}\systemGain(\requestBatchVec_t, \loadVec_t, \allocVecFrac_t)$. Observe that $\subgradVec^v_{t}$ is a vector. To get the expression of each of its elements, let us observe that, for any 

% \rev{}{}{Should be moved to state computation}
% {The update is based on the value of the subgradient. 
% To compute it, for any $m\in\mathcal{M}$ and $v\in\vertices$, let us first decompose $\allocVecFrac=(\allocFrac^v_m, \allocVecFrac^{-v}_{-m})$, where $\allocVecFrac^{-v}_{-m}$ is the entire vector $\allocVecFrac$ except the element $\allocFrac^v_m$. Let us denote with $\subgrad^v_{t,m}$ the subderivative of the function $\allocFrac^v_m\rightarrow\systemGain \left(\requestBatchVec_t, \loadVec_t, (\allocFrac^v_m,\allocVecFrac^{-v}_{-m,t}) \right)$ caculated at $y^v_{m,t}$. We denote with $\subgradVec^v_t$ the vector $[\subgrad^v_{m,t}]_{m\in\mathcal{M}}$. It can be shown that the vector $\subgradVec_t=[\subgradVec^v_t]_{v\in\vertices}$ is a subgradient of function $\allocVecFrac\rightarrow\systemGain \left(\requestBatchVec_t, \loadVec_t, \allocVecFrac \right)$ calculated at $\allocVecFrac_t$. Note that there usually are several subgradients. We just need to take one of them.}

\noindent
\textbf{State computation.}
The fractional state $\allocVecFrac^v$ is updated through an iterative procedure aiming to maximize $\systemGain(\requestBatchVec_t,\loadVec_t, \allocVecFrac)$.
This could be the standard gradient ascent method, which updates the fractional state at each node as $\allocVecFrac^v_{t+1} = \allocVecFrac^v_{t} + \eta \subgradVec^v_{t}$,
\iffalse
follows:
% The fractional state $\allocVecFrac^v$ is updated by seeking an optimization of the relaxed function $\systemGain(\requestBatchVec_t,\loadVec_t, \allocVecFrac)$. This can be achieved by a gradient ascent method, i.e., 
\begin{equation}
\nonumber
    \allocVecFrac^v_{t+1} = \allocVecFrac^v_{t} + \eta_t \subgradVec^v_{t},
\end{equation}
\fi
where $\eta_t \in \reals_+$ is the step size  and $\subgradVec^v_{t}$ is 
%the contribution of node $v$ to 
a subgradient of $\systemGain(\requestBatchVec_t,\loadVec_t, \allocVecFrac)$ with respect to $\allocVecFrac^v$. 

In our work, we use a generalized version of the gradient method called Online Mirror Ascent (OMA) \cite[Ch.~4]{bubeck2015convexbook}. OMA uses a function $\Phi^v: \mathcal{D}^v \to \mathbb{R}_+$ (mirror map) to map $\allocVecFrac$ to a dual space before applying the gradient ascent method; %then it maps back the state to the primal space (lines 3-5 of Algorithm~\ref{algo:idn}).
then the obtained state is mapped back to the primal space (lines 3--5 of Algorithm~\ref{algo:idn}).
OMA reduces to the classic gradient ascent method if $\Phi^v$ is the squared Euclidean norm (in this case  the primal space coincides with the dual one).
%, OMA coincides with the classic gradient ascent method. 
Instead, we use the weighted negative entropy map $\Phi^v(\allocVecFrac^v) = \sum_{m\in\modelSet}s^v_{m}\allocFrac^v_{m} \log(\allocFrac^v_{m})$, %which is known to achieve stronger guarantees when the optimization problem has constraints similar to ours~\cite[Sect.~4.3]{bubeck2015convexbook}.
which is known to achieve better convergence rate in high dimensional spaces when each subgradient component is bounded.\footnote{Technically, the advantage in this setting derives from the infinite norm of the subgradient being independent from the space dimension, while the Euclidean norm grows proportionally to the squared root of the space dimension \cite[Sec.~4.3]{bubeck2015convexbook}.}
%when optimizing functions with similar properties to $\systemGain$~\cite[Sect.~4.3]{bubeck2015convexbook} (i.e., the strength of the adversary is reflected by the norms of the revealed subgradients; when using a negative entropy map, the associated norm is the uniform norm that can be much smaller than the Euclidean norm for objectives that induce non-concentrated subgradients).
% \todo{aa: In the extended version we have the space to mention what these properties are.}
%of fractional allocations that satisfy the capacity constraint \eqref{eq:capacity-constraint}
To compute a feasible fractional state $\allocVecFrac^v$, we then perform a projection to the set $\allocSetFrac^v$  on node $v$ (line 6 of Algorithm~\ref{algo:idn}).
We adapt the projection algorithm from \cite{sisalem21icc} to obtain a negative entropy projection $\mathcal{P}^{\Phi^v}_{\allocSetFrac^v\cap \mathcal{D}^v} (\,\cdot \,)$.
\ifnum\isextended=1
Our adaptation is described in Appendix~\ref{appendix:projection}.
\else
Our adaptation is described in~\cite{sisalem21techrep}.
\fi

% At time slot $t$, the potential available capacities $l^{t,v}_{\rho, m}$ (see Sec.~\ref{subsec:load}) are evaluated as follows: at the end of the time slot, at every node $v \in \vertices$ and for each model $m \in \modelSet$  we check how many requests of type $\rho' \neq \rho$ are being served and  we evaluate $\left(L^v_m - \sum_{\rho' \in \requestSet \setminus \{\rho\} } \mathrm{load}(\rho')\right)$, where $ \sum_{\rho' \in \requestSet \setminus \{\rho\} }\mathrm{load}(\rho') $ is the aggregate load of requests of type $\rho' \neq \rho $, i.e.,  the reserved space of the other requests driving $L^v_m$ (maximum model capacity) down.  Finally, this value is capped to $r^t_\rho$ if it is exceeded (see Eq.~\eqref{eq:load_lim}). This defines the subjective $l^{t,v}_{\rho,m}$ (potential available capacity) seen by request type $\rho$.

%\todoi{Update rule of OMA is missing, and it's fair to say we adapted the projection from our paper \cite{salem2021no} rather than \cite{wang2015projection} since it's very different. The adaptation made is unclear, it is described in the document draft.}

\noindent
\textbf{Allocation update.}
Once the fractional state $\allocVecFrac^v$ has been updated, the final step of \AlgoName{} is to determine a new random discrete allocation $\allocVec^v$ and update the local models accordingly.
%reshuffle the local allocated models accordingly. 
The sampled allocation $\allocVec^v$ should \textit{(i)}~comply with the budget constraint
\eqref{eq:capacity-constraint} on node $v$ and \textit{(ii)} be consistent with the fractional state, i.e., $\mathds{E}[\alloc^v_{m}] = \allocFrac^v_{m} \; \forall  m \in \modelSet$. To this purpose, we use the DepRound \cite{byrka2014improved} subroutine (line 7 of Algorithm~\ref{algo:idn}).

In the remainder of this section we detail how each node computes its contribution to the global subgradient, and the rounding strategy used to determine the discrete allocation.

%%%%%%%%%%%%%%%%%%%%%%%%%%%%%%%%%%%%%
\subsection{Subgradient Computation}
\label{subsec:subgradient}
%%%%%%%%%%%%%%%%%%%%%%%%%%%%%%%%%%%%%

%At the end of every time slot $t$, the subgradient $\vec{g}_t$ of the gain function in Eq.~\eqref{eq:gain-compact} at point $\allocVecFrac_t \in \allocSetFrac$ (see
%\ifnum\isextended=1
%Lemma~\red{lemma:subgradient})
%\else
%\cite{sisalem21techrep}) 
%\fi is computed distributionally, where each node $v$ evaluates the $(v,m)$-th components of the subgradient for any $m \in \modelSet$ is given by

At the end of every time slot $t$, a subgradient $\vec{g}_t$ of the gain function in Eq.~\eqref{eq:gain-compact} at point $\allocVecFrac_t \in \allocSetFrac$  is computed in a distributed fashion: each node $v$ evaluates the $(v,m)$-th component of the subgradient for any $m \in \modelSet$ as follows (see Appendix~\ref{appendix:subgradient_expression}):
% Lemma~\ref{lemma:subgradient})
% \ifnum\isextended=1
% Appendix~\ref{appendix:subgradient_expression}
% Lemma~\ref{lemma:subgradient})
% \else
% \cite{sisalem21techrep}) 
% \fi

%(Lemma~\ref{lemma:subgradient})
%local component $\subgradVec^v_{t}$ of a subgradient $\subgradVec_{t} \in \partial_{\allocVecFrac_{t}} \systemGain(\requestBatchVec_t, \loadVec_t, \allocVecFrac_t)$.
%This is computed using only information from the control messages collected within time slot $t$.
%We henceforth omit the subscript $\cdot_{t}$ for brevity.
%Below, we first provide the formulation of a subgradient and then demonstrate how each node can compute it in a distributed fashion.
%Let us denote by $\kappa_{\request} (v, m)$ the order of the cost of model $m$ allocated at node $v$ along the serving path of $\request$, i.e., model $m$ at node $v\in\requestPathVec$ is the $\kappa_{\request} (v, m)$-th best option in terms of cost~\eqref{eq:cost-expression} to serve requests $\request = (i, \requestPathVec)$. If $v\notin \requestPathVec$ we have $\kappa_{\request} (v, m) \triangleq \infty$. At the end of time slot $t\in[\timehorizon]$, we can pick $\subgradVec^v_{t} \in \partial_{\allocVecFrac^v}\systemGain(\requestBatchVec_t, \loadVec_t, \allocVecFrac_t)$ as (Lemma~\ref{lemma:subgradient})
%A subgradient $\subgradVec^v_{t}$ can be computed as follows (Lemma~\ref{lemma:subgradient}):
\begin{align}
    % \nonumber
      \subgrad^v_{t,m} =  \sum_{\request \in \requestSet}  l^{t,v}_{\rho,m} &\cdot
    \left(\smallCost^{\modelsNo^*_{\request}(\allocVecFrac_t)}_{\request} - C^v_{\requestPathVec, m}\right)\cdot\mathds{1}_{\{\kappa_{\request} (v, m) < \modelsNo^*_{\request}(\allocVecFrac_t)\}},
    \label{eq:subgradient_expression}
\end{align}
%where $\modelsNo^*_{\request}(\allocVecFrac_t)$ is the order of the \emph{worst needed} model (i.e., the model with the highest cost and used to serve requests of type $\rho$) that would be needed to serve all the request batch $\requestBatch^t_{\request}$ in the fractional state~ $\allocVecFrac_t$, i.e., $\modelsNo^*_{\request}(\allocVecFrac_t) \triangleq \min\big\{ k \in [\modelsNo_{\request}-1]: \sum^{k}_{k'=1} \auxalloc^{k'}_{\request}(\loadVec_t, \allocVecFrac_t) \geq \requestBatch^t_{\request}\big\}$.
where $\modelsNo^*_{\request}(\allocVecFrac_t)$ is the order of the \emph{worst needed} model, i.e., the model with the highest cost that is  needed to serve all the $\requestBatch^t_{\request}$ requests in the batch given the fractional state~$\allocVecFrac_t$. Formally,   $\modelsNo^*_{\request}(\allocVecFrac_t) \triangleq \min\big\{ k \in [\modelsNo_{\request}-1]: \sum^{k}_{k'=1} \auxalloc^{k'}_{\request}(\loadVec_t, \allocVecFrac_t) \geq \requestBatch^t_{\request}\big\}$.

% The potential available capacities $l^{t,v}_{\rho, m}$ (see Sec.~\ref{subsec:load}) are evaluated as follows. Each node $v$ keeps track of the number of requests of type $\rho' \neq \rho$ served locally by model $m$ during the slot, say it $\mathrm{load}_m(\rho')$, and computes, at the end of the slot, $l^{t,v}_{\rho, m}=\left(L^v_m - \sum_{\rho' \in \requestSet \setminus \{\rho\} } \mathrm{load}(\rho')\right)$.
 %at every node $v \in \vertices$ and for each model $m \in \modelSet$  we check how 
 %track many requests of type $\rho' \neq \rho$ are being served and  we evaluate $\left(L^v_m - \sum_{\rho' \in \requestSet \setminus \{\rho\} } \mathrm{load}(\rho')\right)$, where $ \sum_{\rho' \in \requestSet \setminus \{\rho\} }\mathrm{load}(\rho') $ is the total number of requests of type $\rho' \neq \rho $ served by model $m$, i.e.,  the reserved space of the other requests driving $L^v_m$ (maximum model capacity) down.  Finally, this value is capped to $r^t_\rho$ if it is exceeded (see Eq.~\eqref{eq:load_lim}). This defines the subjective $l^{t,v}_{\rho,m}$ (potential available capacity) seen by request type $\rho$.

For the sake of clarity assume that the mirror map is Euclidean, and then the  dual and primal spaces collapse and $\hat \allocVecFrac_t= \allocVecFrac_t$. At each iteration, each component $\allocFrac^{t,v}_m$ of the fractional allocation vector is updated by adding a mass equal to the product of $\eta >0$ and the corresponding component of the subgradient  (Algorithm~\ref{algo:idn}, line 5). Observe that $g^v_{m,t}$ is the sum of different contributions, one per each request type $\rho$. Thanks to the indicator function, only the terms of the request types that are served by model $m$ on $v$ contribute to $g^v_{m,t}$. This contribution is proportional to the potential available capacity of model $m$ on node $v$ and to 
%(the more requests a model can serve the more is deemed important); 
%moreover, the contribution is also scaled by 
the relative gain $\left(\smallCost^{\modelsNo^*_{\request}(\allocVecFrac_t)}_{\request} - C^v_{\requestPathVec, m}\right)$,
 i.e., the cost reduction achieved when serving request type $\rho$ with model $m$ on $v$, rather than with the worst needed model.  
 Then,  gradient updates add more mass to the models that can contribute more to increase the gain.
 %Because of such update, $\allocVecFrac_t^v + \eta \vec g^v_t$ has extra mass on 
 %``important models.'' 
 On the contrary, the projection step tends to remove the added mass from \emph{all} components to satisfy the constraints. The overall effect is that fractional allocations of more (resp. less) useful models tend to increase (resp. decrease).
 %to  leading to the 
 %The final result is that the fractional  the lesser models, and repeating this operation will drive their allocation to 0.

% \iffalse
% In other words the $\modelsNo^*_{\request}(\allocVecFrac_t)$ models with smallest cost along the path $\requestPathVec$ are required and sufficient to satisfy all requests of type $\request=(i. \requestPathVec)$ in the fractional state.\footnote{
%     Remember that the fractional state is only an abstraction useful to conceive the algorithm, but finally a model is either allocated to a node or not.
% }
% \fi

%\rev{aa}{}{ In other words, using the models along $\requestPathVec$ of quality up to the $\modelsNo^*_{\request}(\allocVecFrac_t)$-th is necessary and sufficient condition to satisfying all requests $\requestBatch^t_{\request}$.}

%\iffalse
{
The subgradient in Eq.~\eqref{eq:subgradient_expression} can be computed at each node using only information from the control messages collected at the end of the time slot $t$. The steps needed to compute the subgradient are as follows.}
%   \item  At the end of the time slot, { when a node generates a request of type $\request \in \requestSet$, it also generates a control message that is propagated along $\requestPathVec$ in parallel to $\request$. The control message contains the quantity $\requestBatch^t_{\request}$ (the multiplicity of the request), and a cumulative counter $Z$ initialized to zero.}
\begin{enumerate}[{1.}]
    \item  At the end of the time slot,  each node generates a control message for every received request type $\rho = (i, \requestPathVec)$ that is propagated along $\requestPathVec$. The control message contains the quantity $\requestBatch^t_{\request} \geq 1$ (the multiplicity of the request), and a cumulative counter $Z$ initialized to zero.
    \item 
    {As the control message travels upstream, intermediate nodes add to $Z$ the local values $\auxalloc^{k}_{\request}(\loadVec_t, \allocVecFrac_t)$ (fractional effective capacity in Eq.~\eqref{eq:several-definitions}). 
    % Note that each node adds $|\mathcal{M}_i|$ fractional effective capacity,  one per each model $m\in\mathcal{M}_i$, where $i$ is the task requested by request type $\rho$. 
    These values are added following increasing values of cost. This message is propagated until $Z \geq \requestBatch^t_{\request}$, that is until the message reaches the $\modelsNo^*_{\request}(\allocVecFrac_t)$-th model.}
    \item 
    % $\smallCost^{\modelsNo^*_{\request}(\allocVecFrac_t)}_{\request}$%
    {Once the $\modelsNo^*_{\request}(\allocVecFrac_t)$-th model is detected, a control message is sent down in the opposite direction, containing %two fields: a cumulative cost $H$ initialized to zero, and the previously seen cost $\Gamma$, initially set to $\gamma_\request^{\modelsNo^*_{\request}(\allocVecFrac_t)}$.}
    the cost $\smallCost^{\modelsNo^*_{\request}(\allocVecFrac_t)}_\request$ of the last checked model. %Every virtual node $(v,m)$ in the reverse direction sniffs these values and before replacing it with its cost value and cumulative cost, learns the quantity
    Every node $v$ in the reverse direction reads the cost value from the control message and, for each model $m\in\modelSet_i$, computes the quantity
    %\rev{ts}{$h^v_{m} = \sum^{\modelsNo^*_{\request}(\allocVecFrac_t) - 1}_{k= \kappa_{\request}(v, m)} \auxload^{k}_{  \request}(\loadVec_t) \left(\smallCost^{k+1}_{\request} - \smallCost^{k}_{\request}\right). $}
    \begin{align}
        h^v_{m} = l^{t,v}_{\rho,m} \cdot \left ( \smallCost^{\modelsNo^*_{\request}(\allocVecFrac_t)}_\request - C^v_{\requestPathVec, m}\right).
    \end{align}}
    %After this quantity is learned the current node replaces the values in the field by its cost $\smallCost^{k}_{\request}$ and sets $Z' =\sum^{\modelsNo^*_{\request}(\allocVecFrac_t) - 1}_{k= \kappa_{\request}(v, m)} \load^{v,t}_{\rho, m}\left(\smallCost^{k+1}_{\request} - \smallCost^{k}_{\request}\right)$ to be used by the virtual node downstream.
   %\rev{ts}{ This value is computed as $H + \load^{t, v}_{\request, m} (\gamma^{\kappa_{\request}(v, m)}_\request - \Gamma)$. After node $v$ learns this quantity for each model $m$, it updates the values in the control message with $H = h^v_m$ and $\Gamma = \gamma_\request^{\kappa_{\request}(v, m)}$.}{}
    
    % \gc{Intuitively, quantities $h^v_{m}$ represent the cumulative cost reduction achieved by serving requests on a local model, compared to the worst model participating in serving requests of type $\rho$ at time slot $t$.}
    \item
    % Let $\mathcal{H}^v_{m}$ be the set of quantities collected in this manner during time slot $t$ on node $v$. Node $v$ can then compute
    Node $v$ can then compute $g_{t,m}^v$ in Eq.~\eqref{eq:subgradient_expression} as follows
    \begin{align}
    \nonumber
        \subgrad^v_{t,m}=  \sum_{m \in \modelSet_i} h^v_{m}.
        %= \sum_{\request \in \requestSet} l^{t,v}_{\rho,m} \cdot \left ( \smallCost^{\modelsNo^*_{\request}(\allocVecFrac_t)}_\request - C^v_{\requestPathVec, m}\right)  \\
        %\quad\cdot \mathds{1}_{\{\kappa_{\request} (v, m) < \modelsNo^*_{\request}(\allocVecFrac_t)\}}
    \end{align}
\end{enumerate}
Note that the cost in Eq.~\eqref{eq:serving_cost2} does not necessarily increase along the path. Therefore, a traversed node is not able to update directly the variable $Z$ when there exist upstream nodes with lower cost. In this case, the node simply appends the information $(\auxalloc^{k}_{\request}(\loadVec_t, \allocVecFrac_t), \smallCost^k_{\request})$ to the message, and lets upstream nodes to apply any pending update in the correct order. %A similar adaptation applies for the downstream message.
%\gc{Note that, as cost \eqref{eq:cost-expression} is not monotonic w.r.t. the position in the path, models may not be checked in order. Therefore, the process above is implemented in practice with a slight modification: any time a not-in-order model with cost index $k$ is checked, the node does not update value $Z$ but appends information $(\auxalloc^{k}_{\request}(\loadVec_t, \allocVecFrac_t), \smallCost^k_{\request})$ to the message instead, and upstream nodes apply any pending update in the correct order. %A similar adaptation applies for the downstream message.
%}
%\fi
\newton{In our work, we assume to operate on a \emph{reliable} communication channel. Nonetheless, we note that \AlgoName{} is robust to noise $\pmb \xi_t$ affecting the sub-gradient, as long as such noise is not biased, i.e., $\mathbb E [\pmb\xi_t] = \pmb 0$ for every timeslot $t$~\cite[Theorem~3.4]{hazan2016introduction}.}

\subsection{State Rounding}
\label{subsec:rounding}
%%%%%%%%%%%%%%%%%%%%%%%%%%%%%%%%

Once the new fractional state $\allocVecFrac_{t+1}$ is computed, each node~$v$ independently draws a random set of models to store locally in such a way that $\mathds{E}[\allocVec^v_{t+1}] = \allocVecFrac^v_{t+1}$. This sampling guarantees that the final allocation $\allocVec^v_{t+1}$  satisfies constraint~\eqref{eq:capacity-constraint} in expectation.
A naive approach is to draw each variable $\alloc^{v,t+1}_{m}$ independently, but it leads to a large variance of the total size of the models selected, potentially exceeding by far the allocation budget at node $v$.
% Berandom with  variable that is assigned to 1 with probability $\allocFrac^v_{m}$
% We execute a rounding algorithm in a distributed fashion at each node $v \in \vertices$. Note that constructing a trivial solution where each $\alloc^v_{m}$ is an independent random variable that is assigned to 1 with probability $\allocFrac^v_{m}$ yields an allocation $\allocVec^v$ that satisfies constraint \eqref{eq:capacity-constraint} only in expectation. Instead, the actual random allocation can largely deviate from the capacity $\capacity^v$.

To construct a suitable allocation we adopt the DepRound procedure from \cite{byrka2014improved}. The procedure modifies the fractional state $\allocVecFrac^v_{t+1}$ iteratively: at each iteration, DepRound operates on two fractional variables $\allocFrac^{v,t+1}_{m},\allocFrac^{v,t+1}_{m'}$ so that at least one of them becomes integral and the aggregate size of the corresponding models $s^v_{m}\allocFrac^{v,t+1}_{m}+s^v_{m'}\allocFrac^{v,t+1}_{m'}$ does not change. This operation is iterated until all variables related to node $v$ are rounded except (at most) one{, which we call residual fractional variable}. This is done in $\mathcal{O}(|\modelSet|)$ steps.

Note that, to satisfy $\mathds{E}[\allocVec^v_{t+1}] = \allocVecFrac^v_{t+1}$, the residual fractional variable, say it  $\allocFrac^{v,t+1}_{\bar m}$, needs to be rounded. At this point $\alloc^{v,t+1}_{\bar m}$ can be randomly drawn. Now the final allocation can exceed the budget bound $\capacity^v$ by at most $s_{\bar m}$.
%according to a Bernoulli 
%this variable can be selected 
%This could be done again selecting with by randomly sampling it. However, in this case the deviation of the random allocation from the capacity bound $\capacity^v$ is limited by $s_{\bar m}$. 
These (slight) occasional violations of the constraint may not be a problem, e.g., at an edge server running multiple applications, where resources may be partially redistributed across different applications; they may be explicitly accounted for in the service level agreements.
If the budget bound cannot be exceeded even temporarily, the node is not  able to store the model $\bar m$, but it may still exploits the residual free resources to deploy the model that provides the best marginal gain among those that fit the available budget. In practice, we expect the corresponding gain decrease to be negligible.
%In practice, we can cope with the dangling variable adopting several heuristics{, e.g., deploy the model that provides the best marginal gain among those that fit the constraint, or admit elastic budgets}.

% \todoi{aa: The fact that at every time-slot (how large is a time-slot?) the models saved (after rounding) are completely independent from the previous round can degenerate in a crazy distributed systems where models flip on and off at every time-slot, imposing massive migrations that would dominate bandwidth utilization or at least imposing continuous transfer from storage to main memory, which would dominate computation time. Shall we discuss this point?}
%!TEX root = ../paper.tex
\section{Theoretical Guarantees}
\label{sec:guarantees}

We provide the optimality guarantees of our \AlgoName{} algorithm in terms of the $\psi$-regret~\cite{krause2014submodular}.
In our scenario, the $\psi$-regret is defined as the gain loss in comparison to the best static allocation in hindsight, i.e., $\allocVec_* \in \argmax_{\allocVec \in \consSet} \sum^T_{t=1}    \systemGain(\requestBatchVec_t, \loadVec_t, \allocVec)$,
%We denote by $\requestBatchSet$ the set of all request batches available to the adversary, i.e., $\requestBatchVec_t \in \requestBatchSet, \forall t \in [\timehorizon]$, and by $\loadSet$ the set of all the potential available capacities available to the adversary, i.e., $\loadVec_t \in \loadSet, \forall t \in [\timehorizon]$.
%The $\psi$-regret is computed discounting the maximum gain 
discounted by a factor $\psi \in (0,1]$. Formally,
\begin{comment}
\begin{align}
\nonumber
&\textstyle \psi \text{-} \mathrm{Regret}_{T, \allocSet} 
=
\\
&\textstyle \underset{
%\{ \requestVec_1, \dots, \requestBatchVec_t\} \in \naturals^{|\requestSet|\times|\timehorizon|}
%\{\requestBatchVec_t\}_{t\in[T], \requestBatchVec_t \in \naturals^\requestSet}
\{\requestBatchVec_t, \loadVec_t \}_{t=1}^{T} \in (\requestBatchSet \times \loadSet)^T 
}{\sup} 
    \left\{ \psi\sum^T_{t=1}    \systemGain(\requestBatchVec_t, \loadVec_t, \allocVec_*) - \mathbb{E} \left[\sum^T_{t=1}    \systemGain(\requestBatchVec_t, \loadVec_t, \allocVec_t)\right] \right\},
    \nonumber
%\label{eq:regret}
\end{align}
\end{comment}
\begin{align}
\label{eq:regret-definition}
&\psi \text{-} \mathrm{Regret}_{T, \allocSet} 
\triangleq
%=
\\ &\underset{
%\{ \requestVec_1, \dots, \requestBatchVec_t\} \in \naturals^{|\requestSet|\times|\timehorizon|}
%\{\requestBatchVec_t\}_{t\in[T], \requestBatchVec_t \in \naturals^\requestSet}
{\scriptscriptstyle{\{\requestBatchVec_t, \loadVec_t \}_{t=1}^{T} \in \advSet^T}}
}{\sup}\hspace{-0.1em}
    \left\{ \psi\sum^T_{t=1}    \systemGain(\requestBatchVec_t, \loadVec_t, \allocVec_*){-}\mathbb{E}\left[\sum^T_{t=1}    \systemGain(\requestBatchVec_t, \loadVec_t, \allocVec_t)\right] \right\}\hspace{-0.1em},
    \nonumber
%\label{eq:regret}
\end{align}
where 
%the supremum is over all the possible sequences $\{\requestBatchVec_t, \loadVec_t \}_{t=1}^{T} \in (\requestBatchSet \times \loadSet)^T$,
allocations $\allocVec_t$ are computed using \AlgoName{} and the expectation is over the randomized choices of \DepRound.
% while $\requestBatchSet$ and $\loadSet$ are respectively the set of all possible request batches and potential available capacities, i.e., $\requestBatchVec_t \in \requestBatchSet$ and $\loadVec_t \in \loadSet, \forall t \in [\timehorizon]$.
%\gc{We assume that the number of requests that the system can receive at any time slot $t \in [T]$ is bounded by $R\in\naturals$, with $R$ arbitrarily big.}
%We assume that $R$ is the maximum number of requests that can be received at any time slot $t \in [T]$, with $R$ arbitrarily big.
Note that, by taking the supremum over all request sequences and potential available capacities, we measure regret in an adversarial setting, i.e., against an adversary that selects, for every $t \in [T]$, vectors $\requestBatchVec_t$ and $\loadVec_t$ to jeopardize 
%\gcdelete{system performance} 
the performance of our algorithm. {Obviously, we do not expect such an adversary would exist in reality, but the adversarial analysis provides bounds on the behavior of \AlgoName{} in the worst case.}

The adversarial analysis is a modeling technique to characterize 
%The assumption of the existence of an adversary is merely a technique to describe the 
system performance under highly volatile external parameters (e.g., the sequence of requests $\requestBatchVec_t$) or difficult to model system interactions (e.g., the available capacities $\loadVec_t$). This technique has been recently successfully used to model caching problems~(e.g., in \cite{paschos19,sisalem21icc}). %bhattacharjee20
%An adversary can represent in many applications a volatile environment that keeps changing the systems parameters.   
Our main result is the following (the full proof is 
\ifnum\isextended=1
 in Appendix \ref{appendix:regret_proof}):
\else
 in \cite{sisalem21techrep}):
\fi

\begin{theorem}
\label{th:regret_bound}
\textnormal{\AlgoName{}} has a sublinear \mbox{$(1-1/e)$-regret} w.r.t.  the time horizon $T$,
% \todo{aa: I wouldn't say it is sublinear ``in the sequence of requests and potential capacity''. I would say it is sublinear with respect to the time horizon, and linear with respect to the sequence of requests (the batch size, more precisely) and the potential capacities}
i.e., there exists a constant $A$ such that:
\begin{align}
   \text{$\left(1-1/e\right)$-} \mathrm{Regret}_{T, \allocSet} \leq A \sqrt{T},
    \label{eq:alpha_regret}
\end{align}
where $A \; {\propto} \; R L_{\max} \Delta_C$.
{$R$, $L_{\max}$, and  $\Delta_C$ are upper bounds, respectively, on the total number of request types at any time slot, on the model capacities, and on the largest serving cost difference between serving at a repository node and at any other node.}
%We assume that the number of requests that the system can receive at any time slot $t \in [T]$ is bounded by $R\in\naturals$, with $R$ arbitrarily big.
\end{theorem}

\begin{proof} (sketch)
We first prove that the expected gain of the randomly sampled allocations $\allocVec_t$ is a $(1{-}1/e)$-approximation of the fractional gain. Then, we use online learning results~\cite{bubeck2015convexbook} to bound the regret of Online Mirror Ascent schemes operating on a convex decision space and against concave gain functions picked by an adversary. The two results are combined to obtain an upper bound on the $(1{-}1/e)$-regret. We fully characterize the regret constant $A$ 
\ifnum\isextended=1
 in Appendix \ref{appendix:regret_proof}.
%  \todo{aa: I think calling $A$ ``constant'' may be misleading.}
\else
 in \cite{sisalem21techrep}.
\fi
\end{proof}
\newton{Note that this result holds over the 
%fractional 
integral domain (see Appendix \ref{apx:regret_constants}, Lemmas~\ref{lemma:util_bounds_lower}--\ref{lemma:el_last}), thus generalizing the approximation techniques in~\cite{stratis, NEURIPS2021_2387337b,shanmugam2013femtocaching} and providing a novel result in approximating budget-additive (submodular) set functions~\cite{budget_additive}.}
% of the form~$f(S) = \min\left\{B, \sum_{i \in S} w_i\right\}$.}

%\toreview{
We observe that the regret bound depends crucially on the maximum number of request types $R$, maximum model capacity $L_{\max}$ and maximum serving cost difference $\Delta_C$. When considering the cost model in Eq.~\eqref{eq:serving_cost2}, we can consider for $\Delta_C$  the sum of the total latency of the heaviest path, the parameter $\alpha$, and the largest inference delay.
%difference between the repository and the most constrained nodes. 
This result is intuitive: when these values are bigger, the adversary has a larger room to select values that can harm the performance of the system.
%}

% \ts{
% When the sum of the model sizes is of the same order as the budget constraint the regret approaches 0, since we can just store all the models and the static optimum is required to perform the same decisions. On the other hand, when the budget capacities are decreased the regret also approaches 0, since the optimization space is reduced (we can't store anything) that harms the performance of our policy as well as the static optimum.
% }

As a direct consequence of Theorem~\ref{th:regret_bound}, the expected time averaged $(1-1/e)$-regret of \AlgoName{} can get arbitrarily close to zero for large time horizon. Hence, \AlgoName{} achieves a time averaged expected gain that is a $(1-1/e-\epsilon)$-approximation of the optimal time averaged static gain, for arbitrarily small~$\epsilon$. 

Observe that \AlgoName{} computes a different $\allocVec_t$ at every time slot. Intuitively, this allows it to ``run after'' the exogenous variation of the adversarial input $\{\requestBatchVec_t, \loadVec_t \}_{t=1}^{T} \in \advSet^T$. An alternative goal that can be achieved by \AlgoName{} is to find a static allocation $\bar{\allocVecFrac}$. In order to do so, we need to \emph{(i)}~run INFIDA for $\tilde T$ time-slots, \emph{(ii)}~based on the $\{\allocVecFrac_t\}_{t=1}^{\tilde T}$ computed by \AlgoName{}, calculate $\bar{\allocVec}$ (the exact calculation is in Proposition~\ref{corollary:offline_solution}), \emph{(iii)}~deploy in the IDN the allocation $\bar{\allocVec}$ and keep it static, in order to avoid switches. Obviously, we would like the quality of $\bar{\allocVec}$ to be close to the best $\allocVec_*$, defined in~\eqref{eq:static_gain}. The following proposition shows that the gain achieved with our $\bar{\allocVec}$ is boundedly close to the optimum. Moreover, since~\eqref{eq:static_gain} is NP-hard, there cannot exist better bounds than the one we achieve, assuming $\mathrm{P} \neq \mathrm{NP}$~\cite{nemhauser1978best}.
%This observation also suggests that \AlgoName{} can be used as a numerical solver for the NP-hard static allocation problem~\eqref{eq:static_gain} obtaining the best approximation bound achievable for this kind of problems assuming $\mathrm{P} \neq \mathrm{NP}$~\cite{nemhauser1978best}. This result is formally stated by the following proposition: %configuration $\allocVec_*$ on average.
\begin{proposition} (offline solution)
\label{corollary:offline_solution}
\iffalse
Let $\bar{\allocVecFrac}$ be the average fractional allocation $\bar{\allocVecFrac} =\frac{1}{T} \sum^T_{t=1} \allocVecFrac_t$. $\forall \epsilon > 0$ and over a sufficiently large time horizon $T$, $\bar{\allocVecFrac}$ satisfies the following:

\begin{align}
        \frac{1}{T}\sum^T_{t=1}\systemGain(\requestBatchVec_t, \loadVec_t, \bar{\allocVecFrac}) \geq \left(1-\frac{1}{e} - \epsilon\right) \frac{1}{T}\sum^T_{t=1}    \systemGain(\requestBatchVec_t, \loadVec_t, \allocVec_*).
\end{align}
\fi
Replace in \AlgoName{} the allocation gain $G(\requestBatchVec_t, \loadVec_t, \allocVecFrac)$ by $G_T(\allocVecFrac)$ (defined in~\eqref{eq:static_gain}). %After $\tilde T$ iterations, and it is run against the fixed time averaged gain $G_T $ over a running time $\tilde T$. 
After $\tilde T$ iterations,
let  $\bar{\allocVecFrac}$ be the average fractional allocation $\bar{\allocVecFrac} =\frac{1}{\tilde{T}} \sum^{\tilde{T}}_{\tilde
 t=1} \allocVecFrac_t$, and $\bar{\allocVec}$ the random state sampled from $\bar{\allocVecFrac}$ using \textnormal{\DepRound{}}.
$\forall \epsilon > 0$, for $\bar T$ large enough,
%over a sufficiently large running time $\tilde{T}$, 
$\bar{\allocVec}$ satisfies %the following:
\begin{align}
\mathbb{E} \left[ G_T(\bar{\allocVec}) \right] \geq \left(1-\frac{1}{e} - \epsilon\right) G_T({\allocVec}_*), 
\end{align}
where $ {\allocVec}_* = \argmax_{\allocVec \in \allocSet} G_T(\allocVec)$.
\end{proposition}
The proof is given in Appendix.~\ref{proof:offline_solution}.
\section{Experimental Results}
\label{sec:validation}

%\todoi{I put gtx 980 at the edge. Nano is too slow. Nothing is served at the edge.}
We evaluate \AlgoName{} by simulating a realistic scenario based on the typical structure of ISP networks. We compare our solution with a greedy heuristic and its online variant (described below), as the greedy heuristic is known to achieve good performance in practice for submodular optimization~\cite{krause2014submodular}. 
%that we designed inspired by related work in caching networks (described below).

\iffalse
\todoi{ Experiments Parameters:
- Tiers: 4 3 2 1 0
- Capacities per server: 4 12 18 30 infinite.
- Nodes:60 20 4 1 1.
- Query nodes: 60.
- Tasks: 20.
- Requests popularity distribution: Zipf with exponent 1.2.(Same for sliding except it is rolled as the figure shows. 

Ask here if something is missing:
..

Additional info:
Scaling: x1.6, x1.11, x5.44, x5.33 (compared to original proposal, which is instead used in figure 1)

Fig.2 and Fig.3 are with 5000 rps.

The alpha for Fig.4 is 0.5

\textbf{How much is the repository cost in Figure 4?} I computed it to be around 70, is it correct? Then a gain of 60 seems too high, doesn't it?

It can't be 70.  Network latency ~60 ms + inference latency ~ 24ms + inaccuracy cost(alpha=1) ~ 35. It's around 120 per request.

Right I skipped the 24ms of inference latency. The mean network latency is 54 ms Since in figure 4 the alpha is 0.5, then the mean cost per request is around 95. 
)
}
\fi

\noindent
\textbf{Topology.}
We simulate a hierarchical topology similar to \cite{ceselli2017mobile} that spans between edge and cloud, with different capacities at each tier. We consider 5 tiers: base stations (tier 4), central offices (tiers 3, 2), ISP data center (tier 1), a remote cloud (tier 0). 
We assume a hierarchical geographic distribution similar to LTE.
% with RTTs ranging from few milliseconds between tiers 2, 3, and 4, to tens of milliseconds between tiers 0, 1, and 2. 
We take the Round-Trip Time (RTT) across the different tiers as follows: tier 4 to tier 3 takes 6 ms, tier 3 to tier 2 takes 6 ms, tier 2 to tier 1 takes 15 ms, and tier 1 to tier 0 takes 40 ms. 
% The average RTT from base stations to cloud is about $50$~ms. 
\gc{We execute our experiments at two different scales:} %\rev{}{}{in \emph{ \emph{Network Topology~I}} we scale the topology to a scenario with 60 base stations (86 nodes in total) to obtain a {hierarchical network}, while \emph{ \emph{Network Topology~II}} is obtained by further scaling down to 2 base stations (5 nodes in total), giving a T-shaped network topology with two leaves.}
  \emph{Network Topology~I} counts 24 base stations and 36 nodes in total, while  \emph{Network Topology~II} is a simpler 5-node scenario with 2 base stations.
 %and 5 nodes in total, giving a T-shaped network topology with two leaves.

\noindent
\textbf{Processing Units.}
 %We assume servers on tiers 0-1 are equipped with high-end GPUs (Titan RTX), those on tiers 2-3 with mid-tier GPUs (GeForce GTX 980), and each base station (tier 1) with a single low capacity GPUs (Geforce GTX 960).
% We assume tier 0 is equipped with high-end GPUs (Titan RTX), while other tiers are equipped with mid-range GPUs (GeForce GTX 980). However, we limit the available GPU memory on 75\% of the base stations to 1 GB (instead of 4GB). We take GPU memory of the computing nodes as the limiting budget. 
We take GPU memory of the computing nodes as the limiting budget. The node at tier 0 can store the entire models catalog. We simulate the performance of two different processing units: the computing nodes at tiers 0 and 1 are equipped with high-end GPUs (Titan RTX), and the remaining tiers 2--4 have mid-tier GPUs (GeForce GTX 980). The budget of each computing tier is given as follows: a tier-1 node has 16GB GPUs, a tier-2 node has 12GB GPUs, a tier-3 node has 8GB GPUs, and a tier-4 node has 4GB GPUs.

\noindent
\textbf{Catalog and requests.}
We simulate performance based on state-of-the-art pre-trained models and their pruned versions~\cite{bochkovskiy2020yolov4, cai2020yolobile}, profiled for each simulated processing unit, for a total of 10 models (Table~\ref{tab:catalog}). We consider a task catalog with $|\mathcal{N}| = 20$ different object detection tasks. We allow $3$ duplicates per model; this gives $|\mathcal{M}_i| = 30$ alternative models per task $i \in \mathcal{N}$. Note how, as model complexity decreases, the number of frames a GPU can process per second increases, and consequently the average inference delay decreases. 

The time slot duration is set to 1 minute and requests arrive at a constant rate of 7,500 requests per second (rps), unless otherwise said. 
Each request type is assigned randomly to two base stations in tier 4. The corresponding task is selected according to two different popularity profiles: \emph{(i)} in the {\emph{Fixed Popularity Profile}} (Fig.~\ref{fig:popularitya}), a request is for task $i$  with constant probability  $p(i) = \frac{(i+1)^{-1.2}}{\sum_{i' \in \mathcal{N}}(i'+1)^{-1.2}}$ (a Zipf distribution with exponent $1.2$), while \emph{(ii)} in the {\emph{Sliding Popularity Profile}} (Fig.~\ref{fig:popularityb}), the $l$-th consecutive request is for task $i$ with probability $\tilde{p} (i,l) = p\left(\left(i + 5\floor{l / W}\right)\mod20\right)$, that is, the popularity of the tasks changes through a cyclic shift of 5 tasks every 1 hour for a request rate of 7,500 rps ($W = 2.7 \times 10^7$).
 
%  corresponds to a window duration of $1$ hour for request rate 7,500 rps.
%  or a in  according to two different popularity profiles. Under the \emph{\emph{Fixed Popularity Profile}} task $i$ is sampled 

% according to a Zipf distribution with exponent $1.2$, i.e., with probability proportional to $(i+1)^{-1.2}$. Under the \emph{\emph{Sliding Popularity Profile}} task  in which tasks are sampled from the same Zipf distribution but the  we change the popularity of the tasks through a cyclic shift at fixed time intervals (Fig.~\ref{fig:popularity}b). %We assume requests are first received from the base stations, and each base station may receive up to 500 requests per second.

% In our simulations, requests are first received from the base stations, and we submit to the system up to 15000 requests per second (rps). The duration of a time slot is  taken to be $0.5$ seconds.

%\subsection{Evaluation Settings}

\begin{table}[t]
\centering
\caption{Catalog for variants of YOLOv4~\cite{bochkovskiy2020yolov4} profiled on two different Processing Units. Accuracy is for the MS COCO dataset. Values for the pruned variants are adapted from~\cite{cai2020yolobile}.\vspace{1em}}
\label{tab:catalog}
\begin{footnotesize}
\begin{tabular}{lrrrr}
\toprule
  \multirow{2}{*}{\shortstack{variants \\ of yolov4}} %
  &\multirow{2}{*}{\shortstack{accuracy \\ (mAP@0.5)}} 
  &\multirow{2}{*}{\shortstack{memory \\ (MB)}} 
  &\multicolumn{2}{c}{frames per second} \\
      & & 
      & Titan RTX % 
      & GTX 980 % 
      %& GTX 960 
      \\    
    \midrule
        608p        & 65.7  & 1577  & 41.7  & 14.2 \\ %& 6.13  \\
        512p        & 64.9  & 1185  & 55.5  & 18.9 \\ %& 8.15  \\
        416p        & 62.8  & 1009  & 73.8  & 25.1 \\ % & 10.9  \\
        320p        & 57.3  & 805   & 100   & 34.1 \\ % & 14.7  \\
        3.99pruned  & 55.1  & 395   & 209   & 71.0 \\ % & 30.7  \\
        8.09pruned  & 51.4  & 195   & 329   & 112  \\ % & 48.4  \\
        10.10pruned & 50.9  & 156   & 371   & 126  \\ % & 54.6  \\
        14.02pruned & 49.0  & 112   & 488   & 166  \\ % & 71.7  \\
        tiny-416p   & 38.7  & 187   & 888   & 302  \\ % & 131   \\
        tiny-288p   & 34.4  & 160   & 1272  & 433  \\ % & 187   \\
    \bottomrule
\end{tabular}
\end{footnotesize}
\end{table}

\begin{figure}[t]
    \centering
    \subcaptionbox{\emph{Fixed Popularity Profile}\label{fig:popularitya}}{\includegraphics[clip= true, width = .45\linewidth, trim= 0.0in 0.0in 0.0in 0.05in]{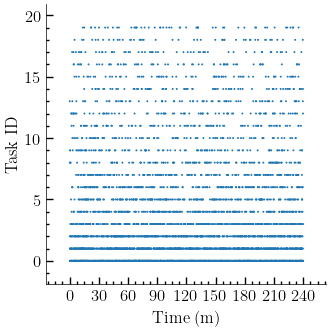}
    }
    \subcaptionbox{\emph{Sliding Popularity Profile}\label{fig:popularityb}}{\includegraphics[clip= true, width = .45\linewidth, trim= 0.0in 0.0in 0.0in 0.05in]{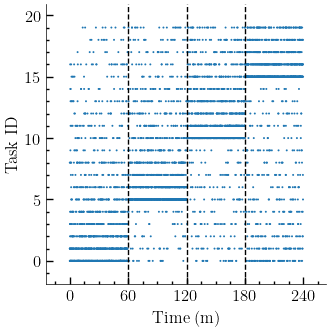}
    }
    % \subcaptionbox{PSN Popularity Profile}{\includegraphics[clip= true, width = .3\linewidth, trim= 0.0in 0.0in 0.0in 0.05in]{sections/figs/requests_PSN.png}}
    \caption{Popularity profiles of inference tasks for request rate $7,500$ rps.  Each dot represents a request for a given task at given time.  Therefore, popular tasks correspond to denser lines. In (b), popularity changes at fixed time intervals through a cyclic shift.
    At any moment the requests are i.i.d. and sampled from a Zipf distribution with exponent $1.2$. The figure shows a random down-sample of 5,000 requests to emphasize the density difference across the different tasks.}
    \label{fig:popularity}
\end{figure}

\noindent
\textbf{Static greedy.} \rev{ts}{}{\gc{We adapt the static greedy (\SG) heuristic} from the cost-benefit greedy  in~\cite{krause2014submodular}. %The selected allocation is initially empty, then we allocate the model with the highest marginal gain scaled down by its size while meeting the budget constraint. When a model is selected, we compute the amount of consumed request load. Then, we obtain a different static objective with the reduced request load. We repeat the same process until all the requests are served by the models placement.
\gc{\SG{} operates in hindsight seeking maximization of the time averaged allocation gain over the whole time horizon $T$, as in Eq.~\eqref{eq:static_gain}.}
\gc{Starting from an empty allocation, this policy progressively allocates the model that provides the highest marginal gain normalized by size, among those that meet the budget constraints. This process is repeated until either the intermediate allocation is capable of serving all requests or none of the remaining valid allocations introduces a positive marginal gain.}
}

\noindent
\textbf{Online load-aware greedy heuristic.}
\iffalse
\rev{ts}{}{The online load-aware greedy (OLAG) heuristic is the online variant of \SG.} When a request is forwarded along $v \in \vertices$, we keep track of the number of times a task is requested and the cost reduction of models serving that task compared to the repository. At each node $v$, an importance weight is computed $w_{m} = \frac{g_m}{s_m} \min\{r_{i,m}, L_m\}$, where $r_{i,m}$ is the number of times a task $i \in \mathcal{N}$ is seen by  node $v$ and model $m$ could serve it, $g_m$ is the cost reduction, and $s_m$ is the size of the model $m$. While $\min\{r_{i,m}, L_m\}$ are the potential requests that can be served by the model $m$.

The node selects the model $m_*$ with the highest importance while respecting the capacity constraint, then we subtract the potential requests that it can serve ($\min\{r_{i,{m_*}}, L_{m_*}\}$) from $r_{i, m_*}$ and all $r_{i, m'}$ if $m'$ can can serve task $i$ but at a higher cost. The same procedure is performed until the capacity of the node is filled.
\fi
%\iffalse
% attempt alternative version
As \AlgoName{} is the first online policy for ML models' allocation in IDNs, there is no clear baseline to compare it with. We then propose an online heuristic based on \SG, which we call online load-aware greedy (\OnlineGreedy). A node $v$ uses counters $\phi^v_{m,\request}$ to keep track of the number of times a request $\request\in\requestSet$ is forwarded upstream but could have been served locally at a lower cost compared to the repository, i.e., using a model $m\in\modelSet$ with positive gain that we denote by $q^v_{m,\request}$.
For every model $m$, an importance weight is computed as $w^v_{m} {=} \frac{1}{s^v_m}\frac{1}{|\requestSet|}\sum_{\request\in\requestSet}q^v_{m, \request}\min\{\phi^v_{m,\request},L^v_m\}$, where $s^v_m$ is the size of model $m$ and $\min\{\phi^v_{m,\request},L^v_m\}$ is the number of requests that could have been improved by $m$.
At the end of a time slot, the node selects the model $m_*$ with the highest importance while respecting the resource budget constraint, then subtracts the quantity $\min\{\phi^v_{m,\request},L^v_m\}$ from $\phi^v_{m_*, \request}$ and from all the $\phi^v_{m',\request} \colon q^v_{m',\request}<q^v_{m_*,\request}$, i.e., models that provide a gain lower than $m_*$. This procedure is repeated until the resource budget of the node is consumed.
%\fi

\noindent\textbf{Offline \AlgoName{}}. Motivated by Proposition~\ref{corollary:offline_solution}, we implemented also an offline version of \AlgoName{} that we call \AlgoNameOffline, which maximizes the time-averaged gain \eqref{eq:static_gain} over the whole time horizon $T$. The potential available capacities are determined at runtime from the current allocations and request batches (rather than by an adversary).
%are not chosen by an adversary, they can be determined experimentally from the picked allocations and the request batches, so we only take the offline set of requests' batches as input and let the potential available capacities be determined from the allocations given by the policy. 

\noindent
\textbf{Performance Metrics.} The performance of a policy $\mathcal{P}$ with the associated sequence of allocation  decisions $\{\allocVec_t\}^T_{t=1}$ is evaluated in terms of the time-averaged gain normalized to the number of requests per time slot (NTAG):
%; The NTAG evaluated over $T$ requests is defined as: 
\begin{align}
       \mathrm{NTAG} (\mathcal{P}) = \sum^T_{t=1} \frac{1}{T  \norm{\vec{r}_t}_1 } G(\vec{r}_t,\vec{l}_t, \vec{x}_t).
\end{align}
\noindent Moreover, we evaluate the update cost of a policy  $\mathcal{P}$ with the associated sequence of allocation decisions $\{\allocVec_t\}^T_{t=1}$  by quantifying the total size of fetched models over $T$ time slots. The update cost is reflected by the Time-Averaged Model Updates (MU) metric defined as:
\begin{align}
    \mathrm{MU} (\mathcal{P}) \triangleq  \frac{1}{T} \sum^{T}_{t=2} \sum_{(v,m) \in \vertices \times \modelSet} s^v_m \max\{0, x^v_{t,m} - x^v_{t-1,m}\}.
    \label{eq:update_cost_def}
\end{align}
% Note that only the models that are being fetched are accounted for and it is captured by the $\max\{0, \,\cdot\,\}$ operator.

\subsection{Trade-off between Latency and Accuracy}

We first evaluate how \AlgoName{} adapts to different trade-offs between end-to-end latency and inference accuracy by varying the trade-off parameter $\alpha$.\footnote{The inaccuracy cost is taken in 0--100, then $\alpha$ picked here corresponds to $100\times$ scaling of the parameter defined in Eq.~\eqref{eq:serving_cost2}.} %For a better visualization of the trade-offs, in this set of experiments we focus on a subset of our catalog (Table~\ref{tab:catalog}) that only features object detection tasks with 10 versions of the YOLOv4 model.i

Figure~\ref{fig:allocation} shows the fractional allocation decision at each tier of the network topology for different values of $\alpha$ (remember that the smaller $\alpha$ the more importance is given to the latency rather than to inaccuracy, see Eq.~\eqref{eq:serving_cost2}). Models are ordered horizontally by increasing accuracy with 3 potential replicas for each model, and only the models able to serve the most popular request are shown.
%and adjacent ticks refer to multiple replicas of the same model. 
Note that the tier-0 node acts as a repository and its allocation is fixed; moreover, in Fig~\ref{fig:allocation} the repository node picks the second most accurate model because it provides the smallest combined cost in Eq.~\eqref{eq:serving_cost2}.
%Tier 0 acts as repository and its allocation is fixed with just replicas of the highest quality model.
% For this set of experiments, we submit a load of 5000 rps.

For $\alpha=3$ (Fig.~\ref{subfig:allocation_30}), \AlgoName{} allocates a considerable amount of small models (which provide low accuracy) near the edge (tiers 1--3 and model IDs 0--18), as they can serve a larger number of requests compared to higher quality models with low inference delay. %Note that, in general, fewer models are deployed at the very edge (tier 4) because resources are more scarce. 
By giving more importance to the accuracy ($\alpha=4$) the system tends to deploy more accurate models and rarely allocates small models (Fig.~\ref{subfig:allocation_40}). 
%For $\alpha=5$, the system allocates no small models in practice (model IDs 0--20) and is forced to allocate multiple replicas of the most accurate models even on the lower tier nodes (Fig.~\ref{subfig:allocation_50}).
For $\alpha=5$, the number of models deployed on lower tiers decreases, as the system allocates no small models in practice (model IDs 0--20) and selects instead multiple replicas of the most accurate models (Fig.~\ref{subfig:allocation_50}). Since higher quality models feature, in general, a lower serving capacity (Table~\ref{tab:catalog}), Fig.~\ref{subfig:allocation_50} suggests that a significant number of requests is served in the cloud (Tier 0) for this value of $\alpha$.
% The few models allocated on tiers 1--3 suggest that a significant amount of requests is served in the cloud, being latency less important.

\begin{figure}[t]
    \centering
    %\subcaptionbox{$\alpha = 0.5$}{\includegraphics[width = .32\linewidth]{sections/figs/levels_models_0.5.pdf}}
    %\subcaptionbox{\label{subfig:allocation_01}$\alpha = 0.1$}{\includegraphics[width = .325\linewidth]{sections/figs/mar12/levels_models_0.1.pdf}}
    %\subcaptionbox{\label{subfig:allocation_05}$\alpha = 0.5$}{\includegraphics[width = .325\linewidth]{sections/figs/mar12/levels_models_0.5.pdf}}
    \subcaptionbox{\label{subfig:allocation_30}$\alpha = 3$}{
    \includegraphics[width = \linewidth]{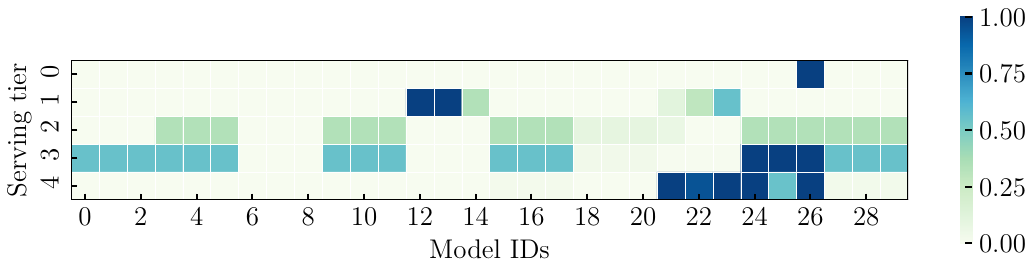}
    }
   \subcaptionbox{\label{subfig:allocation_40}$\alpha = 4$}{
    \includegraphics[width = \linewidth]{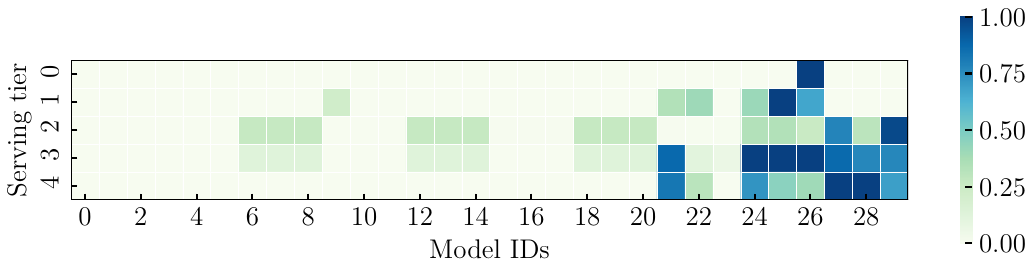}
    }
    \subcaptionbox{\label{subfig:allocation_50}$\alpha = 5$}{
    \includegraphics[width = \linewidth]{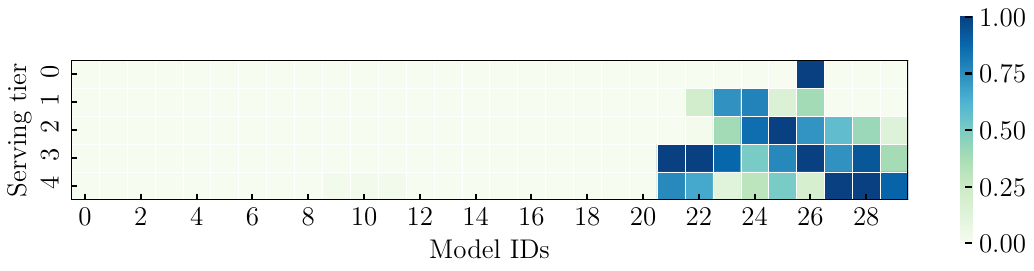}
    }
    %   \subcaptionbox{\label{subfig:trade_off}}{\includegraphics[width=.3\linewidth]{sections/figs/alpha_latency_accuracy.pdf}}
    \caption{Fractional allocation decisions $\allocFrac^v_m$ of \AlgoName{} on the various tiers of \gc{ \emph{Network Topology~I}} under \emph{Fixed Popularity Profile}. We only show the allocations corresponding to the models capable to serve the most popular request.
    %The decision variables are restricted to the subset of models serving the most popular task. 
    The model IDs are sorted by increasing accuracy. 
    }
    \label{fig:allocation}
\end{figure}

\begin{figure}[t]
    \centering
    % \subcaptionbox{\label{subfig:trade_off}}{
    \includegraphics[width=.6\linewidth]{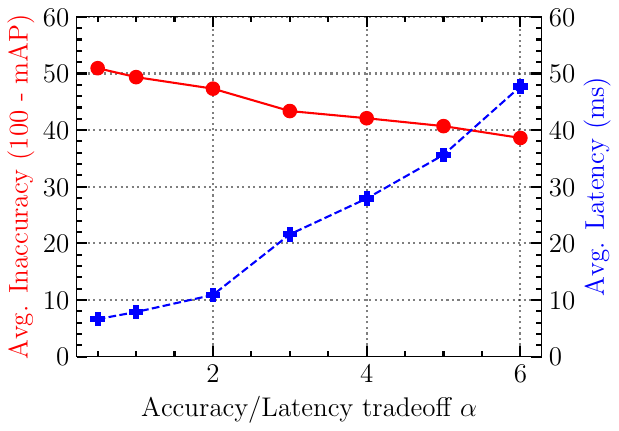}
    % }
    
    \caption{Average latency (dashed line) and inaccuracy (solid line) costs experienced with \AlgoName{} for different values of $\alpha$ under  \emph{Network Topology~I} and \emph{Fixed Popularity Profile}.}% learning rate 0.00001
    
    \label{fig:gain_alpha}
\end{figure}

\begin{figure}[t]
    \centering
\subcaptionbox*{\vspace{-1em}}{\includegraphics[trim=.5 175px .5 .5,clip,width = .49\linewidth]{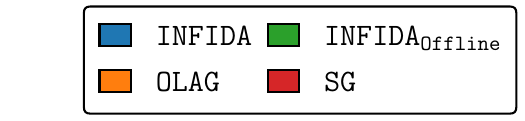}}\\
    \subcaptionbox{\label{subfig:gain_alpha21} }{\includegraphics[trim=.5 .5 .5 70px,clip,width = .49\linewidth]{sections/figs/Jul_big_changing_greedy_vs_oma.pdf}}
    \subcaptionbox{\label{subfig:gain_alpha22} }{\includegraphics[trim=.5 .5 .5 70px,clip,width = .49\linewidth]{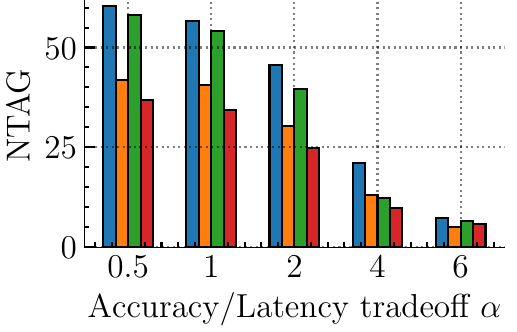}}
    
    \caption{NTAG of the different policies under \emph{Sliding Popularity Profile} and network topologies: (a) \emph{Network Topology~I}, and (b) \emph{Network Topology~II}. }
    
    \label{fig:gain_alpha2}
\end{figure}

\begin{figure}[t]
    \centering
 \subcaptionbox*{\vspace{-1em}}{\includegraphics[trim=.5 174px .5 .5,clip,width = .6\linewidth]{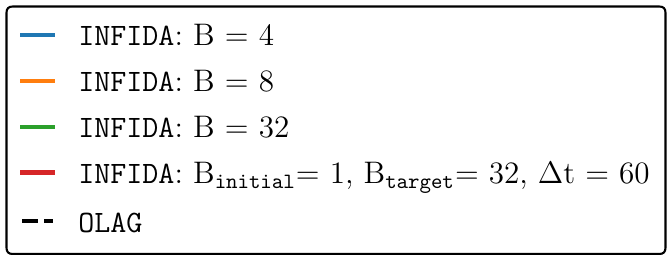}}\\
    \subcaptionbox{\label{subfig:models_update} Models Updates (MU)}{\includegraphics[trim=.5 .5 .5 .5, clip,width = .49\linewidth]{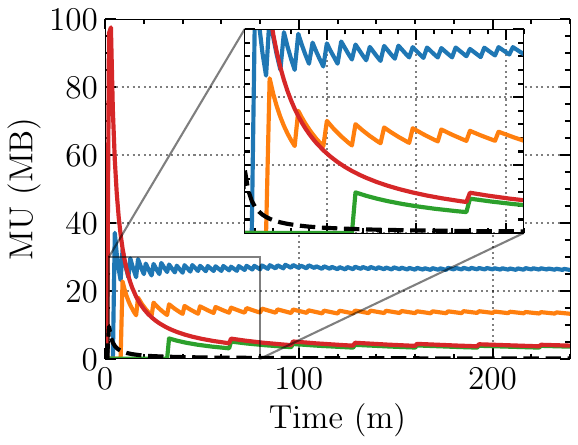}}
    \subcaptionbox{\label{subfig:NTAG} NTAG}{\includegraphics[trim=.5 .5 .5 .5,clip, width = .49\linewidth]{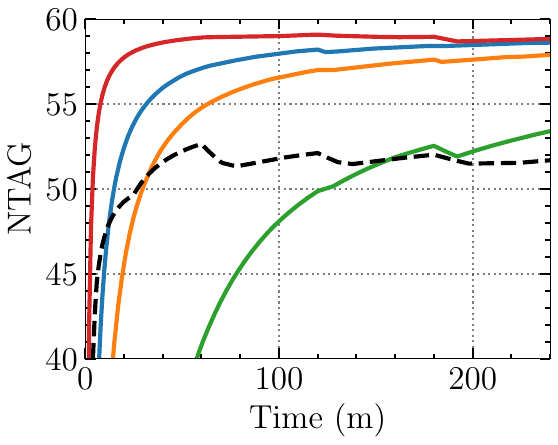}}
    
    \caption{(a) Models Updates (MU) and (b) NTAG of \OnlineGreedy{} and \AlgoName{} for different values of refresh period $B \in \{4,8,16\}$, and for a dynamic refresh period with initial value $B_\text{init} = 1$, target value $B_\text{target} = 32$ and stretching duration $\Delta t = 60$ (1H). The experiment is run under \emph{Network Topology~I} and \emph{Sliding Popularity Profile}.}
    
    \label{fig:update_cost}
\end{figure}

% \rev{ts}{\rev{ts}{For better visualizing the}{To better visualize the} trade-off, we describe below a set of results obtained in a topology with smaller nodes (up to 5 times), which allowed us to stress the system even with a moderate load of 5000 requests per second (Fig~\ref{fig:gain_alpha}).}{}

Figure~\ref{fig:gain_alpha} shows the average experienced inaccuracy (inaccuracy is given by 100 $-$ mAP and mAP is the mean average precision) and latency for different values of $\alpha$ under  \emph{Network Topology~I} and \emph{Fixed Popularity Profile}. When accuracy is not important (i.e., $\alpha \approx 0$), \AlgoName{} effectively achieves very low end-to-end latency (few milliseconds) by prioritizing the deployment of small and inaccurate models near to the edge nodes. Noticeably, the trend in both curves (decreasing inaccuracy and increasing latency) suggests that, when higher accuracy is required, the system starts to prefer models deployed close to the cloud, leading to a sudden change in the trade-off and to a significant increase in latency.

In Fig.~\ref{fig:gain_alpha2} we show the normalized time-averaged gain of~\AlgoName{} compared to \OnlineGreedy{}, \SG{}, and \AlgoNameOffline{} for different values of $\alpha$ under the \emph{Sliding Popularity Profile}. Results are shown both for  \emph{Network Topology~I} (Fig.~\ref{subfig:gain_alpha21}) and for  \emph{Network Topology~II} (Fig.~\ref{subfig:gain_alpha22}).
%(\rev{ts}{}{In Fig.~\ref{fig:gain_alpha2}a the experiments are run on  \emph{Network Topology~I}, and in Fig.~\ref{fig:gain_alpha2}b we run over  \emph{Network Topology~II} where we are also able to evaluate SG offline.}

% The gain is normalized to the number of requests per second. 
The plot shows that the gain decreases by increasing $\alpha$. This is expected since the gain~\eqref{eq:gain} is defined as the improvement w.r.t. the repository allocation (tier 0). Therefore, when the latency is not important, high accuracy models at tier 0 are preferred, and there is no much room for improvement (the optimal gain eventually tends to zero for $\alpha {\to} {+}\infty$).
% Noticeably, in this situation our algorithm performs much better than \OnlineGreedy{} and \SG{} (the gain is up to 2.5 times higher). This suggests that, when there are few remaining allocations that may introduce a cost reduction compared to the repository, \AlgoName{} does a better job in spotting them.
\rev{ts}{}{Note that, in general, \SG{} and \AlgoNameOffline{} policies perform worse than their offline counterparts, as they pick a single allocation that is the best w.r.t. the whole sequence of requests. However, in the \emph{Sliding Popularity Profile} (Fig.~\ref{fig:popularity}) the best decision changes periodically, and only the online policies have the ability to adapt to such change.} Moreover, we observe that consistently \AlgoNameOffline{} has better performance than \SG{}: although both policies are offline, \AlgoNameOffline{} manages to provide a better allocation.
%Greedy struggles when few intermediate allocations can introduce small improvements compared the cloud.

% \todoi{
% Figure 2: only put the proper repository allocation for each value of alpha.
% }
\subsection{Trade-off between model updates and service cost.}

% \todoi{aa: The fact that at every time-slot (how large is a time-slot?) the models saved (after rounding) are completely independent from the previous round can degenerate in a crazy distributed systems where models flip on and off at every time-slot, imposing massive migrations that would dominate bandwidth utilization or at least imposing continuous transfer from storage to main memory, which would dominate computation time. Shall we discuss this point?}

In this set of experiments, we evaluate how the frequency at which \AlgoName{} updates the model allocation affects the update cost incurred by the system. Indeed, frequent updates could lead to massive migrations with an overhead on network bandwidth. As an evaluation metric, we measure the total size of fetched models averaged over time (see the performance metric in Eq.~\eqref{eq:update_cost_def}). We introduce $B$ that we call the \emph{refresh period}, and we restrict \AlgoName{} to only sample a physical allocation every $B \in \{4,8, 32\}$ time slots (line 8 in Algorithm~\ref{algo:idn}). Additionally, we experiment linear stretching of the refresh period $B$ with initial period $B_\text{init} = 1$ and target period $B_\text{target} = 32$ in a stretching duration of $\Delta$t = 1H. We run this experiment under  \emph{Network Topology~I} and \emph{Sliding Popularity Profile}. We set the trade-off parameter $\alpha = 1$.

In particular, Figure~\ref{subfig:models_update} shows the update cost (MU) for different refresh periods, while Figure~\ref{subfig:NTAG} shows the NTAG. Both plots include the performance of the \OnlineGreedy{} heuristic. We observe that, by increasing the refresh period $B$, the system fetches a smaller number of models, and therefore the update cost decreases, at the expense of reactivity. \newton{This tradeoff was characterized formally in \cite{ascent}, wherein the regret is sublinear for $B = \Theta \left(T^\beta \right)$ for $\beta \in [0, 1)$, and the update costs are sublinear for $\beta \in (0,1)$.} Nevertheless, even for large values of $B$ \AlgoName{} eventually exceeds \OnlineGreedy{} in performance: this result is expected since the algorithm continues to learn on the fractional (virtual) states and only the physical allocations are delayed and eventually catch-up for a large time horizon. On the other hand, we observe that \OnlineGreedy{} is relatively conservative in updating its allocation, as it quickly picks a sub-optimal allocation and rarely updates it.

The previous observation motivates the use of a dynamic refresh period. By refreshing more frequently at the start we allow the physical allocation generated by \AlgoName{} to catch up quickly with the fractional states as shown in Fig.~\ref{subfig:NTAG}: a dynamic refresh period that stretches from $B_\text{init} = 1$ to $B_\text{target} = 32$ attains much faster and more precise convergence. This is achieved at the expense of a high update cost at the start, which is, however, quickly dampened until it matches the same update cost of fixing the refresh period to $B_{\text{target}}$.

\subsection{Scalability on Requests Load}

\begin{figure}[t]
    \centering
    \subcaptionbox*{\vspace{-1.4em}}{\includegraphics[trim=.5 175px .5 .5,clip, width = .6\linewidth]{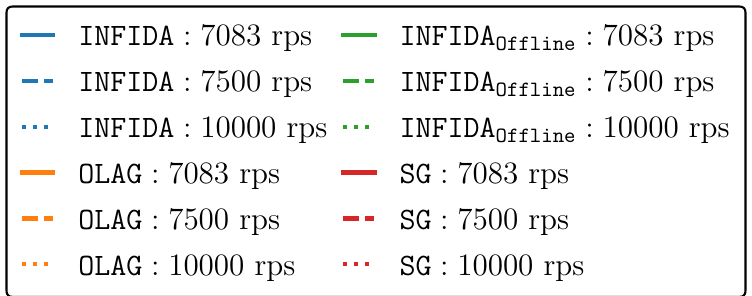}}\\
    \subcaptionbox{\label{subfig:requests_fixed}\emph{Fixed Popularity Profile}}{\includegraphics[width = .45\linewidth]{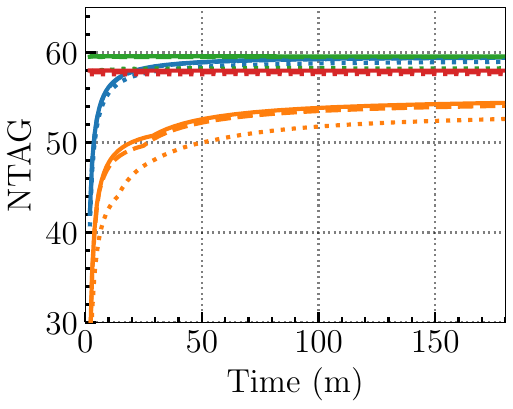}}
    \subcaptionbox{\label{subfig:requests_sliding}\emph{Sliding Popularity Profile}}{\includegraphics[width = .45\linewidth]{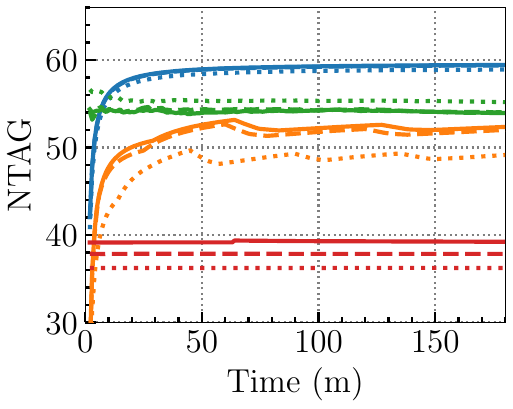} }
    
    \caption{NTAG of the different policies for different request rates \rev{ts}{}{under  \emph{Network Topology~I}.}}
    \label{fig:requests}
\end{figure}

We show how the system performs under different requests loads. For this set of experiments, we set $\alpha=1$.
\gc{Figure~\ref{fig:requests} compares the results for the different allocation policies.}

\gc{We notice that, being \AlgoNameOffline{} and \SG{} offline policies, they perform well when the popularity profile is static (Fig.~\ref{fig:requests}a), but deteriorate under \emph{Sliding Popularity Profile}. Notably, the performance degradation of \AlgoNameOffline{} ($\approx$8\%) is considerably limited compared to \SG{} ($\approx$30\%), which even gets worse when increasing the requests load.}

Figure~\ref{fig:requests} shows that, in general, \AlgoName{} provides a higher gain compared to the \OnlineGreedy{} heuristic. 
\gc{In particular, under \emph{Fixed Popularity Profile} \AlgoName{} manages to converge to the same NTAG provided by its offline counterpart (Fig.~\ref{fig:requests}a), which is $\approx$10\% better than the one provided by \OnlineGreedy{} when the load is 7,083 rps.}
Additionally, \OnlineGreedy{}'s performance deteriorates when the requests load increases from 7,083 rps to 10,000 rps. It is also noteworthy that \OnlineGreedy{} visibly suffers from perturbed performance when the popularity of the tasks changes over time (Fig.~\ref{subfig:requests_sliding}). On the other hand, results show the robustness of \AlgoName{} against changing request loads and popularity: the algorithm preserves its performance in terms of normalized time-averaged gain for the analyzed request loads and under both \emph{Fixed Popularity Profile} and \emph{Sliding Popularity Profile}\gc{, always converging to the highest NTAG}.

% to be changed
%  Last, in Fig.~\ref{fig:accuracy_latency} we evaluate separately the Time-Averaged Latency and Inaccuracy incurred by the different policies. We observe that 
% \AlgoName{} consistently provides the best inaccuracy and latency under both high request load (9,000 rps) and default request load (7,500 rps). Note that \OnlineGreedy{} is initialized with an empty allocation; therefore, at the start most of the requests are being forwarded to the cloud yielding high latency costs and low inaccuracy costs. Whereas \AlgoName{} starts with a random allocation that already gives a significant reduction in latency compared to an empty allocation.
 
Last, in Fig.~\ref{fig:accuracy_latency} we evaluate separately the average latency and inaccuracy attained by the different policies using different values of $\alpha \in \{0.5, 1, 2, 3, 4, 5, 6\}$ under \emph{Fixed Popularity Profile} and \emph{Network Topology~II}.
We observe that \AlgoName{} and its offline counterpart \AlgoNameOffline{} consistently provide the lowest average inaccuracy and latency under both high request load (10,000 rps) and default request load (7,500 rps). \AlgoNameOffline{} is run with hindsight and serves as a lower bound on the achievable latency and inaccuracy under a fixed popularity request process (as in Fig.~\ref{subfig:requests_fixed}).

% \begin{figure}[t!]
%     \centering
%     \subcaptionbox{\emph{Fixed Popularity Profile}}{\includegraphics[width = .45\linewidth]{sections/figs/fixed_greedy_vs_oma.png}}
%     \subcaptionbox{\emph{Sliding Popularity Profile}}{\includegraphics[width = .45\linewidth]{sections/figs/changing_greedy_vs_oma.png}}
    
%     \caption{Robustness of \OMA{} against changing popularity requests \rev{aa}{}{in terms of time-averaged Gain (TAG)}. }
%     \label{fig:my_label}
% \end{figure}

% \begin{figure}[t!]
%     \centering
%     \subcaptionbox{\emph{Fixed Popularity Profile}}{\includegraphics[width = .45\linewidth]{sections/figs/fixed_greedy_vs_oma_temportal.png}}
%     \subcaptionbox{\emph{Sliding Popularity Profile}}{\includegraphics[width = .45\linewidth]{sections/figs/changing_greedy_vs_oma_temportal.png}}
    
%     \caption{time-averaged Gain \rev{aa}{}{(TAG)} under different $\alpha$\rev{aa}{ for \OMA{} and \OnlineGreedy.}{} }
%     \label{fig:my_label}
% \end{figure}

% \iffalse
% \begin{figure}[t]
%     \centering
%     \subcaptionbox{Keras pre-trained models}{\includegraphics[width = .40\linewidth]{sections/figs/Keras_catalog.png}}
%     \subcaptionbox{Pytorch pre-trained models}{\includegraphics[width = .40\linewidth]{sections/figs/Pytorch_catalog.png}}
%     \caption{Models catalog for classification. }
%     \label{fig:my_label}
% \end{figure}
% \fi

\begin{figure}[t]
    \centering
\includegraphics[width=.7\linewidth]{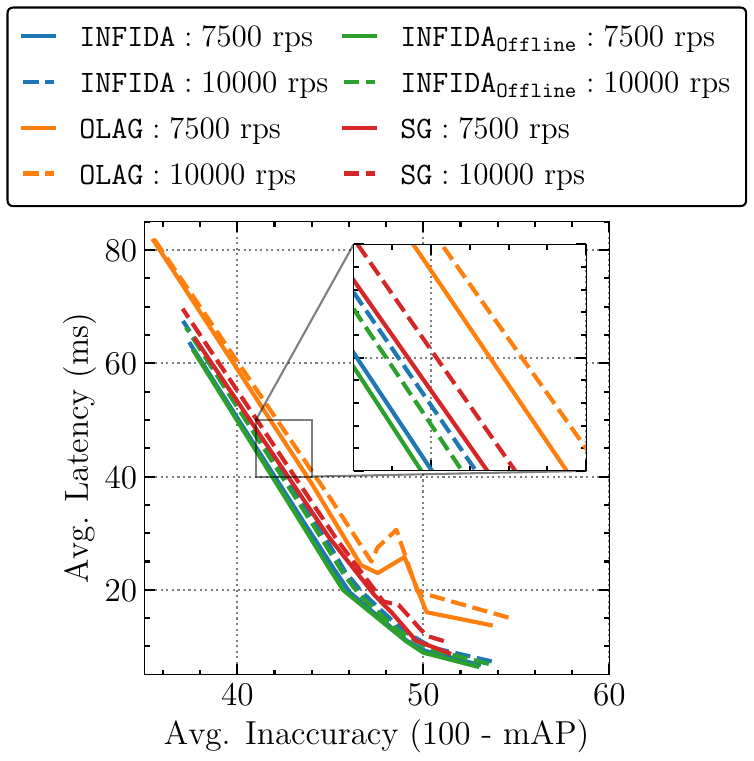}
    
    \caption{Average Latency vs. Average Inaccuracy obtained for different values of $\alpha \in \{0.5, 1, 2, 3, 4, 5, 6\}$ under \emph{Fixed Popularity Profile} and \emph{Network Topology~II}.}
    
    \label{fig:accuracy_latency}
\end{figure}

% \begin{figure}[t]
%     \centering
%  \subcaptionbox*{\vspace{-1em}}{\includegraphics[trim=.5 170px .5 .5,clip,width = .49\linewidth]{sections/figs/Jul_mulobj_latency_legend.pdf}}\\
%     \subcaptionbox{\label{subfig:multiobj_latency}Time-Averaged Latency}{\includegraphics[trim=.5 .5 .5 .5,clip,width = .49\linewidth]{sections/figs/Jul_mulobj_latency.pdf}}
%     \subcaptionbox{\label{subfig:multiobj_innacuracy}Time-Averaged Inaccuracy}{\includegraphics[trim=.5 .5 .5 0.5,clip, width = .49\linewidth]{sections/figs/Jul_mulobj_inaccuracy.pdf}}
    
%     \caption{Time-Averaged Latency~(a) and Inaccuracy~(b) for different request rates 7,500 rps and 9,000 rps under \emph{Network Topology~I} and \emph{Sliding Popularity Profile}. \textcolor{red}{[Choose between this and 10}}
    
%     \label{fig:accuracy_latency}
% \end{figure}

% \subsection{New Experiments}

%!TEX root = ../paper.tex
\section{Conclusions}
\label{sec:conclusions}

In this paper, we introduced the idea of inference delivery networks (IDNs), networks of computing nodes that coordinate to satisfy inference requests in the continuum between Edge and Cloud. IDN nodes can serve inference requests with different levels of accuracy and end-to-end latency, based on their geographic location and processing capabilities. We formalized the NP-hard problem of allocating ML models on IDN nodes, capturing the trade-off between latency and accuracy.
%
%To optimize the ML model allocation in IDNs, 
We proposed \AlgoName{}, a dynamic ML model allocation algorithm that operates in a distributed fashion and provides strong guarantees in an adversarial setting.
We evaluated \AlgoName{} simulating the realistic scenario of an ISP network, and compared its performance \gc{under two different topologies with both an offline greedy heuristic and its online variant.}
%policy based on the greedy heuristic, known to perform well in practice for this family of problems. 
Our results show that \AlgoName{} adapts to different latency/accuracy trade-offs and scales well with the number of requests, outperforming the greedy policies in all the analyzed settings.
\section{Acknowledgement}
This work has been carried out in the framework of a common lab agreement between Inria and Nokia Bell Labs. This research was supported in part by the French Government through the ``Plan de Relance'' and ``Programme d'investissements d`avenir'' and by Inria under the exploratory action MAMMALS.

\bibliographystyle{IEEEtran}
\bibliography{paper}

\begin{thebibliography}{10}
\providecommand{\url}[1]{#1}
\csname url@samestyle\endcsname
\providecommand{\newblock}{\relax}
\providecommand{\bibinfo}[2]{#2}
\providecommand{\BIBentrySTDinterwordspacing}{\spaceskip=0pt\relax}
\providecommand{\BIBentryALTinterwordstretchfactor}{4}
\providecommand{\BIBentryALTinterwordspacing}{\spaceskip=\fontdimen2\font plus
\BIBentryALTinterwordstretchfactor\fontdimen3\font minus
  \fontdimen4\font\relax}
\providecommand{\BIBforeignlanguage}[2]{{%
\expandafter\ifx\csname l@#1\endcsname\relax
\typeout{** WARNING: IEEEtran.bst: No hyphenation pattern has been}%
\typeout{** loaded for the language `#1'. Using the pattern for}%
\typeout{** the default language instead.}%
\else
\language=\csname l@#1\endcsname
\fi
#2}}
\providecommand{\BIBdecl}{\relax}
\BIBdecl

\bibitem{stoica2017berkeley}
\BIBentryALTinterwordspacing
I.~Stoica, D.~Song, R.~A. Popa \emph{et~al.}, ``{A Berkeley View of Systems
  Challenges for AI},'' EECS Department, University of California, Berkeley,
  Tech. Rep. UCB/EECS-2017-159, Oct 2017. [Online]. Available:
  \url{https://www2.eecs.berkeley.edu/Pubs/TechRpts/2017/EECS-2017-159.html}
\BIBentrySTDinterwordspacing

\bibitem{simsek16}
M.~{Simsek}, A.~{Aijaz}, M.~{Dohler} \emph{et~al.}, ``{5G-Enabled Tactile
  Internet},'' \emph{IEEE Journal on Selected Areas in Communications},
  vol.~34, no.~3, pp. 460--473, 2016.

\bibitem{howard17mobilenets}
A.~G. Howard, M.~Zhu, B.~Chen \emph{et~al.}, ``{MobileNets: Efficient
  Convolutional Neural Networks for Mobile Vision Applications},''
  \emph{preprint arXiv:1704.04861}, 2017.

\bibitem{deng20}
L.~{Deng}, G.~{Li}, S.~{Han} \emph{et~al.}, ``{Model Compression and Hardware
  Acceleration for Neural Networks: A Comprehensive Survey},''
  \emph{Proceedings of the IEEE}, 2020.

\bibitem{hazan2016introduction}
E.~Hazan \emph{et~al.}, ``{Introduction to Online Convex Optimization},''
  \emph{Foundations and Trends{\textregistered} in Optimization}, vol.~2, no.
  3-4, pp. 157--325, 2016.

\bibitem{babu2013massively}
S.~Babu and H.~Herodotou, ``{Massively Parallel Databases and MapReduce
  Systems},'' \emph{Foundations and Trends{\textregistered} in Databases},
  2013.

\bibitem{zaharia2012resilient}
M.~Zaharia, M.~Chowdhury, T.~Das \emph{et~al.}, ``{Resilient Distributed
  Datasets: A Fault-Tolerant Abstraction for In-Memory Cluster Computing},'' in
  \emph{9th {USENIX} Symposium on Networked Systems Design and Implementation
  (NSDI 12)}, 2012, pp. 15--28.

\bibitem{fed1}
P.~Kairouz, H.~B. McMahan, B.~Avent \emph{et~al.}, ``{Advances and Open
  Problems in Federated Learning},'' \emph{arXiv preprint arXiv:1912.04977},
  2021.

\bibitem{fed2}
T.~Li, A.~K. Sahu, A.~Talwalkar \emph{et~al.}, ``{Federated Learning:
  Challenges, Methods, and Future Directions},'' \emph{IEEE Signal Processing
  Magazine}, vol.~37, no.~3, p. 50–60, May 2020.

\bibitem{konevcny2016federateda}
J.~Kone{\v{c}}n{\`y}, H.~B. McMahan, D.~Ramage \emph{et~al.}, ``{Federated
  Optimization: Distributed Machine Learning for On-Device Intelligence},''
  \emph{arXiv preprint arXiv:1610.02527}, 2016.

\bibitem{konevcny2016federatedb}
J.~Kone{\v{c}}n{\`y}, H.~B. McMahan, F.~X. Yu \emph{et~al.}, ``{Federated
  Learning: Strategies for Improving Communication Efficiency},'' \emph{arXiv
  preprint arXiv:1610.05492}, 2016.

\bibitem{wu2020accelerating}
W.~Wu, L.~He, W.~Lin \emph{et~al.}, ``{Accelerating Federated Learning over
  Reliability-Agnostic Clients in Mobile Edge Computing Systems},'' \emph{IEEE
  Transactions on Parallel and Distributed Systems}, 2020.

\bibitem{neglia2019role}
G.~Neglia, G.~Calbi, D.~Towsley \emph{et~al.}, ``{The Role of Network Topology
  for Distributed Machine Learning},'' in \emph{IEEE Conference on Computer
  Communications (INFOCOM)}.\hskip 1em plus 0.5em minus 0.4em\relax IEEE, 2019,
  pp. 2350--2358.

\bibitem{olston2017tensorflow}
C.~Olston, N.~Fiedel, K.~Gorovoy \emph{et~al.}, ``{TensorFlow-Serving:
  Flexible, High-Performance ML Serving},'' \emph{preprint arXiv:1712.06139},
  2017.

\bibitem{chappell2015introducing}
D.~Chappell, ``{Introducing Azure Machine Learning},'' \emph{A guide for
  technical professionals, sponsored by microsoft corporation}, 2015.

\bibitem{enginegoogle}
``{Vertex AI $\vert$} google cloud.''

\bibitem{crankshaw2017clipper}
D.~Crankshaw, X.~Wang, G.~Zhou \emph{et~al.}, ``{Clipper: A Low-Latency Online
  Prediction Serving System},'' in \emph{14th {USENIX} Symposium on Networked
  Systems Design and Implementation (NSDI 17)}, 2017, pp. 613--627.

\bibitem{mao2019learning}
H.~Mao, M.~Schwarzkopf, S.~B. Venkatakrishnan \emph{et~al.}, ``{Learning
  Scheduling Algorithms for Data Processing Clusters},'' in \emph{Proceedings
  of the ACM Special Interest Group on Data Communication}, 2019, pp. 270--288.

\bibitem{romero2019infaas}
F.~Romero, Q.~Li, N.~J. Yadwadkar \emph{et~al.}, ``{INFaaS: Managed and
  Model-less Inference Serving},'' \emph{preprint arXiv:1905.13348}, 2019.

\bibitem{crankshaw2020inferline}
D.~Crankshaw, G.-E. Sela, X.~Mo \emph{et~al.}, ``{InferLine: Latency-aware
  Provisioning and Scaling for Prediction Serving Pipelines},'' in
  \emph{Proceedings of the 11th ACM Symposium on Cloud Computing}, 2020, pp.
  477--491.

\bibitem{meng2016mllib}
X.~Meng, J.~Bradley, B.~Yavuz \emph{et~al.}, ``{MLlib: Machine Learning in
  Apache Spark},'' \emph{The Journal of Machine Learning Research}, vol.~17,
  no.~1, pp. 1235--1241, 2016.

\bibitem{jia2014caffe}
Y.~Jia, E.~Shelhamer, J.~Donahue \emph{et~al.}, ``{Caffe: Convolutional
  Architecture for Fast Feature Embedding},'' in \emph{Proceedings of the 22nd
  ACM international conference on Multimedia}, 2014, pp. 675--678.

\bibitem{pedregosa2011scikit}
F.~Pedregosa, G.~Varoquaux, A.~Gramfort \emph{et~al.}, ``{Scikit-learn: Machine
  Learning in Python},'' \emph{the Journal of machine Learning research},
  vol.~12, pp. 2825--2830, 2011.

\bibitem{deepX}
N.~D. {Lane}, S.~{Bhattacharya}, P.~{Georgiev} \emph{et~al.}, ``{DeepX: A
  Software Accelerator for Low-Power Deep Learning Inference on Mobile
  Devices},'' in \emph{2016 15th ACM/IEEE International Conference on
  Information Processing in Sensor Networks (IPSN)}, 2016.

\bibitem{Fernando2019}
N.~{Fernando}, S.~W. {Loke}, and W.~{Rahayu}, ``{Computing with Nearby Mobile
  Devices: A Work Sharing Algorithm for Mobile Edge-Clouds},'' \emph{IEEE
  Transactions on Cloud Computing}, vol.~7, no.~2, pp. 329--343, 2019.

\bibitem{kang2017neurosurgeon}
Y.~Kang, J.~Hauswald, C.~Gao \emph{et~al.}, ``{Neurosurgeon: Collaborative
  Intelligence Between the Cloud and Mobile Edge},'' \emph{ACM SIGARCH Computer
  Architecture News}, 2017.

\bibitem{Branchynet}
S.~{Teerapittayanon}, B.~{McDanel}, and H.~T. {Kung}, ``{BranchyNet: Fast
  Inference via Early Exiting from Deep Neural Networks},'' in \emph{2016 23rd
  International Conference on Pattern Recognition (ICPR)}, 2016.

\bibitem{teerapittayanon2017distributed}
S.~Teerapittayanon, B.~McDanel, and H.-T. Kung, ``{Distributed Deep Neural
  Networks Over the Cloud, the Edge and End Devices},'' in \emph{2017 IEEE 37th
  International Conference on Distributed Computing Systems (ICDCS)}, 2017.

\bibitem{Deng2020}
S.~Deng, H.~Zhao, J.~Yin \emph{et~al.}, ``{Edge Intelligence: The Confluence of
  Edge Computing and Artificial Intelligence},'' \emph{IEEE Internet of Things
  Journal}, 2020.

\bibitem{Ogden2018}
S.~S. Ogden and T.~Guo, ``{MODI: Mobile Deep Inference Made Efficient by Edge
  Computing},'' in \emph{{USENIX} Workshop on Hot Topics in Edge Computing
  (HotEdge 18)}, 2018.

\bibitem{LaiJiao2020}
Y.~Jin, L.~Jiao, Z.~Qian \emph{et~al.}, ``{Provisioning Edge Inference as a
  Service via Online Learning},'' in \emph{2020 17th Annual IEEE International
  Conference on Sensing, Communication, and Networking (SECON)}, 2020.

\bibitem{hung2018videoedge}
C.-C. Hung, G.~Ananthanarayanan, P.~Bodik \emph{et~al.}, ``{VideoEdge:
  Processing Camera Streams using Hierarchical Clusters},'' in \emph{IEEE/ACM
  Symposium on Edge Computing}, 2018.

\bibitem{Qiu2015}
X.~Qiu, H.~Li, C.~Wu \emph{et~al.}, ``{Cost-Minimizing Dynamic Migration of
  Content Distribution Services into Hybrid Clouds},'' \emph{IEEE Transactions
  on Parallel and Distributed Systems}, 2015.

\bibitem{xu18}
J.~{Xu}, L.~{Chen}, and P.~{Zhou}, ``{Joint Service Caching and Task Offloading
  for Mobile Edge Computing in Dense Networks},'' in \emph{IEEE INFOCOM 2018 -
  IEEE Conference on Computer Communications}, 2018, pp. 207--215.

\bibitem{Leung2017}
S.~Wang, R.~Urgaonkar, T.~He \emph{et~al.}, ``{Dynamic Service Placement for
  Mobile Micro-Clouds with Predicted Future Costs},'' \emph{IEEE Transactions
  on Parallel and Distributed Systems}, 2017.

\bibitem{Ben-Ameur2021}
A.~{Ben Ameur}, A.~Araldo \emph{et~al.}, ``{On the Deployability of Augmented
  Reality Using Embedded Edge Devices},'' in \emph{2021 IEEE 18th Annual
  Consumer Communications \& Networking Conference (CCNC)}, 2021.

\bibitem{JoongheonKim2018}
M.~Choi, J.~Kim, and J.~Moon, ``{Wireless Video Caching and Dynamic Streming
  Under Differentiated Quality Requirements},'' \emph{IEEE Journal on Selected
  Areas in Communications}, vol.~36, 2018.

\bibitem{ye2017quality}
Z.~Ye, F.~De~Pellegrini, R.~El-Azouzi \emph{et~al.}, ``{Quality-Aware DASH
  Video Caching Schemes at Mobile Edge},'' in \emph{2017 29th International
  Teletraffic Congress (ITC 29)}, vol.~1.\hskip 1em plus 0.5em minus
  0.4em\relax IEEE, 2017, pp. 205--213.

\bibitem{zhan2017content}
C.~Zhan and Z.~Wen, ``{Content Cache Placement for Scalable Video in
  Heterogeneous Wireless Network},'' \emph{IEEE Communications Letters},
  vol.~21, no.~12, pp. 2714--2717, 2017.

\bibitem{araldo2016representation}
A.~Araldo, F.~Martignon, and D.~Rossi, ``{Representation Selection Problem:
  Optimizing Video Delivery through Caching},'' in \emph{2016 IFIP Networking
  Conference (IFIP Networking) and Workshops}.\hskip 1em plus 0.5em minus
  0.4em\relax IEEE, 2016, pp. 323--331.

\bibitem{SongGuo2020}
Z.~Qu, B.~Ye, B.~Tang \emph{et~al.}, ``{Cooperative Caching for Multiple
  Bitrate Videos in Small Cell Edges},'' \emph{IEEE Transactions on Mobile
  Computing}, vol.~19, no.~2, pp. 288--299, 2020.

\bibitem{poularakis2014video}
K.~Poularakis, G.~Iosifidis, A.~Argyriou \emph{et~al.}, ``{Video delivery over
  heterogeneous cellular networks: Optimizing cost and performance},'' in
  \emph{IEEE INFOCOM 2014-IEEE Conference on Computer Communications}.\hskip
  1em plus 0.5em minus 0.4em\relax IEEE, 2014, pp. 1078--1086.

\bibitem{poularakis2016caching}
------, ``{Caching and Operator Cooperation Policies for Layered Video Content
  Delivery},'' in \emph{IEEE INFOCOM 2016-The 35th Annual IEEE International
  Conference on Computer Communications}.\hskip 1em plus 0.5em minus
  0.4em\relax IEEE, 2016, pp. 1--9.

\bibitem{ShalevOnlineLearning}
S.~Shalev-Shwartz, ``{Online Learning and Online Convex Optimization},''
  \emph{Found. Trends Mach. Learn.}, vol.~4, no.~2, p. 107–194, Feb. 2012.

\bibitem{mcmahan2017survey}
H.~B. McMahan, ``{A survey of algorithms and analysis for adaptive online
  learning},'' \emph{The Journal of Machine Learning Research}, vol.~18, no.~1,
  pp. 3117--3166, 2017.

\bibitem{stratis}
S.~Ioannidis and E.~Yeh, ``{Adaptive Caching nNetworks with Optimality
  Guarantees},'' \emph{ACM SIGMETRICS Performance Evaluation Review}, vol.~44,
  no.~1, pp. 113--124, 2016.

\bibitem{NEURIPS2021_2387337b}
D.~Paria and A.~Sinha, ``{ \texttt{LeadCache}: Regret-Optimal Caching in
  Networks},'' in \emph{Advances in Neural Information Processing Systems},
  M.~Ranzato, A.~Beygelzimer, Y.~Dauphin \emph{et~al.}, Eds., vol.~34.\hskip
  1em plus 0.5em minus 0.4em\relax Curran Associates, Inc., 2021, pp.
  4435--4447.

\bibitem{shanmugam2013femtocaching}
K.~Shanmugam, N.~Golrezaei, A.~G. Dimakis \emph{et~al.}, ``{FemtoCaching:
  Wireless Content Delivery Through Distributed Caching Helpers},'' \emph{IEEE
  Transactions on Information Theory}, vol.~59, no.~12, pp. 8402--8413, 2013.

\bibitem{budget_additive}
J.~Garg, M.~Hoefer, and K.~Mehlhorn, \emph{{Approximating the Nash Social
  Welfare with Budget-Additive Valuations}}, pp. 2326--2340.

\bibitem{TCS-057}
A.~Mehta, ``{Online Matching and Ad Allocation},'' \emph{Foundations and
  Trends® in Theoretical Computer Science}, vol.~8, no.~4, pp. 265--368, 2013.

\bibitem{sub_o2}
A.~Mehta, A.~Saberi, U.~Vazirani \emph{et~al.}, ``{AdWords and Generalized
  Online Matching},'' \emph{J. ACM}, vol.~54, no.~5, p. 22–es, oct 2007.

\bibitem{sub-m1}
M.~Feldman, N.~Gravin, and B.~Lucier, ``{Combinatorial Walrasian
  Equilibrium},'' \emph{SIAM Journal on Computing}, vol.~45, no.~1, pp. 29--48,
  2016.

\bibitem{sub_m2}
T.~Roughgarden and I.~Talgam-Cohen, ``{Why Prices Need Algorithms},'' in
  \emph{Proceedings of the Sixteenth ACM Conference on Economics and
  Computation}, ser. EC '15.\hskip 1em plus 0.5em minus 0.4em\relax New York,
  NY, USA: Association for Computing Machinery, 2015, p. 19–36.

\bibitem{garetto20infocom}
M.~Garetto, E.~Leonardi, and G.~Neglia, ``{Similarity Caching: Theory and
  Algorithms},'' in \emph{{IEEE INFOCOM 2020-IEEE Conference on Computer
  Communications}}, 2020.

\bibitem{falchi2008metric}
F.~Falchi, C.~Lucchese, S.~Orlando \emph{et~al.}, ``{A Metric Cache for
  Similarity Search},'' in \emph{Proceedings of the 2008 ACM workshop on
  Large-Scale distributed systems for information retrieval}, 2008, pp. 43--50.

\bibitem{pandey2009nearest}
S.~Pandey, A.~Broder, F.~Chierichetti \emph{et~al.}, ``{Nearest-Neighbor
  Caching for Content-Match Applications},'' in \emph{{Proceedings of the 18th
  international conference on World wide web}}, 2009, pp. 441--450.

\bibitem{drolia2017cachier}
U.~Drolia, K.~Guo, J.~Tan \emph{et~al.}, ``{Cachier: Edge-Caching for
  Recognition Applications},'' in \emph{2017 IEEE 37th International Conference
  on Distributed Computing Systems (ICDCS)}.\hskip 1em plus 0.5em minus
  0.4em\relax IEEE, 2017, pp. 276--286.

\bibitem{Sermpezis18}
P.~{Sermpezis}, T.~{Giannakas}, T.~{Spyropoulos} \emph{et~al.}, ``{Soft Cache
  Hits: Improving Performance Through Recommendation and Delivery of Related
  Content},'' \emph{IEEE Journal on Selected Areas in Communications}, 2018.

\bibitem{zhou20}
J.~{Zhou}, O.~{Simeone}, X.~{Zhang} \emph{et~al.}, ``{Adaptive Offline and
  Online Similarity-Based Caching},'' \emph{IEEE Networking Letters}, vol.~2,
  no.~4, pp. 175--179, 2020.

\bibitem{garetto21arxiv}
M.~Garetto, E.~Leonardi, and G.~Neglia, ``{Content Placement in Networks of
  Similarity Caches},'' \emph{Computer Networks}, vol. 201, p. 108570, 2021.

\bibitem{blalock2020state}
D.~Blalock, J.~J.~G. Ortiz, J.~Frankle \emph{et~al.}, ``{What is the State of
  Neural Network Pruning?}'' \emph{arXiv preprint arXiv:2003.03033}, 2020.

\bibitem{Hinton2015}
G.~Hinton \emph{et~al.}, ``{Distilling the Knowledge in a Neural Network},''
  \emph{preprint arXiv:1503.02531}, 2015.

\bibitem{Ravi2018}
S.~Ravi, ``{Custom On-Device ML Models with Learn2Compress},'' 2018.

\bibitem{Mengshoel2017}
Z.~Fang, T.~Yu, O.~J. Mengshoel \emph{et~al.}, ``{QoS-Aware Scheduling of
  Heterogeneous Servers for Inference in Deep Neural Networks},'' in \emph{ACM
  Conference on Information and Knowledge Management (CIKM)}, vol. Part F1318,
  2017, pp. 2067--2070.

\bibitem{Brooks2019}
Y.~Wang, G.~Y. Wei, and D.~Brooks, ``{Benchmarking TPU, GPU, and CPU platforms
  for deep learning},'' \emph{arXiv preprint arXiv:1907.10701}, 2019.

\bibitem{ascent}
T.~Si~Salem, G.~Neglia, and D.~Carra, ``{Ascent Similarity Caching With
  Approximate Indexes},'' \emph{IEEE/ACM Transactions on Networking}, vol.~31,
  no.~3, pp. 1173--1186, 2023.

\bibitem{nemhauser1978best}
G.~L. Nemhauser and L.~A. Wolsey, ``{Best Algorithms for Approximating the
  Maximum of a Submodular Set Function},'' \emph{Mathematics of operations
  research}, 1978.

\bibitem{fairstein20201}
Y.~Fairstein, A.~Kulik, J.~S. Naor \emph{et~al.}, ``{A
  (1-1/e-$\varepsilon$)-Approximation for the Monotone Submodular Multiple
  Knapsack Problem},'' in \emph{28th Annual European Symposium on Algorithms
  (ESA 2020)}.\hskip 1em plus 0.5em minus 0.4em\relax Schloss
  Dagstuhl-Leibniz-Zentrum f{\"u}r Informatik, 2020.

\bibitem{bubeck2015convexbook}
S.~Bubeck, ``{Convex Optimization: Algorithms and Complexity},''
  \emph{Foundations and Trends{\textregistered} in Machine Learning}, vol.~8,
  no. 3-4, pp. 231--357, Nov. 2015.

\bibitem{sisalem21icc}
T.~Si~Salem, G.~Neglia, and S.~Ioannidis, ``{No-Regret Caching via Online
  Mirror Descent},'' \emph{ACM Trans. Model. Perform. Eval. Comput. Syst.},
  vol.~8, no.~4, aug 2023.

\bibitem{byrka2014improved}
J.~Byrka, T.~Pensyl, B.~Rybicki \emph{et~al.}, ``{An Improved Approximation for
  $k$-median, and Positive Correlation in Budgeted Optimization - Extended
  Arxiv version},'' in \emph{ACM-SIAM symposium on Discrete algorithms}, 2014.

\bibitem{krause2014submodular}
A.~Krause and D.~Golovin, ``{Submodular Function Maximization},''
  \emph{Tractability}, vol.~3, pp. 71--104, 2014.

\bibitem{paschos19}
G.~S. {Paschos}, A.~{Destounis}, L.~{Vigneri} \emph{et~al.}, ``{Learning to
  Cache With No Regrets},'' in \emph{IEEE INFOCOM}, 2019.

\bibitem{ceselli2017mobile}
A.~Ceselli, M.~Premoli, and S.~Secci, ``{Mobile Edge Cloud Network Design
  Optimization},'' \emph{IEEE/ACM Transactions on Networking}, vol.~25, no.~3,
  pp. 1818--1831, 2017.

\bibitem{bochkovskiy2020yolov4}
A.~Bochkovskiy, C.-Y. Wang, and H.-Y.~M. Liao, ``{YOLOv4: Optimal Speed and
  Accuracy of Object Detection},'' \emph{preprint arXiv:2004.10934}, 2020.

\bibitem{cai2020yolobile}
Y.~Cai \emph{et~al.}, ``{YOLObile: Real-Time Object Detection on Mobile Devices
  via Compression-Compilation Co-Design},'' \emph{preprint arXiv:2009.05697},
  2020.

\bibitem{similaritycachington}
G.~Neglia, M.~Garetto, and E.~Leonardi, ``{Similarity Caching: Theory and
  Algorithms},'' \emph{IEEE/ACM Trans. Netw.}, vol.~30, no.~2, p. 475–486,
  dec 2021.

\bibitem{mordukhovich2017geometric}
B.~S. Mordukhovich and N.~M. Nam, ``{Geometric Approach to Convex
  Subdifferential Calculus},'' \emph{Optimization}, vol.~66, no.~6, pp.
  839--873, 2017.

\bibitem{rockafellar2015convex}
R.~T. Rockafellar, \emph{{Convex Analysis}}.\hskip 1em plus 0.5em minus
  0.4em\relax Princeton University Press, 2015.

\bibitem{shalev2007online}
S.~Shalev-Shwartz, ``{Online Learning: Theory, Algorithms, and Applications},''
  Ph.D. dissertation, The Hebrew University of Jerusalem, 2007.

\bibitem{goemans1994new}
M.~X. Goemans and D.~P. Williamson, ``{New $\frac{3}{4}$-Approximation
  Algorithms for the Maximum Satisfiability Problem},'' \emph{SIAM Journal on
  Discrete Mathematics}, vol.~7, no.~4, pp. 656--666, 1994.

\bibitem{wilf1963some}
H.~S. Wilf, ``{Some Applications of the Inequality of Arithmetic and Geometric
  Means to Polynomial Equations},'' in \emph{Proceedings of the American
  Mathematical Society}, vol.~14, no.~2.\hskip 1em plus 0.5em minus 0.4em\relax
  JSTOR, 1963, pp. 263--265.

\end{thebibliography}

\ifnum\isextended=1

\newpage
\onecolumn
\begin{center}
{\Large  Supplementary Material for Paper:}\\
\vspace{0.4 cm}
{ \huge Towards Inference Delivery Networks: \\Distributing Machine Learning with Optimality Guarantees}
\end{center}

\appendices

%%%%%%%%%%%%%%%%%%%%%%%%%%%%%%%%%%%%%%%%
\section{Submodularity of the Gain Function}
\label{appendix:submodularity}

%%%%%%%%%%%%%%%%%%%%%%%%%%%%%%%%%%%%%%%%
As submodularity is defined for set functions, let us associate to the gain function $\systemGain(\,\cdot\,)$ an opportune set function as follows.
Given a set $S\subseteq\vertices\times\modelSet$ of pairs $(v,m)\in\vertices\times\modelSet$ (nodes and models), we define the corresponding associated vector $\allocVec(S)$ with $\alloc_m^v(S)=1$ if  $((v,m) \in S \lor \repo^v_m = 1) $, $x^v_m(S) = 0$ otherwise.
Now we can define the set function $f_t\colon 2^{\vertices\times\modelSet}\to\reals$:
%Note that finding the allocation vector $\allocVec$ that maximizes the allocation gain in Eq.~\eqref{eq:gain-compact} is equivalent to finding the set $S\subseteq\vertices\times\modelSet$ of pairs $(v,m)\in\vertices\times\modelSet$ (nodes and models) so that the set function $f\colon 2^{\vertices\times\modelSet}\to\reals$, defined as
\begin{align}
    f_t(S) \triangleq \systemGain(\requestBatchVec_t,\loadVec_t, \allocVec(S)).
\label{eq:set-function-gain}
\end{align}

We can also define the set function $F_T: 2^{\vertices \times V} \to \reals$ associated to the time-averaged gain in Eq.~\eqref{eq:static_gain} as 

\begin{align}
    F_T (S) \triangleq \frac{1}{T} \sum^T_{t=1} f_t (S) =  \frac{1}{T} \sum^T_{t=1}  \systemGain(\requestBatchVec_t,\loadVec_t, \allocVec(S)).
    \label{eq:set-function-gainT}
\end{align}
%is maximized, for any $(v,m) \in \vertices \times \modelSet$ if $((v,m) \in S \lor \repo^v_m = 1) $ then $ \alloc^v_m = 1$, otherwise $x^v_m = 0$.

We first start proving that $f_t$ is submodular in the following lemma.

% By construction of $\systemGain$, we have that $g$ is a monotonically non-decreasing set function.

\begin{lemma}
The set function $f_t\colon 2^{\vertices\times\modelSet}\to\reals$ in Eq.~\eqref{eq:set-function-gain} is normalized (i.e., $f_t(\emptyset) = 0$), submodular, and monotone.
\label{lemma:submodular}
\end{lemma}
\begin{proof}\ \\
\noindent\textbf{Normalization.}
The constructed set function is normalized (as in $f_t(\emptyset) = 0$), we have
\begin{align}
    f_t(\emptyset) = \systemGain(\requestBatchVec_t,\loadVec_t, \allocVec(\emptyset)) =\systemGain(\requestBatchVec_t,\loadVec_t, \repoVec)=
     C(\requestBatchVec_t,\loadVec_t, \repoVec) - C(\requestBatchVec_t,\loadVec_t, \repoVec) = 0.
\end{align}

\noindent\textbf{Submodularity.} A function $f_t\colon 2^{\vertices\times\modelSet}\to\reals$ is submodular~\cite{krause2014submodular} if for every $S' \subset S'' \subset \vertices \times \modelSet$ and $(\bar{v}, \bar{m}) \in (\vertices \times \modelSet) \setminus S''$ it holds that 
\begin{equation}
    f_t(S''\cup \{(\bar{v}, \bar{m})\}) - f_t(S'') \leq f_t(S'\cup \{(\bar{v}, \bar{m})\}) - f_t(S'). \nonumber
\end{equation}

% Let us consider $(\bar{v},\bar{m}) \in (\vertices \times \modelSet) \setminus S''$. Let us also consider the associated allocation vectors $\allocVec'$, where elements ${x'}^v_m=1 \iff (v,m)\in S'$, and $\allocVec''$, where elements ${x''}^v_m=1 \iff (v,m)\in S''$. 

Let us consider $(\bar{v},\bar{m}) \in (\vertices \times \modelSet) \setminus S''$. We take $\allocVec'$, $\allocVec''$, $\bar \allocVec'$, and $\bar \allocVec''$ as short hand notation for $\allocVec(S')$, $\allocVec(S'')$, $\allocVec(S' \cup \{(\bar v, \bar m)\} )$, and  $\allocVec(S'' \cup \{(\bar v, \bar m)\})$, respectively.
%We use short hands $\bar{\allocVec}'$ and $\bar{\allocVec}''$ to denote allocation vectors identical to ${\allocVec}'$ and ${\allocVec}'$ except for having $x^v_m=1$, i.e., they both additionally allocate model $m$ on node $v$.

Since $S'\subset S''$, we have that if ${x'}^v_m = 1 \implies {x''}^v_m = 1$, and thus $\auxalloc^{k}_{  \request}(\loadVec_t, \allocVec')\le \auxalloc^{k}_{  \request}(\loadVec_t, \allocVec'')$, for any $k$ (see Eq.~\eqref{eq:several-definitions}). Therefore
\begin{align}
\label{eq:subm-disequality1}
 \Auxalloc^k_\request(\requestBatchVec_t,\loadVec_t, \allocVec') \leq \Auxalloc^k_\request(\requestBatchVec_t,\loadVec_t, \allocVec''), \forall k \in  [\modelsNo_{\request}-1], \forall \request\in \requestSet.
\end{align} 
% Eq.~\eqref{eq:subm-disequality1} implies that
% \begin{align}
% \Auxalloc^k_\request(\requestBatchVec_t,\loadVec_t, \allocVec') \leq \Auxalloc^k_\request(\requestBatchVec_t,\loadVec_t, \allocVec'').
% \label{eq:subm-disequality22}
% \end{align}

Let us denote ${\bar{k}_\rho} \triangleq \kappa_{\request} (\bar{v}, \bar{m})$. Due to Eq.~\eqref{eq:subm-disequality1}, $\forall k \in  [\modelsNo_{\request}-1], \forall \request\in \requestSet$ the following inequality holds:

\begin{align}
\min\{r_\request^t - \Auxalloc^k_\request(\requestBatchVec_t,\loadVec_t, \allocVec''), \auxload^{{\bar{k}_\rho}}_{\request}{(\loadVec_t)}\} & \leq \min\{r_\request^t - \Auxalloc^k_\request(\requestBatchVec_t,\loadVec_t, \allocVec'), \auxload^{{\bar{k}_\rho}}_{\request}{(\loadVec_t)}\},
\end{align}
or equivalently:
\begin{align}
\min\{r_\request^t, \Auxalloc^k_\request(\requestBatchVec_t,\loadVec_t, \allocVec'') + \auxload^{{\bar{k}_\rho}}_{\request}{(\loadVec_t)}\} - \Auxalloc^k_\request(\requestBatchVec_t,\loadVec_t, \allocVec'') & \leq \min\{r_\request^t, \Auxalloc^k_\request(\requestBatchVec_t,\loadVec_t, \allocVec') + \auxload^{{\bar{k}_\rho}}_{\request}{(\loadVec_t)}\} - \Auxalloc^k_\request(\requestBatchVec_t,\loadVec_t, \allocVec').
\label{eq:subm-disequality2}
\end{align}
% We use the shorthand $\bar{\allocVec}$ to denote an allocation vector identical to ${\allocVec}$ except for having $x^{\bar{v}}_{\bar{m}}=1$, i.e., it additionally allocates model $\bar{m}$ on node $\bar{v}$.
Observe that 
\[
    {\auxalloc}^{k}_{\request}(\loadVec_t, \bar{\allocVec}'') 
    =
    \begin{cases}
    {\auxalloc}^{k}_{\request}(\loadVec_t, \allocVec'')  & \text{if }k\neq {\bar{k}_\rho},
    \\
    {\auxalloc}^{k}_{\request}(\loadVec_t, \allocVec'') + \vGroup{x^{\bar{v}}_{\bar{m} } }{=1}\cdot\load^{t,\bar{v}}_{\rho,\bar{m}}
    \stackrel{\text{\eqref{eq:several-definitions}} }{=}
     {\auxalloc}^{k}_{\request}(\loadVec_t, \allocVec'') +
    \auxload^{{\bar{k}_\rho}}_{\request}{(\loadVec_t)}
    & \text{if }k= {\bar{k}_\rho},\vspace{-.5em}
    \end{cases}
\]
and thus
%We have the following equality  
\begin{align}
   \Auxalloc^k_\request(\requestBatchVec_t,\loadVec_t, \bar{\allocVec}'') =  
   \begin{cases}
   \Auxalloc^k_\request(\requestBatchVec_t,\loadVec_t, \allocVec'')\quad \text{if}\,\, k < {\bar{k}_\rho}
   \\
\min\left\{
   \textstyle \requestBatch_\request^t,  \sum^{k}_{k'=1} {\auxalloc}^{k'}_{\request}(\loadVec_t, \allocVec'')  
   +
   \auxload^{{\bar{k}_\rho}}_{\request}{(\loadVec_t)}
   \right\} = 
   \min\{r_\request^t, \Auxalloc^k_\request(\requestBatchVec_t,\loadVec_t, \allocVec'') + \auxload^{{\bar{k}_\rho}}_{\request}{(\loadVec_t)}\} \quad \text{if}\,\, k \geq {\bar{k}_\rho}.
   \end{cases} 
\end{align}
Note that the same equality holds between $\bar \allocVec'$ and $ \allocVec'$ since $(\bar v, \bar m)\not\in S'$.

The marginal gain of adding pair $(\bar{v}, \bar{m})\in\vertices\times\modelSet$ to the allocation set $S''$ is

\begin{align}
    &f_t(S''\cup \{(\bar{v}, \bar{m})\}) - f_t(S'') 
    =
    \systemGain(\requestBatchVec_t,\loadVec_t, \bar{\allocVec}'') - \systemGain(\requestBatchVec_t,\loadVec_t, \allocVec'') \nonumber \\ 
    & \stackrel{\eqref{eq:gain-compact}}{=}
    \sum_{\request\in\requestSet}\Bigg[
        \sum_{k=1}^{{\bar{k}_\rho}-1}\left(\smallCost^{k+1}_{\request} - \smallCost^{k}_{\request}\right)\left(\Auxalloc^k_\request(\requestBatchVec_t,\loadVec_t, \allocVec'') - \Auxalloc^k_\request(\requestBatchVec_t,\loadVec_t, \repoVec)\right) \nonumber \\
        &\qquad\qquad+ \sum_{k={\bar{k}_\rho}}^{\modelsNo_\request-1}\left(\smallCost^{k+1}_{\request} - \smallCost^{k}_{\request}\right)\left(\min\left\{r_\request^t, \Auxalloc^k_\request(\requestBatchVec_t,\loadVec_t, \allocVec'') + \auxload^{{\bar{k}_\rho}}_{\request}{(\loadVec_t)}\right\} - \Auxalloc^k_\request(\requestBatchVec_t,\loadVec_t, \repoVec)\right) \nonumber \\
        &\qquad\qquad- \sum_{k=1}^{\modelsNo_\request-1}\left(\smallCost^{k+1}_{\request} - \smallCost^{k}_{\request}\right)\left(\Auxalloc^k_\request(\requestBatchVec_t,\loadVec_t, \allocVec'') - \Auxalloc^k_\request(\requestBatchVec_t,\loadVec_t, \repoVec)\right)
    \Bigg] \nonumber \\
    &= \sum_{\request\in\requestSet}\sum_{k={\bar{k}_\rho}}^{\modelsNo_\request-1}\left(\smallCost^{k+1}_{\request} - \smallCost^{k}_{\request}\right)\Bigg[\left(\min\left\{r_\request^t, \Auxalloc^k_\request(\requestBatchVec_t,\loadVec_t, \allocVec'') + \auxload^{{\bar{k}_\rho}}_{\request}{(\loadVec_t)}\right\} - \Auxalloc^k_\request(\requestBatchVec_t,\loadVec_t, \repoVec)\right) - \left(\Auxalloc^k_\request(\requestBatchVec_t,\loadVec_t, \allocVec'') - \Auxalloc^k_\request(\requestBatchVec_t,\loadVec_t, \repoVec) \right)\Bigg] \nonumber \\ 
    &= \sum_{\request\in\requestSet}\sum_{k={\bar{k}_\rho}}^{\modelsNo_\request-1}\left(\smallCost^{k+1}_{\request} - \smallCost^{k}_{\request}\right)\left(\min\left\{r_\request^t, \Auxalloc^k_\request(\requestBatchVec_t,\loadVec_t, \allocVec'') + \auxload^{{\bar{k}_\rho}}_{\request}{(\loadVec_t)}\right\} - \Auxalloc^k_\request(\requestBatchVec_t,\loadVec_t, \allocVec'')\right). \label{eq:marginal-gain}
\end{align}

By bounding up each term of the marginal gain \eqref{eq:marginal-gain} as in \eqref{eq:subm-disequality2} we get the following:
\begin{equation}
    f_t(S''\cup \{(\bar{v}, \bar{m})\}) - f_t(S'') \leq f_t(S'\cup \{(\bar{v}, \bar{m})\}) - f_t(S'). \nonumber
\end{equation}

We conclude that $f_t$ is a submodular set function.

\noindent \textbf{Monotonicity.} 
The function $f_t\colon 2^{\vertices\times\modelSet}\to\reals$ is monotone~\cite{krause2014submodular} if for every $S' \subset S'' \subset \vertices \times \modelSet$ it holds that 
\begin{equation}
    f_t(S'') \geq  f_t(S').
\end{equation}
We have 
\begin{align}
    f_t(S'') - f_t(S') &=  \systemGain(\requestBatchVec_t,\loadVec_t, {\allocVec}'') - \systemGain(\requestBatchVec_t,\loadVec_t, \allocVec') \\
    &=  \sum_{\request\in\requestSet}
        \sum_{k=1}^{K_\request-1}\left(\smallCost^{k+1}_{\request} - \smallCost^{k}_{\request}\right)\left(\Auxalloc^k_\request(\requestBatchVec_t,\loadVec_t, \allocVec'') - \Auxalloc^k_\request(\requestBatchVec_t,\loadVec_t, \allocVec') \right) \\
        &\geq 0.
\end{align}
The last inequality is obtained using Eq.~\eqref{eq:subm-disequality1}. We conclude that $f_t$ is  monotone.
\end{proof}
\begin{lemma}
The set function $ F_T : 2^{\vertices\times\modelSet}\to\reals$ in Eq.~\eqref{eq:set-function-gainT} is normalized (i.e., $f_T(\emptyset) = 0$), submodular, and monotone.
\label{lemma:submodularT}
\end{lemma}
\begin{proof}

The set function $F_T$ is a nonnegative linear combination of submodular and monotone functions $f_t$ (see Lemma~\ref{lemma:submodular}), then $F_T$ is also submodular and monotone~\cite{krause2014submodular}. Moreover, each $f_t$ is normalized, then it follows that $F_T$ is normalized.

\end{proof}

\section{NP-hardness}
\label{appendix:np_hard}
% \change{}
\begin{theorem}
The problem of maximizing  the time-averaged allocation gain~\eqref{eq:static_gain} is NP-hard.
\label{theorem:np_hard}
\end{theorem}
\begin{proof}

  We demonstrate the hardness of the problem by a reduction of the similarity caching problem, which is NP-hard~\cite{similaritycachington} (a result that follows a reduction of the dominating set problem). We first define the similarity caching problem and characterize its inputs and variables.

   \noindent\textbf{Similarity Caching Problem.} 
    Consider the network topology comprising of a repository node and a single cache node. The repository node stores a catalog of files $\hat{\mathcal{N}} = \set {1,2,\dots, \hat N}$. 
    A cache node can store a subset $\hat{\mathcal S} \subset \hat{\mathcal{N}}$ of $\hat k$ files  (i.e., $\card{\hat{\mathcal S}} = \hat k$), where $\hat k \in \set{1, 2, \dots, \card{\hat{\mathcal{N}}}-1}$.  
    %A cache node can store a subset of $\hat k \in \set{1, 2, \dots, \card{\hat{\mathcal{N}}}-1}$ files $\hat{\mathcal S} \subset \hat{\mathcal{N}}$ and $\card{\hat{\mathcal S}} = \hat k$.  
    The system incurs an approximation cost $C_a(\hat r,\hat s) \in \mathbb R \cup \set{-\infty, +\infty}$ when a request for file $\hat r\in \hat{\mathcal{N}}$ is satisfied by the cache serving file $\hat s \in \hat{\mathcal{N}}$. The approximation cost is null if $\hat s = \hat r$  (i.e., $C_a(\hat r,\hat r) =0 $ for all $\hat r \in \hat{\mathcal{N}}$).
    %The system incurs an approximation cost $C_a(\hat r,\hat s) \in \mathbb R \cup \set{-\infty, +\infty}$ when a requested file $\hat r\in \hat{\mathcal{N}}$ is approximated with a file $\hat s \in \hat{\mathcal{N}}$  at the level of the cache node; moreover, approximating a file with an identical one incurs no cost $C_a(\hat r,\hat r) =0 $ for all $\hat r \in \hat{\mathcal{N}}$.
    The system also incurs an additional retrieval cost $C_r \in \mathbb R \cup \set{-\infty, +\infty}$ to serve the request with a file stored at the repository node. For a given request $\hat r \in \hat{\mathcal{N}}$ and cache state $\hat{\mathcal{S}}$, the system incurs the following overall cost. 
    \begin{align}
        C_{\mathrm{overall}}(\hat r, \hat{\mathcal{S}}) \triangleq \min \set{C_r, \min\set{C_a (\hat r, \hat s): \hat s \in \hat{\mathcal{S}}}},\label{e:c_overall}
    \end{align}
    where the term $\min\set{C_a (\hat r, \hat s): \hat s \in \hat{\mathcal{S}}}$ signifies that the best approximating file stored at the cache is selected as a candidate to serve the request $\hat r$, and the term $\min \set{C_r, \,\cdot\,}$ signifies that when the cost of approximating the request with the candidate file exceeds the retrieval cost, the system serves the request by fetching an identical file from the repository and incurs a retrieval cost $C_r$. At timeslot $t \in \set{1,2, \dots, T}$, the system receives a request $\hat r_t \in \hat{\mathcal{N}}$. 
    
    The static offline problem is formulated as follows.
    \begin{align}
        \hat{\mathcal{S}}_{\star} \in \argmin_{\hat{\mathcal{S}}\subset \hat{\mathcal{N}}, |\hat{\mathcal{S}}| = \hat k}  \sum^T_{t=1}C_{\mathrm{overall}}(\hat r_t, \hat{\mathcal{S}}).\label{eq:problem}
    \end{align}
    
\noindent \emph{Inputs.} The similarity caching problem~\eqref{eq:problem} takes the following inputs: $\hat{\mathcal{N}} \in \set{1,2,\dots, \hat N
%GN: \card{\hat{\mathcal{N}}}
}, \hat k \in \set{1,2, \dots, \card{\hat{\mathcal{N}}-1}}$, $C_r \in \reals\cup\set{-\infty+\infty}$, $C_a: \hat{\mathcal{N}}  \times \hat{\mathcal{N}} \to  \reals\cup\set{-\infty+\infty}$, $\set{\hat r_1,\hat r_2, \dots, \hat r_T} \in \hat{\mathcal{N}}^T$.

\noindent \emph{Decision variable.} The   similarity caching problem~\eqref{eq:problem} seeks a cache allocation $\hat{\mathcal{S}} \subset \hat{\mathcal{N}}$ such that $\card{\hat{\mathcal{S}}} = \hat k$.

\noindent\textbf{Reduction of similarity caching problem.} 
We show that the similarity caching problem~\eqref{eq:problem} can be reduced to the static model allocation problem in Eq.~(14). We start observing that problem~(14) is equivalent to  
\begin{align}
        \pmb{x}_{\star} \in \argmin_{\pmb{x} \in \mathcal X}  \sum^T_{t=1} C(\pmb r_t, \pmb l_t, \pmb x).~\label{eq:problem2}
    \end{align}
%GN: We consider a sub-problem of optimizing the allocation of ML models in IDNs. 
The set of nodes is $\mathcal{V} = \set{1,2}$, and the set of edges is $\mathcal{E} = \set{(1,2)}$ (a topology with a single repository and a single cache, which we assume to be node~2 and node~1, respectively ). The set of models is $\mathcal{M} = \hat{\mathcal{N}}$ and the set of tasks is $\mathcal{N} =  \hat{\mathcal{N}}$. The  set  of request types is $\mathcal R = \mathcal{N} \times \set{(1,2)}$. The maximum capacity is $L^v_m=1$, and the potential available
capacity is $l^{t,v}_{\rho,m} = 1$ for every node $v \in \mathcal{V}$, timeslot $t \in \set {1,2,\dots, T}$, $\rho \in \mathcal{R}$, and model $m \in \mathcal{M}$. Given a request type $\rho = (i, \pmb p) \in \mathcal{R}$, we define a cost $C^{p_{j}}_{\pmb p, m} \triangleq C_a(i, m)$ for node $p_j =1$ and $C^{p_{j}}_{\pmb p, m}  = C_r$ for node $p_j=2$ for $m \in \mathcal{M}$. 
% The minimal allocation is $\omega^{p}_m = 1$ if  $p =2$, and $\omega^{p}_m = 0$ if $p=1$ for every $m \in \mathcal{M}$. 
The allocation vector at node $1$ is $\pmb x^1 = \left( \mathds{1}_{\set{m \in \hat{\mathcal{S}}}}\right)_{m \in \mathcal{M}}$, and at node $2$ is $\pmb x^2 = \left(1\right)_{m \in \mathcal{M}}$. The model size is $s^v_m =1$ for every $v \in \mathcal{V}$ and $m \in \mathcal{N}$. The allocation budget is $b^1_m = \hat k$, and $b^2_m = \card{\mathcal M}$.
% At timeslot $t$, request batch $\pmb{r}_t$ is composed of just one request, of type $\rho_t=(\hat{r}_t, \mathbf{p}_t)$, i.e., $\pmb r_t = \parentheses{\mathds{1}_{\set{i = \hat r_t}}}_{(i, \pmb p) \in \mathcal{R}}$.
At timeslot $t$, a single request constitutes the request batch $\pmb r_t = \parentheses{\mathds{1}_{\set{i = \hat r_t}}}_{(i, \pmb p) \in \mathcal{R}}$.

The aggregate cost in Eq.~(12) incurred by the system at time slot $t$ is given by
\begin{align}
    C(\pmb r_t, \pmb l_t, \pmb x) &= \sum_{\rho \in \mathcal{R}} \sum^{K_\rho}_{k=1} \gamma^k_{\rho} \min\set{r^t_\rho - \sum^{k-1}_{k'=1} z^{k'}_{\rho} (\pmb l_t, \pmb x), z^{k}_{\rho} (\pmb l_t, \pmb x)} \cdot \mathds{1}_{\set{\sum^{k-1}_{k'=1} z^{k'}_{\rho}(\pmb l_t, \pmb x) < r^t_\rho}}\label{e:np_1} 
    \\
    &=  \sum^{2 \hat N}_{k=1} \gamma^k_{{(\hat r_t, \pmb p)}} \min\set{1- \sum^{k-1}_{k'=1} z^{k'}_{{(\hat r_t, \pmb p)}} (\pmb l_t, \pmb x), z^{k}_{{(\hat r_t, \pmb p)}} (\pmb l_t, \pmb x)} \cdot \mathds{1}_{\set{\sum^{k-1}_{k'=1} z^{k'}_{{(\hat r_t, \pmb p)}}(\pmb l_t, \pmb x) < 1}}\label{e:np_2}
    \\
    &=  \sum^{2 \hat N}_{k=1} \gamma^k_{{(\hat r_t, \pmb p)}} \min\set{1 - 0, z^{k}_{{(\hat r_t, \pmb p)}} (\pmb l_t, \pmb x)}\cdot  \mathds{1}_{\set{\sum^{k-1}_{k'=1} z^{k'}_{{(\hat r_t, \pmb p)}}(\pmb l_t, \pmb x) < 1}}\label{e:np_andrea}
    \\
    &= \sum^{2 \hat N}_{k=1} \gamma^k_{{(\hat r_t, \pmb p)}} 
     z^{k}_{{(\hat r_t, \pmb p)}}(\pmb l_t, \pmb x)\cdot 
    \mathds{1}_{\set{\sum^{k-1}_{k'=1} z^{k'}_{{(\hat r_t, \pmb p)}}(\pmb l_t, \pmb x) < 1}} \label{e:np_3}
    \\
    % &= \sum^{K_{(\hat r_t, \pmb p)}}_{k=1} \gamma^k_{{(\hat r_t, \pmb p)}} z^{k}_{{(\hat r_t, \pmb p)}}\\
    &= \min\set{\gamma^k_{{(\hat r_t, \pmb p)}}: z^k_{{(\hat r_t, \pmb p)}} (\pmb l_t, \pmb x) = 1, k \in \set{1,2,\dots, 2\hat N} }\label{e:np_4}\\
    &= \min \parentheses{\set{C_r} \cup \set{C_a(\hat r_t, m): m \in \hat{\mathcal{S}}}}\label{e:np_5}\\
    &=   C_{\mathrm{overall}}(\hat r_t, \hat{\mathcal{S}}).\label{e:np_6}
\end{align}
Equation~\eqref{e:np_1} is
%GN: follows from 
the definition of $C(\pmb r_t, \pmb l_t, \pmb x)$ that we provide in Eq.~~\eqref{eq:cost-expression}. The batch of requests consists only of a single request $\pmb r_t = \parentheses{\mathds{1}_{\set{i = \hat r_t}}}_{(i, \pmb p) \in \mathcal{R}}$; this gives Equation~\eqref{e:np_2}. Note that $K_\rho = 2 \hat N$ since a request can be served with $\hat N$ different models at node~2 and node 1.
Equation~\eqref{e:np_andrea} can be understood as follows. Consider a model rank $k$ that makes the indicator function $\mathds{1}_{\set{\sum^{k-1}_{k'=1} z^{k'}_{{(\hat r_t, \pmb p)}}(\pmb l_t, \pmb x) < 1}}$ true. This means that no model with lower rank $k'<k$ is able to satisfy the single request we are considering, which implies that the effective capacity is always $z^{k'}_{{(\hat r_t, \pmb p)}} (\pmb l_t, \pmb x)=0$ for any $k'<k$.
Equation~\eqref{e:np_3} follows from $z^{k'}_{{(\hat r_t, \pmb p)}} (\pmb l_t, \pmb x) \in \{0,1\}$ for $k'\in\set{1, \dots, 2 \hat N}$.
% We observe that $\min\set{1- \sum^{k-1}_{k'=1} z^{k'}_{{(\hat r_t, \pmb p)}} (\pmb l_t, \pmb x), z^{k}_{{(\hat r_t, \pmb p)}} (\pmb l_t, \pmb x)}$ is redundant w.r.t. $\mathds{1}_{\set{\sum^{k-1}_{k'=1} z^{k'}_{{(\hat r_t, \pmb p)}} < 1}}$; thus,  Equation~\eqref{e:np_3}  holds. 
Equation~\eqref{e:np_4} follows from the fact that 1) for each request $\rho$, model service costs are ordered in non-decreasing order (i.e., $\gamma^k_{{(\hat r_t, \pmb p)}} \le \gamma^{k+1}_{{(\hat r_t, \pmb p)}}$) and 2) the term $z^k_{{(\hat r_t, \pmb p)}} (\pmb l_t, \pmb x) \mathds{1}_{\set{\sum^{k-1}_{k'=1} z^{k'}_{{(\hat r_t, \pmb p)}}(\pmb l_t, \pmb x) < 1}}$ can only be non-zero for a single value of $k$; indeed, as we only need to satisfy one request, only one model will be involved in serving that request.
%definition of $\gamma^k_{{(\hat r_t, \pmb p)}}$ in Eq.~(11) and observing that $\mathds{1}_{\set{\sum^{k-1}_{k'=1} z^{k'}_{{(\hat r_t, \pmb p)}} < 1}}$ can only be non-zero for a single value of $k$. 
Equation~\eqref{e:np_5} follows from the definition of  the costs  $C^{p_{j}}_{\pmb p, m} = C_a(i, m)$ for $p_j =1$ and $C^{p_{j}}_{\pmb p, m}  = C_r$ for $p_j=2$ for $m \in \mathcal{M}$. Equation~\eqref{e:np_6} is a direct result from the definition of $C_{\mathrm{overall}}$ in Eq.~\eqref{e:c_overall}

Hence, finding the optimal allocation $\pmb x_*$ that minimizes 
%(resp. maximize) 
the total allocation cost in~\eqref{eq:problem2}
%\frac{1}{T}\sum^T_{t=1} C(\pmb r_t, \pmb l_t, \pmb x) = \frac{1}{T} \sum^T_{t=1}  C_{\mathrm{overall}}(\hat r_t, \hat{\mathcal{S}})$ (resp. gain) in this subproblem 
is equivalent to solve problem~\eqref{eq:problem}, which is NP-hard.
    
\end{proof}

\stopchange
\section{Equivalent Expression of the Gain Function}
\label{appendix:gain_expression}

\begin{lemma}
\label{lemma:mins_subtraction}
{Let us fix the threshold} $c \in \naturals \cup \{0\}$, {request type} $\request \in \requestSet$, {model rank} $k \in [{\modelsNo_{\request}}]${, time slot $t$, load vector $\loadVec_t$ and allocation vector $\allocVec$. For brevity, let us denote $\auxalloc^{k'}_{\request}=\auxalloc^{k'}_{\request}(\loadVec_t, \allocVec)$. The following formula holds:}
\begin{align}
    \min\left\{c,  \sum^{k}_{k'=1} \auxalloc^{k'}_{\request} \right\} -   \min\left\{c, \sum^{k-1}_{k'=1} \auxalloc^{k'}_{\request} \right\} = \min\left\{c-\sum^{k-1}_{k'=1} \auxalloc^{k'}_{\request} , \auxalloc^{k}_{\request}\right\}
    \cdot \mathds{1}\left\{\sum^{k-1}_{k'=1} \auxalloc^{k'}_{\request} < c\right\}.
    \label{eq:mins_subtract}
\end{align}
\end{lemma}

\begin{proof}
We distinguish two cases:
\begin{enumerate}[{(I)}]
\item When {the $k-1$ less costly models have at least $c$ effective capacity, i.e.,} $\sum^{k-1}_{k'=1} \auxalloc^{k'}_{\requestWithPath} \geq c$, we obtain:
\begin{align}
    \sum^{k}_{k'=1} \auxalloc^{k'}_{\request}  = \sum^{k-1}_{k'=1} \auxalloc^{k'}_{\request}  + \auxalloc^{k}_{\request}\geq c + \auxalloc^{k}_{\request} \geq c
    \label{eq:mins_sub_case1}.
\end{align}
The last inequality is obtained using $\auxalloc^{k}_{\request} \geq 0$. {Therefore, the left term of Eq.~\eqref{eq:mins_subtract} becomes $c-c = 0$, and the indicator function of the right term becomes zero. {Hence}, Eq.~\eqref{eq:mins_subtract} is verified in this case.}
\item When $\sum^{k-1}_{k'=1} \auxalloc^{k'}_{\request} < c$, we obtain:
	\begin{align}
	   % \nonumber
	    \min\left\{ c,  \sum^{k}_{k'=1} \auxalloc^{k'}_{\request} \right\} -   \min\left\{c,  \sum^{k-1}_{k'=1} \auxalloc^{k'}_{\request} \right\} &=\min\left\{c,  \sum^{k}_{k'=1} \auxalloc^{k'}_{\request}\right\} - \sum^{k-1}_{k'=1} \auxalloc^{k'}_{\request} =\min\left\{c - \sum^{k-1}_{k'=1} \auxalloc^{k'}_{\request}, \auxalloc^{k}_{\request}  \right\}.
	    \label{eq:mins_sub_case2}
	\end{align}
Hence Eq.~\eqref{eq:mins_subtract} is verified, being the indicator function equal to 1 in this case.
\end{enumerate}

By combining Eq.~\eqref{eq:mins_sub_case1} and Eq.~\eqref{eq:mins_sub_case2} we obtain Eq.~\eqref{eq:mins_subtract}.
\end{proof}

\begin{lemma}
\label{lemma:cost_equivalent_expression}
The cost function given by Eq.~\eqref{eq:cost-expression} can be expressed as:
\begin{align}
\label{eq:cost_equivalent}
    \systemCost(\requestBatchVec_t, \loadVec_t, \allocVec)  =     \sum_{\request \in \requestSet}\sum_{k=1}^{\modelsNo_{\requestWithPath}-1}   \left(\smallCost^k_{\request} - \smallCost^{k+1}_{\request}\right)\min\left\{\requestBatch^t_{\request},  \sum^{k}_{k'=1} \auxalloc^{k'}_{\request}( \loadVec_t, \allocVec) \right\}  + \smallCost^{\modelsNo_{\request}}_{\request} \requestBatch^t_{\request}.
\end{align}
\end{lemma}
\begin{proof}
% The expression in Eq. \eqref{eq:cost_equivalent} is then obtained with several algebraic manipulations.
% \todoi{aa: I think we should add the calculations here. Otherwise this proof is just saying to the reader ``In order to prove this lemma, you just need to trust me. I did the right calculation on my piece of paper.'' If it is too long, we may only keep it in an extended technical report, outsied the TON paper.}
The sum  $\sum^{k}_{k'=1} \auxalloc^{k'}_{\request}(\loadVec_t, \allocVec)$ for $k = K_\request$ surely includes a repository model as it sums all the models along the path of request type $\rho$; thus, we have $\sum^{k}_{k'=1} \auxalloc^{k'}_{\request}(\loadVec_t, \allocVec) \geq r^t_\request$ (see Eq.~\eqref{eq:repo_feasibility_constraint}) and 
\begin{align}
    \smallCost^{K_\request}_{\request} \min\left\{\requestBatch^t_{\request}, \sum^{K_\request}_{k'=1} \auxalloc^{k'}_{\request}(\loadVec_t, \allocVec)\right\} =\smallCost^{K_\request}_{\request} r^t_\request.
    \label{eq:repo_inclusion}
\end{align}

Now we use Lemma~\ref{lemma:mins_subtraction} to express the cost function in Eq. \eqref{eq:cost-expression} as a sum of the difference of min functions. 

\begin{align}
      \systemCost(\requestBatchVec_t, \loadVec_t, \allocVec)\; &=
	   \sum_{\request \in \requestSet} \sum_{k=1}^{\modelsNo_{\request}}
	   \smallCost^k_{\request}
	   \cdot
	   \min\left\{\requestBatch^t_{\request}\;{-}\hspace{-0.1em}\sum^{k-1}_{k'=1} \auxalloc^{k'}_{\request}(\loadVec_t, \allocVec),\hspace{0.1em} 
	   \auxalloc^{k}_{\request}(\loadVec_t, \allocVec)
	   \right\} 
	   \cdot
	   \mathds{1}_{\left\{\sum^{k-1}_{k'=1} \auxalloc^{k'}_{\request}(\loadVec_t, \allocVec) < \requestBatch^{t}_{\request}\right\}}
	   \\
	   &= \sum_{\request \in \requestSet} \sum_{k=1}^{\modelsNo_{\request}}
	   \smallCost^k_{\request} \left( \min\left\{\requestBatch^t_{\request}, \sum^{k}_{k'=1} \auxalloc^{k'}_{\request}(\loadVec_t, \allocVec)\right\} - \min\left\{\requestBatch^t_{\request}, \sum^{k-1}_{k'=1} \auxalloc^{k'}_{\request}(\loadVec_t, \allocVec)\right\} \right) \\
	   &\stackrel{\eqref{eq:repo_inclusion}}{=}\sum_{\request \in \requestSet} \sum_{k=1}^{\modelsNo_{\request}-1}
	   \smallCost^k_{\request} \min\left\{\requestBatch^t_{\request}, \sum^{k}_{k'=1} \auxalloc^{k'}_{\request}(\loadVec_t, \allocVec)\right\} - \sum_{\request \in \requestSet} \sum_{k=1}^{\modelsNo_{\request}}
	   \smallCost^k_{\request} \min\left\{\requestBatch^t_{\request}, \sum^{k-1}_{k'=1} \auxalloc^{k'}_{\request}(\loadVec_t, \allocVec)\right\} + \smallCost^{\modelsNo_{\request}}_{\request} \requestBatch^t_{\request}\\
	   &=\sum_{\request \in \requestSet} \sum_{k=1}^{\modelsNo_{\request}-1}
	   \smallCost^k_{\request} \min\left\{\requestBatch^t_{\request}, \sum^{k}_{k'=1} \auxalloc^{k'}_{\request}(\loadVec_t, \allocVec)\right\} - \sum_{\request \in \requestSet} \sum_{k=2}^{\modelsNo_{\request}}
	   \smallCost^k_{\request} \min\left\{\requestBatch^t_{\request}, \sum^{k-1}_{k'=1} \auxalloc^{k'}_{\request}(\loadVec_t, \allocVec)\right\} + \smallCost^{\modelsNo_{\request}}_{\request} \requestBatch^t_{\request}\\
	   &= \sum_{\request \in \requestSet} \sum_{k=1}^{\modelsNo_{\request}-1}
	   \smallCost^k_{\request} \min\left\{\requestBatch^t_{\request}, \sum^{k}_{k'=1} \auxalloc^{k'}_{\request}(\loadVec_t, \allocVec)\right\} - \sum_{\request \in \requestSet} \sum_{k=1}^{\modelsNo_{\request}-1}
	   \smallCost^{k+1}_{\request} \min\left\{\requestBatch^t_{\request}, \sum^{k}_{k'=1} \auxalloc^{k'}_{\request}(\loadVec_t, \allocVec)\right\} + \smallCost^{\modelsNo_{\request}}_{\request} \requestBatch^t_{\request}\\
	   &=  \sum_{\request \in \requestSet} \sum_{k=1}^{\modelsNo_{\request}-1}
	   \left(\smallCost^k_{\request} - \smallCost^{k+1}_{\request}\right) \min\left\{\requestBatch^t_{\request}, \sum^{k}_{k'=1} \auxalloc^{k'}_{\request}(\loadVec_t, \allocVec)\right\} + \smallCost^{\modelsNo_{\request}}_{\request} \requestBatch^t_{\request}.
\end{align}

\end{proof}

\noindent\textbf{Proof of Lemma~\ref{lemma:gain}}.
\begin{proof}
% the following is wrong because it was assuming the repository will not serve with the worst model in general
%We subtract the cost in Eq.~\eqref{eq:cost_equivalent}, from the cost incurred without intermediate allocation, given by 
%\begin{equation}
%\systemCost(\requestBatchVec_t, \loadVec_t, \repoVec) = \sum_{\request \in \requestSet}  \smallCost^{\modelsNo_{\request}}_{\request} \requestBatch^t_{\request} .
% 	    \label{eq:repository-cost}
%\end{equation}
By using the expression Eq.~\eqref{eq:cost_equivalent} for a generic allocation vector $\allocVec$ and for $\repoVec$, we obtain:
\begin{align}
&\systemGain(\requestBatchVec,\loadVec_t, \allocVec)  = \systemCost(\requestBatchVec,\loadVec_t, \repoVec) -   \systemCost(\requestBatchVec,\loadVec_t, \allocVec)  
\\
    &= \sum_{\request \in \requestSet}\sum_{k=1}^{\modelsNo_{\requestWithPath}-1}   \left(\smallCost^{k}_{\request} - \smallCost^{k+1}_{\request}\right) \min\left\{\requestBatch^t_{\request},  \sum^{k}_{k'=1} {\auxalloc}^{k'}_{\request}(\loadVec_t, \repoVec)  \right\} - \sum_{\request \in \requestSet}\sum_{k=1}^{\modelsNo_{\request}-1}   \left(\smallCost^{k}_{\request} - \smallCost^{k+1}_{\request}\right) \min\left\{\requestBatch^t_{\request},  \sum^{k}_{k'=1} \auxalloc^{k'}_{\request}(\loadVec_t, \allocVec)  \right\} \\
    &= \sum_{\request \in \requestSet}\sum_{k=1}^{\modelsNo_{\requestWithPath}-1}   
    \left(
        \smallCost^{k}_{\request} - \smallCost^{k+1}_{\request}
    \right) \cdot
    \left\{
        \min\left\{\requestBatch^t_{\request},  \sum^{k}_{k'=1} \auxalloc^{k'}_{\request}(\loadVec_t, \repoVec)  \right\}
        -
        \min\left\{\requestBatch^t_{\request},  \sum^{k}_{k'=1} {\auxalloc}^{k'}_{\request}(\loadVec_t, \allocVec)  \right\}
    \right\} \\
    &=  \sum_{\request \in \requestSet}\sum_{k=1}^{\modelsNo_{\requestWithPath}-1}   
    \left(
        \smallCost^{k+1}_{\request} - \smallCost^{k}_{\request}
    \right) \cdot
    \left\{
                \min\left\{\requestBatch^t_{\request},  \sum^{k}_{k'=1} {\auxalloc}^{k'}_{\request}(\loadVec_t, \allocVec)   \right\}  - \min\left\{\requestBatch^t_{\request},  \sum^{k}_{k'=1} \auxalloc^{k'}_{\request}(\loadVec_t, \repoVec)  \right\}
    \right\}
    \\
    &=  
    \sum_{\request \in \requestSet}\sum_{k=1}^{\modelsNo_{\requestWithPath}-1}   
    \left(
        \smallCost^{k+1}_{\request} - \smallCost^{k}_{\request}
    \right) 
   {
    \left(
        \Auxalloc^k_\request(\requestBatchVec_t,\loadVec_t, \allocVec)
        - 
        \Auxalloc^k_\request(\requestBatchVec_t,\loadVec_t, \repoVec)
    \right)
    }
    .
\end{align}

% \rev{aa}{where $\beta^{k}_{ \request} (\requestBatchVec_t,\loadVec_t, \repoVec) \triangleq \min\left\{\requestBatch^t_\request,  \sum^{k}_{k'=1} {\auxalloc}^{k'}_{\request}(\loadVec_t, \repoVec)\right\}$.}{}
\end{proof}
\newpage
\section{Projection Algorithm}
\label{appendix:projection}

% At the end of a time slot, after performing a subgradient update (in the dual space then map back) at each node $v \in \vertices$
In order to project a fractional allocation $\vec y'$ lying outside the constraint set $\allocSetFrac$ to a feasible allocation $\allocVecFrac$, we perform a Bregman projection associated to the global mirror map $\Phi: \mathcal{D} \to \mathbb{R}$, where the Bregman divergence associated to the mirror map $\Phi: \mathcal{D}\to \mathbb{R}$ is given by
\begin{align}
    D_\Phi (\vec y, \vec y') = \Phi(\vec y) - \Phi(\vec y') - {\nabla \Phi(\vec y')}^T (\vec y - \vec y'), 
    \label{def:bregman_global_mirrormap}
\end{align}
and the Bregman divergences associated to the mirror maps $\Phi^v:  \mathcal{D}^v\to \mathbb{R}$ are also given by 
\begin{align}
    D_{\Phi^v} (\vec y^v, \vec y'^v) = {\Phi^v}(\vec y^v) - {\Phi^v}(\vec y'^v) - {\nabla {\Phi^v}(\vec y'^v)}^T (\vec y^v - \vec {y}'^v).
    \label{def:bregman_local_mirrormap}
\end{align}
The projection operation yields a constrained minimization problem, i.e.,  
\begin{align}
    	\vec{y} = \textstyle\prod_{\allocSetFrac \cap \mathcal{D}}^\Phi(\allocVecFrac') &= \underset{\allocVecFrac \in \allocSetFrac\cap \mathcal{D}}{\mathrm{argmin}} \, D_{\Phi}(\allocVecFrac,\allocVecFrac') 
    	\label{eq:global_projection}.
\end{align}

% Consider the weighted negative entropy map defined f $\Phi^v: \mathcal{D}^v \to \reals$ given by $\Phi^v(\allocVecFrac^v) = \sum_{m\in\modelSet}s^v_{m}\allocFrac^v_{m} \log(\allocFrac^v_{m})$

The global mirror map $\Phi: \mathcal{D} \to \reals$ is defined as the sum of the weighted negative entropy maps $\Phi^v: \mathcal{D}^v \to \reals$, where $\mathcal{D} = \reals^{\vertices \times \modelSet}_{+}$ is the domain of $\Phi$, and the set $\mathcal{D}^v = \reals^{\vertices}_{+} $ is the domain of $\Phi^v$  for all $v \in \vertices$ . Thus, it follows that the global Bregman divergence is the sum of the  Bregman divergences local to each node $v \in \vertices$, i.e., 
\begin{align}
    D_{\Phi}(\allocVecFrac,\allocVecFrac') = \sum_{v \in \vertices} D_{\Phi^v} (\allocFrac^v, \allocFrac'^v) 
    \label{eq:sum_divergences}
\end{align} 
where $\vec y \in \allocSetFrac = \bigtimes_{v\in\vertices}\allocSetFrac^v$, and $\vec y^v \in \allocSetFrac^v, \forall v \in \vertices$. In order to minimize the value $D_{\Phi}(\allocVecFrac,\allocVecFrac')$ for $\allocVecFrac \in \allocSetFrac\cap \mathcal{D}$, we can independently minimize the values $D_{\Phi^v} (\allocFrac^v, \allocFrac'^v)$ for $\allocVecFrac^v \in \allocSetFrac^v\cap \mathcal{D}^v$ giving $|\vertices|$ subproblems; for every $v \in \vertices$ we perform the following projection

\begin{align}
    	\vec{y}^v = \textstyle\prod_{\allocSetFrac^v \cap \mathcal{D}^v}^{\Phi^v}(\allocVecFrac'^v) &= \underset{\allocVecFrac^v \in \allocSetFrac^v\cap \mathcal{D}^v}{\mathrm{argmin}} \, D_{\Phi^v}(\allocVecFrac^v,\allocVecFrac'^v).
    	\label{eq:sub-projection}
\end{align}

\begin{algorithm}[!t]
	\caption{Weighted negative entropy Bregman projection onto the weighted capped simplex}
	\begin{algorithmic}[1]
		\Require $|\modelSet|$; $\capacity^v$; $\vec{s}^v$; Sorted $\vec{y}'^v$ where ${{y_{|\modelSet|}'^v}} \geq \dots \geq {{y_1'^v}}$
		\State ${y_{|\modelSet|+1}'^v}\gets +\infty$
		\For{$k \in \{|\modelSet|,|\modelSet|-1,  \dots , 1\}$ }
		\State $m_k \gets \frac{\capacity^v  - \sum^{|\modelSet|}_{m=k+1}s^v_m }{ \sum_{m = 1}^{k} s^v_m {y_m'^v}}$
		\If {${y_{k}'^v} m_{k} < 1 \leq {y_{k+1}'^v} m_{k}$}\Comment{Appropriate $k$ is found}
		\For{$k' \in \{1, 2, \dots, k\}$}
		\State $y^v_{k'} \gets m_k y'^v_{k'}$ \Comment{Scale the variable's components}
		\EndFor
	    \For{$k' \in \{k+1, k+2, \dots, |\modelSet|\}$}
		\State $y^v_{k'} \gets 1$\Comment{Cap the variable's components to 1}
		\EndFor
		\State \Return $\vec{y}^v$\Comment{$\vec{y}^v$ is the result of the projection}
		\EndIf
		\EndFor
	\end{algorithmic}
	\label{alg:bregman_divergence_projection}
\end{algorithm}

% For completeness, we provide the complete proof with the appropriate changes.
\begin{theorem}
\label{theorem:bregman_divergence}
Algorithm \ref{alg:bregman_divergence_projection} when executed at node $v \in \vertices$ returns  $\Pi^{\Phi^v}_{\allocSetFrac^v \cap \mathcal{D}^v}({\vec{y}'^v})$, i.e., the projection of the vector $\vec{y}'^v$ onto the weighted capped simplex $\allocSetFrac^v\cap \mathcal{D}^v$ under the weighted negative entropy $\Phi^v(\allocVecFrac^v) = \sum_{m\in\modelSet}s^v_{m}\allocFrac^v_{m} \log(\allocFrac^v_{m})$. The time complexity of the projection is $\BigO{|\modelSet| \log(|\modelSet|)}$.
\end{theorem}
\begin{proof}

% We take $\mathcal{M} {=} \cup_{i \in \taskcatalog} \mathcal M_i=\{1,2,\dots, |\modelSet|\} \in \naturals$ as the catalog of all the available models.

\begin{align}
	\textstyle\prod_{\allocSetFrac^v \cap \mathcal{D}^v}^{\Phi^v}(\allocVecFrac'^v) &= \underset{\allocVecFrac^v \in \allocSetFrac^v\cap \mathcal{D}^v}{\mathrm{argmin}} \, D_{{\Phi^v}}(\allocVecFrac^v,\allocVecFrac'^v) \\
	&= \underset{\allocVecFrac^v \in \allocSetFrac^v \cap \mathcal{D}^v}{\mathrm{argmin}} \sum_{m \in \modelSet} s^v_m \left({y^v_m}\mathrm{log}\left(\frac{{y^v_m}}{{y'^v_m}}\right) - {y^v_m} + {y'^v_m} \right).
\end{align}
We adapt the negative entropy projection algorithm in~\cite{sisalem21icc}. The constraints ${y^v_m} > 0, \forall m \in \modelSet$ are implicitly enforced by the negentropy mirror map ${\Phi^v}$ and $D_{\Phi^v}(\allocVecFrac^v, \allocVecFrac'^v)$ is convex in $\allocVecFrac^v$. The Lagrangian function of the above problem:
\begin{align}
	\mathcal{J} (\allocVecFrac^v,  \mathbf{\beta}, \tau) =  \sum_{m \in \modelSet} {s^v_m}\left({y^v_m}\mathrm{log}\left(\frac{{y^v_m}}{{y'^v_m}}\right)- {y^v_m} + {y'^v_m}\right) - \sum_{m \in \modelSet} \beta_m \left(1-{y^v_m}\right) - \tau \left(\sum_{m \in \modelSet}{s^v_m} {y^v_m} - b^v\right).
\end{align}
At optimal point $\vec{\hat y}^v$ the following KKT conditions hold:
\begin{subequations}
	\begin{align}
		{s^v_m} \mathrm{log}({\hat y^v_m}) - {s^v_m} \mathrm{log}({y'^v_m}) + \beta_m - {s^v_m} \tau   = 0 \label{eqkkt1}, \\
		{\hat y^v_m} \leq 1 \label{eqkkt2},                                                    \\
		\beta_m \geq0 \label{eqkkt3}       ,                                          \\
		\sum_{m \in \modelSet} {s^v_m} {\hat y^v_m} = \capacity^v \label{eqkkt4},                                                \\
		\beta_m (1-{\hat y^v_m}) =0 \label{eqkkt5}.
	\end{align}
	\label{eqKKT}
\end{subequations}
Without loss of generality, assume the components of $\vec{\hat{y}}^v$ are in non-decreasing order. Let $k$ be the index of the largest component of $\vec{\hat{y}}^v$ strictly smaller than , i.e., 
\begin{align}
{\hat y}^v_1 \leq ... \leq {\hat y}^v_k < {\hat y}^v_{k+1} = ... = {\hat y}^v_{|\modelSet|}  = 1 \,\,\, \text{if} \,  k < {|\modelSet|}, \label{eq:kktcond1}\\
{\hat y}^v_1 \leq {\hat y}^v_2 \leq... \leq {\hat y}^v_{|\modelSet|} < 1 \,\,\, \text{if} \,  k = {|\modelSet|} \label{eq:kktcond2}.
\end{align}

The goal here is to identify a valid value for $k$ (number of components of $\vec {\hat{y}}^v$ different from 1) and $\tau$. For now assume that $\tau$ is known, so a valid $k \in \modelSet$ should satisfy the following:
\begin{itemize}
	\item For $m_L = 1, ..., k$,  we have from \eqref{eqkkt5} that $\beta_{m_L} = 0$, and then from \eqref{eqkkt1},  $s^v_{m_L}\mathrm{log}(y'^v_{m_L}) + s^v_{m_L}\tau  = s^v_{m_L}\mathrm{log}({\hat y}^v_{m_L})< s^v_{m_L}\mathrm{log}(1) =0$, and can be simplified to \begin{align}
	    y'^v_{m_L} e^\tau < 1 \label{eq:kkt:p1}, \forall {m_L} \in  \{1,\dots, k\} .
	\end{align}
	\item For ${m_U} = k+1, ..., {|\modelSet|}$: as $\beta_{m_U} \geq 0$ from \eqref{eqkkt3}, we get $0 =s^v_{m_U} \mathrm{log}({\hat y}^v_{m_U}) = s^v_{m_U} \mathrm{log}(y'^v_{m_U}) - \beta_{m_U} + s^v_{m_U} \tau \leq s^v_{m_U}\mathrm{log}(y'^v_{m_U}) + s^v_{m_U}\tau$, and can be simplified to \begin{align}
	    y'^v_{m_U} e^\tau \geq 1, \forall {m_U} \in  \{k+1, \cdots, {|\modelSet|}\} .
	    \label{eq:kkt:p2}
	\end{align}
\end{itemize}
Consider Eqs.~\eqref{eq:kktcond1} and \eqref{eq:kktcond2}, and since for $m_L \in \{1, \dots, k\}$ we have $y'^v_{m_L} e^\tau = \hat{y}^v_{m_L}$ (the order is preserved), then the conditions in Eq.~\eqref{eq:kkt:p1} are
\begin{align}
    y'^v_1 e^\tau \leq ...\leq  y'^v_{k} e^\tau  < 1.
    \label{eq:p1_proj}
\end{align}
If the components of $\vec{y}'^v$ are ordered in ascending order, then it is enough to check if $y'^v_{k} e^\tau  < 1$ holds for Eq.~\eqref{eq:p1_proj} to be true. Moreover, for $m_U \in \{k+1, \dots, |\modelSet|\}$, we have ${\hat y}^v_{m_U} = 1$ and $y'^v_{m_U} e^\tau \geq 1$. Then, by taking ${y_{|\modelSet|+1}'^v}\triangleq +\infty$ ($k$ can be equal to $|\modelSet|$ as in Eq.~\eqref{eq:kktcond2}) it is enough to check with the smallest $y'^v_{m_U}$ to summarize all the conditions in  Eq.~\eqref{eq:kkt:p2}. Thus, all the needed conditions can be further simplified to:
$$y'^v_{k} e^\tau  < 1 \leq y'^v_{k+1} e^\tau.$$  
Note that the r.h.s inequality is ignored when $k = |\modelSet|$ by construction ($y'^v_{|\modelSet|+1} = +\infty$). 

Now we established how to verify if a given $k \in \modelSet$ is valid, what remains is to give the expression of $\tau$ using the knapsack constraint in Eq.~\eqref{eqkkt4}:
\begin{equation*}
	\capacity^v= \sum^{|\modelSet|}_{m=1}{{s^v_m} {\hat {y}^v_m}} 	   =  \sum^{|\modelSet|}_{m=k+1} {s^v_m} + e^\tau \sum_{m = 1}^{k} {s^v_m} {y'^v_m}.
\end{equation*}
For a given $k \in \modelSet$, we define
\begin{align}
	m_k \triangleq e^{\tau} & = \frac{\capacity^v  - \sum^{|\modelSet|}_{m=k+1} {s^v_m} }{ \sum_{m = 1}^{k} s^v_m {y'^v_m}}.
\end{align}
Thus, a valid $k$ is the value satisfying the following inequalities (line 7 of Algorithm~\ref{alg:bregman_divergence_projection}): 
\begin{align}
    y'^v_{k} m_k  < 1 \leq y'^v_{k+1} m_k. 
\end{align}
The appropriate $k$ satisfying the KKT conditions is contained in $\modelSet$, and due to the sorting operation this gives total time complexity of $\BigO{|\modelSet| \log(|\modelSet|)}$ per iteration. In practice, the online mirror ascent method quickly sets irrelevant items in the fractional allocation vector $\allocVecFrac'^v$ very close to 0. Therefore, we can keep track only of items  with a fractional value above a threshold $\epsilon >0$ , and the size of this subset is practically $\ll |\modelSet|$. Therefore, the projection can be very efficient in practice. 
\end{proof}

%%%%%%%%%%%%%%%%%%%%%%%%%%%%%%%%%%%%%%%%%%%%%%%%%%%%%%%
%%%%%%%%%%%%%%%%%%%%%% SUBGRADIENT EXPRESSION %%%%%%%%%
%%%%%%%%%%%%%%%%%%%%%%%%%%%%%%%%%%%%%%%%%%%%%%%%%%%%%%%
\section{Subgradient Expression}
\label{appendix:subgradient_expression}

% The one-sided directional derivative \cite{rockafellar2015convex} of a function $\systemGain: \allocSetFrac \to \reals$ at a point $\vec{y} \in \allocSetFrac$ is defined to be
% \begin{align}
%     \partial_{\vec{v}} G(\vec{y}) \triangleq \underset{h \to 0}{\lim} \frac{\systemGain(\vec{y} + h \vec{v}) - \systemGain(\vec{y})}{h}.
% \end{align}
% We use the shorthand notations $\partial^+_{i} \systemGain(\vec{y})$ and $\partial^-_{i} \systemGain(\vec{y})$ to denote the directed one-sided derivatives w.r.t the basis vector $\vec{e}_i$ given by $\partial_{\vec{e}_i} G(\vec{y})$ and $-\partial_{-\vec{e}_i} G(\vec{y})$, respectively.    

% \begin{lemma}
% The vector $\vec{g}$ satisfying $g_i \in \left\{\partial^+_{i} \systemGain(\vec{y}), \partial^-_{i} \systemGain(\vec{y})\right\}$ is a subgradient of the concave function $\systemGain: \allocSetFrac \to \reals$ at point $\vec{y} \in \allocSetFrac$.
% \end{lemma}
% \begin{proof}
% % From \cite[Theorem 23.2]{rockafellar2015convex}, $\vec{g}$ is a subgradient of the concave function $G$ at point $\vec{y}$ $\emph{iff}$ 

% % Replacing $\vec{v}$ by $\vec{e}_i$ we obtain 
% % \begin{align}
    
% % \end{align}

% \end{proof}

\begin{lemma}
The gain function in  Eq.~\eqref{eq:gain-compact} has a subgradient $\vec g_t$ at point $\allocVecFrac_t \in \allocSetFrac$ given by
\begin{align}
    \vec g_t = \left[\sum_{\request \in \requestSet}
     l^{t,v}_{\rho, m}\left(\smallCost^{\modelsNo^*_{\request}(\allocVecFrac_t)}_{\request} - C^v_{\requestPathVec, m}\right) \mathds{1}_{\{\kappa_{\request} (v, m) < \modelsNo^*_{\request}(\allocVecFrac_t)\}} \right]_{(v,m) \in \vertices \times \modelSet}, 
\end{align}
where  $\modelsNo^*_{\request}(\allocVecFrac_t) =  \min\big\{ k \in [\modelsNo_{\request}-1]: \sum^{k}_{k'=1} \auxalloc^{k'}_{\request}(\loadVec_t, \allocVecFrac_t) \geq \requestBatch^t_{\request}\big\}$.

% The vector constructed through values \eqref{eq:subgradient_expression} is a subgradient at point $\allocVecFrac_t \in \allocSetFrac$ of the gain function given by 
% {
% \begin{align}
%     \systemGain(\requestBatchVec_t,\loadVec_t, \allocVecFrac) =
%     \sum_{\request \in \requestSet}
%     \sum_{k=1}^{\modelsNo_{\request}-1} 
%     \left(\smallCost^{k+1}_{\request} - \smallCost^{k}_{\request}\right)
%     \cdot
%     \left(
%     \Auxalloc^k_\request(\requestBatchVec_t,\loadVec_t,\allocVecFrac)
%     - \Auxalloc^k_\request(\requestBatchVec_t,\loadVec_t, \repoVec)
%     \right), \forall \allocVecFrac \in \allocSetFrac
%     \label{eq:gain-compact}\end{align}.
% }

\label{lemma:subgradient}
\end{lemma}

\begin{proof}

The function given by $ \Auxalloc^k_\request(\requestBatchVec_t,\loadVec_t,\allocVecFrac_t) = \min\left\{
   \textstyle \requestBatch^t_\request,  \sum^{k}_{k'=1} {\auxalloc}^{k'}_{\request}(\loadVec_t, \allocVecFrac_t)\right\}$ is a minimum of two concave differentiable functions (a constant, and a linear function). We can characterize its subdifferential (set of all possible subgradients), using \cite[Theorem 8.2]{mordukhovich2017geometric}, at point $\allocVecFrac_t \in \allocSetFrac$ as
   \begin{align}
      \textstyle \partial \Auxalloc^k_\request(\requestBatchVec_t,\loadVec_t,\allocVecFrac_t) = \begin{cases}
      \left\{ \nabla ( \sum^{k}_{k'=1} {\auxalloc}^{k'}_{\request}(\loadVec_t, \allocVecFrac_t))\right\} &\mathrm{if}\,\, \Auxalloc^k_\request(\requestBatchVec_t,\loadVec_t,\allocVecFrac_t) <  \requestBatch^t_\request\qquad(\text{r.h.s. argument of the min is active}), \\
      \convexhull{\left\{\vec{0}, \nabla ( \sum^{k}_{k'=1} {\auxalloc}^{k'}_{\request}(\loadVec_t, \allocVecFrac_t))\right\}} &\mathrm{if}\,\, \Auxalloc^k_\request(\requestBatchVec_t,\loadVec_t,\allocVecFrac_t) =  \requestBatch^t_\request \qquad(\text{both arguments of the min are active}), \\
      \{\vec{0}\}  &\mathrm{otherwise}\
      \qquad\qquad\qquad(\text{l.h.s. argument of the min is active}),
       \end{cases}
   \end{align}
where $\convexhull{\,\cdot\,}$ is the convex hull of a set, and the gradient $\nabla$ is given by $\nabla (\,\cdot\,) = [\frac{\partial\phantom{y^v_m}}{\partial y^v_m} (\,\cdot\,)]_{(v,m) \in \vertices \times \modelSet}$. The operator $\frac{\partial\phantom{y^v_m}}{\partial y^v_m} (\,\cdot\,)$ is the partial derivative w.r.t $y^v_m$ (not to be confused with the subdifferential notation). 

We restrict ourselves to the valid subgradient $\tilde{\vec{g}}^k_{\rho,t} \in \partial \Auxalloc^k_\request(\requestBatchVec_t,\loadVec_t,\allocVecFrac_t)$ given by
\begin{align}
    \tilde{\vec{g}}^k_{\rho,t} = \begin{cases}
       \nabla ( \sum^{k}_{k'=1} {\auxalloc}^{k'}_{\request}(\loadVec_t, \allocVecFrac_t)) & \mathrm{if}\,\, \Auxalloc^k_\request(\requestBatchVec_t,\loadVec_t,\allocVecFrac_t) <  \requestBatch^t_\request, \\
     \vec{0} , & \mathrm{otherwise}.
       \end{cases}
       \label{eq:subgradient_instance}
\end{align}

Note that for every $(v,m) \in \vertices \times \modelSet$ we have
\begin{align}
\nonumber
\frac{\partial \phantom{y^v_m}}{\partial y^v_m} 
\sum^{k}_{k'=1} \auxalloc^{k'}_{\request}(\loadVec_t, \allocVecFrac_t)
&=
\sum^{k}_{k'=1} \frac{\partial\phantom{y^v_m}}{\partial y^v_m} 
\auxalloc^{k'}_{\request}(\loadVec_t, \allocVecFrac_t)
\stackrel{\eqref{eq:several-definitions}}{=}
\load^{t,v}_{\request,m} \cdot \mathds{1}_{\{\kappa_{\request} (v, m) \leq k\}}.
\end{align}
{The indicator variable $\mathds{1}_{\{\kappa_{\request} (v, m) \leq k\}}$ is introduced since the partial derivative of $\sum^{k}_{k'=1} \auxalloc^{k'}_{\request}(\loadVec_t, \allocVecFrac_t)$ w.r.t. $y^v_m$ is non-zero only if model $m$ at node $v$ is among the $k$ best models to serve requests of type $\rho$ (in this case, the variable $y^v_{m}$ appears once in the summation).}

We obtain from Eq.~\eqref{eq:subgradient_instance}
\begin{align}
   \tilde{g}^{k,v}_{\rho,t,m} &= \begin{cases}
    \load^{t,v}_{\request,m} \cdot  \mathds{1}_{\{\kappa_{\request} (v, m) \leq k\}} &\mathrm{if}\,\,  \Auxalloc^k_\request(\requestBatchVec_t,\loadVec_t,\allocVecFrac_t) < r^t_\rho, \\
     0  &\mathrm{otherwise.}
    \end{cases}\\
    &= \load^{t,v}_{\request,m}  \cdot\mathds{1}_{\{\kappa_{\request} (v, m) \leq k \,\,\land\,\, \Auxalloc^k_\request(\requestBatchVec_t,\loadVec_t,\allocVecFrac_t) < r^t_\rho \}}, \forall (v,m) \in \vertices \times \modelSet.
\end{align}
By considering the subdifferential 
\begin{align}
     \partial  \systemGain(\requestBatchVec_t,\loadVec_t, \allocVecFrac_t) = \partial \left(\sum_{\request \in \requestSet}\sum_{k=1}^{\modelsNo_{\request}-1}  \left(\smallCost^{k+1}_{\request} - \smallCost^{k}_{\request}\right) \left(
    \Auxalloc^k_\request(\requestBatchVec_t,\loadVec_t,\allocVecFrac)
    - \Auxalloc^k_\request(\requestBatchVec_t,\loadVec_t, \repoVec)
    \right) \right),
\end{align}
and using \cite[Theorem 23.6]{rockafellar2015convex}, we get 
\begin{align}
     \partial  \systemGain(\requestBatchVec_t,\loadVec_t, \allocVecFrac_t) = \sum_{\request \in \requestSet}\sum_{k=1}^{\modelsNo_{\request}-1} \partial\left( \left(\smallCost^{k+1}_{\request} - \smallCost^{k}_{\request}\right) \left(
    \Auxalloc^k_\request(\requestBatchVec_t,\loadVec_t,\allocVecFrac)
    - \Auxalloc^k_\request(\requestBatchVec_t,\loadVec_t, \repoVec)
    \right)\right).
\end{align}
The constant factors $\left(\smallCost^{k+1}_{\request} - \smallCost^{k}_{\request}\right)$ are non-negative, so we can multiply both sides of the subgradient inequality by a non-negative constant \cite[Sec.~23]{rockafellar2015convex}; furthermore, the subgradient of the constants $\Auxalloc^k_\request(\requestBatchVec_t,\loadVec_t, \repoVec)$ is $\vec{0}$. We get
\begin{align}
     \partial  \systemGain(\requestBatchVec_t,\loadVec_t, \allocVecFrac_t) = \sum_{\request \in \requestSet}\sum_{k=1}^{\modelsNo_{\request}-1}  \left(\smallCost^{k+1}_{\request} - \smallCost^{k}_{\request}\right) \partial \left(
    \Auxalloc^k_\request(\requestBatchVec_t,\loadVec_t,\allocVecFrac)
    - \Auxalloc^k_\request(\requestBatchVec_t,\loadVec_t, \repoVec)
    \right) = \sum_{\request \in \requestSet}\sum_{k=1}^{\modelsNo_{\request}-1}  \left(\smallCost^{k+1}_{\request} - \smallCost^{k}_{\request}\right) \partial
    \Auxalloc^k_\request(\requestBatchVec_t,\loadVec_t,\allocVecFrac).
\end{align}
Then, a subgradient $\vec{g}_t \in   \partial  \systemGain(\requestBatchVec_t,\loadVec_t, \allocVecFrac_t)$ at point $\allocVecFrac_t \in \allocSetFrac$ is given by
\begin{align}
    \vec{g}_t = \sum_{\request \in \requestSet}
    \sum_{k=1}^{\modelsNo_{\request}-1} 
    \left(\smallCost^{k+1}_{\request} - \smallCost^{k}_{\request}\right)
     \tilde{\vec{g}}^k_{\rho,t}.
\end{align}
The $(v,m)$-th component of the subgradient $\vec{g}_t$ is
\begin{align}
    g^v_{t,m} &= \sum_{\request \in \requestSet}
    \sum_{k=1}^{\modelsNo_{\request}-1} 
    \left(\smallCost^{k+1}_{\request} - \smallCost^{k}_{\request}\right)
      \cdot \tilde{g}^{k,v}_{\rho,t,m} \\
       &= 
    \sum_{\request \in \requestSet}
    \sum_{k=1}^{\modelsNo_{\request}-1} \load^{t,v}_{\request,m} 
    \left(\smallCost^{k+1}_{\request} - \smallCost^{k}_{\request}\right) \cdot 
  \mathds{1}_{\{\kappa_{\request} (v, m) \leq k \,\,\land\,\, \Auxalloc^k_\request(\requestBatchVec_t,\loadVec_t,\allocVecFrac_t) < r^t_\rho \}}
   \\
      &= 
    \sum_{\request \in \requestSet}
    \sum_{k=\kappa_{\request} (v, m)}^{\modelsNo_{\request}-1} \load^{t,v}_{\request,m} 
    \left(\smallCost^{k+1}_{\request} - \smallCost^{k}_{\request}\right) \cdot 
    \mathds{1}_{\{\sum^{k}_{k'=1} \auxalloc^{k'}_{\request}(\loadVec_t, \allocVecFrac_t) < r^t_{\request}\}}
   \\
   &= 
    \sum_{\request \in \requestSet}
    \sum_{k=\kappa_{\request} (v, m)}^{\modelsNo^*_{\request}(\allocVecFrac_t)-1} \load^{t,v}_{\request,m} 
    \left(\smallCost^{k+1}_{\request} - \smallCost^{k}_{\request}\right) 
    \\
    &=
    \sum_{\request \in \requestSet}
    \load^{t,v}_{\request,m} 
    \cdot
    \sum_{k=\kappa_{\request} (v, m)}^{\modelsNo^*_{\request}(\allocVecFrac_t)-1} 
    \left(\smallCost^{k+1}_{\request} - \smallCost^{k}_{\request}\right) 
   \\
   &= 
    \sum_{\request \in \requestSet}
     \load^{t,v}_{\request,m} 
    \left(\smallCost^{\modelsNo^*_{\request}(\allocVecFrac_t)}_{\request} - \smallCost^{\kappa_{\request} (v, m)}_{\request}\right) \cdot\mathds{1}_{\{\kappa_{\request} (v, m) < \modelsNo^*_{\request}(\allocVecFrac_t)\}}\\
    &\stackrel{\eqref{eq:several-definitions}}{=}\sum_{\request \in \requestSet}
     \load^{t,v}_{\request,m} 
    \left(\smallCost^{\modelsNo^*_{\request}(\allocVecFrac_t)}_{\request} - C^v_{\requestPathVec, m}\right) \cdot\mathds{1}_{\{\kappa_{\request} (v, m) < \modelsNo^*_{\request}(\allocVecFrac_t)\}},    \forall (v,m) \in \vertices \times \modelSet,
\end{align}
where  $\modelsNo^*_{\request}(\allocVecFrac_t) =  \min\big\{ k \in [\modelsNo_{\request}-1]: \sum^{k}_{k'=1} \auxalloc^{k'}_{\request}(\loadVec_t, \allocVecFrac_t) \geq \requestBatch^t_{\request}\big\}$.

\end{proof}

\section{Supporting Lemmas for the Proof of Theorem~\ref{th:regret_bound}}
\label{apx:regret_constants}
\subsection{Concavity of the Gain Function}
\label{sec:concavity}
\begin{lemma}
The gain function given by Eq.~\eqref{eq:gain-compact} is concave over its domain $\allocSetFrac$ of possible fractional allocations.
\label{lemma:concavity}
\end{lemma}
\begin{proof}
Since $\auxload^k_\rho$ is defined to be the $k$-th smallest cost for any $k \in [K_\request]$ (see Eq.~\eqref{eq:several-definitions}), then the factors $\smallCost^{k+1}_{\request} - \smallCost^{k}_{\request}$ are always non-negative. Moreover $\auxalloc^{k}_{  \request}(\loadVec_t, \allocVecFrac) = \allocFrac^{v}_{m} \load^{t,v}_{\request,m}$, where $v,m$ are such that $\kappa_{\request} (v, m)=k$. Therefore, $\Auxalloc^k_\request(\loadVec_t, \allocVecFrac)$ is the minimum between a constant $\requestBatch^t_\request$ and a sum  $\sum^{k}_{k'=1} {\auxalloc}^{k'}_{\request}(\loadVec_t, \allocVec)$ of linear functions of $\allocVecFrac$. Such minimum is thus a concave function of $\allocVecFrac$. Therefore, the gain in Eq.~\eqref{eq:gain-compact} is a weighted sum with positive weights of concave functions in $\allocVecFrac$, which is concave. 
\end{proof}

\subsection{Strong convexity of the Mirror Map}
\label{sec:strong_convexity}
\begin{lemma}
The global mirror map $\Phi(\vec{y}) = \sum_{v \in \vertices} \Phi(\vec{y}^v) = \sum_{v \in \vertices} \sum_{m \in \modelSet} s^v_m y^v_m \log(y^v_m)$ defined over the domain $\mathcal{D} = \mathbb{R}_{>0}^{|\modelSet| \times |\vertices|}$  is 
$\theta$-strongly convex w.r.t. the norm $\norm{\,\cdot\,}_{l_1(\vec{s})}$ over $\mathcal{Y} \cap \mathcal{D}$, where 
\begin{align}
    \theta &\triangleq\frac{1}{ s_{\max} |\vertices| |\modelSet|},
    \label{eq:theta}\\
  \textstyle \norm{\vec{y}}_{l_1(\vec{s})} &\triangleq \sum_{(v,m) \in \vertices \times \modelSet} {s}^v_m|y^v_m| \quad \text{(weighted $l_1$ norm)},
  \label{eq:primary_norm}
\end{align}
and $s_{\max} \triangleq \left\{ s^v_m : (v,m) \in \vertices\times\modelSet \right\}$ is the maximum model size. 
% The map $\sigma^v: \modelSet \to \modelSet$ gives the ordering of the models sizes at node $v \in \vertices$, such that $\textstyle s^v_{\sigma^v_1} \leq s^v_{\sigma^v_2} \leq\dots \leq s^v_{\sigma^v_{|\modelSet|}}$. 
In words, this means that the mirror map $\Phi$'s growth is lower bounded by a quadratic with curvature $\frac{1}{ s_{\max} |\vertices| |\modelSet|}$. 
% and $y^v_{\sigma^v_i}$ is the associated fractional variable with size order $i$.
\label{lemma:strong_convexity}
\end{lemma}
\begin{proof}
% Assume that  We define the constant $U$  as 
% \begin{align}
    % \textstyle U \triangleq\max\left\{ i \in \modelSet: \sum^{i}_{j=1} s^v_{\sigma^v_j} \geq  \capacity^v : (i,v) \in \taskcatalog \times \vertices\right\}.
%     \label{eq:U_def}
% \end{align}
% Consider that we want to maximize the number of stored models at a node $v \in \vertices$, then we need fit as many small models as we can  without exceeding the budget constraint. Then, the value $U$ provides then an upper bound to the maximum number of models we can store at a node. It also represents the maximum number of components that can be set to 1 for any  $\vec{y}^v \in \mathcal{Y}^v$. Obviously, $U$ is bounded by $|\modelSet|$. 
% %
% %
% From the definition of $U$ we have for any $\allocVecFrac \in \allocSetFrac$
% \begin{align}
%     \sum_{m \in \modelSet} y^v_m \leq U, \forall v \in \vertices.
%     \label{eq:upper_bound_on_sum_variables}
% \end{align}

We extend the proof of the strong convexity of the negative entropy w.r.t. to the $l_1$ norm over the simplex given in~\cite[Lemma 16]{shalev2007online}. The map $\Phi(\vec{y})$ is differentiable over $\mathcal{Y}  \cap \mathcal{D} $, so a sufficient (and also necessary) condition  for $\Phi(\vec{y})$ to be $\theta$-strongly convex w.r.t. $\norm{\,\cdot\,}_{l_1(\vec{s})}$ is:
\begin{align}
    (\nabla\Phi(\vec{y'}) -\nabla\Phi(\vec{y}))^T (\vec{y'} - \vec{y}) \geq \theta \norm{\vec{y'} - \vec{y}}^2_{l_1(\vec{s})}, \quad\forall \vec{y'},\vec{y} \in \mathcal{Y}  \cap \mathcal{D} .
\end{align}
We have
\begin{align}
 (\nabla\Phi(\vec{y'}) -\nabla\Phi(\vec{y}))^T (\vec{y'} - \vec{y}) =
\sum_{(v,m) \in \vertices \times \modelSet} s^v_m(\log({y'}^v_m) - \log({y}^v_m)) ({y'}^v_m - y^v_m).
\end{align}
Take $\mu^v_m \triangleq s^v_m(\log({y'}^v_m) - \log(y^v_m)) ({y'}^v_m - y^v_m)$, and note that $\mu^v_m \geq 0$ (because $\log$ is an increasing function).
\begin{align*}
\norm{\vec{y'}- \vec{y}}^2_{l_1(\vec{s})} &= \left(\sum_{(v,m) \in \vertices \times \modelSet} s^v_m|{y'}^v_m - y^v_m| \right)^2= \left( \sum_{(v,m) \in \vertices \times \modelSet: \mu^v_m \neq 0} \sqrt{\mu^v_m} \frac{ s^v_m |{y'}^v_m - y^v_m|}{\sqrt{\mu^v_m}} \right)^2 \\
						&\leq \left(\sum_{(v,m) \in \vertices \times \modelSet: \mu^v_m \neq 0} \mu^v_m\right) \left( \sum_{(v,m) \in \vertices \times \modelSet: \mu^v_m \neq 0} (s^v_m)^2 \frac{({y'}^v_m - y^v_m)^2}{\mu^v_m}\right)\\
						&= \left(\sum_{(v,m) \in \vertices \times \modelSet: \mu^v_m \neq 0}  s^v_m(\log({y'}^v_m) {-} \log(y^v_m)) ({y'}^v_m {-} y^v_m) \right) \left( \sum_{(v,m) \in \vertices \times \modelSet: \mu^v_m \neq 0} s^v_m \frac{{y'}^v_m - y^v_m}{\log({y'}^v_m) - \log(y^v_m)}\right).
\end{align*}
The inequality is obtained using Cauchy–Schwarz inequality. Take ${s^v_m} \leq s_{\max}, \forall (v,m) \in \vertices \times \modelSet$, we obtain: 
\begin{align*}
\sum_{(v,m) \in \vertices \times \modelSet: \mu^v_m \neq 0}   {s^v_m} \frac{{y'}^v_m - y^v_m}{\log({y'}^v_m) - \log(y^v_m)}    &\leq s_{\max} \sum_{(v,m) \in \vertices \times \modelSet: \mu^v_m \neq 0} \frac{{y'}^v_m - y^v_m}{\log({y'}^v_m) - \log(y^v_m)} \\
&= s_{\max} \sum_{v \in \vertices} \sum_ {m \in \modelSet: \mu^v_m \neq 0} \frac{{y'}^v_m - y^v_m}{\log({y'}^v_m) - \log(y^v_m)} \\
&\leq s_{\max} \sum_{v \in \vertices} \sum_{m \in  \modelSet} \frac{{y'}^v_m + y^v_m}{2}\\
&\leq  s_{\max} |\vertices| |\modelSet|.  
\end{align*}

The second inequality is shown in~\cite[Eq.~(A.16)]{shalev2007online}. We find that $\forall \vec{y'}, \vec{y} \in \mathcal{Y}  \cap \mathcal{D}$:
\begin{align}
\frac{1}{ s_{\max} |\vertices| U}\norm{\vec{y'} - \vec{y}}^2_{l_1(\vec{s})} \leq  \sum_{(v,m) \in \vertices \times \modelSet: \mu_m^v \neq 0}  {s^v_m}^v(\log({y'}_m^v) - \log(y_m^v)) ({y'}_m^v - y_m^v) = (\nabla\Phi(\vec{y'}) -\nabla\Phi(\vec{y}))^T (\vec{y'}- \vec{y}).
\end{align}
The strong convexity constant $\theta$ is $\frac{1}{s_{\max} |\vertices| |\modelSet|}$.
\end{proof}

\subsection{Subgradient Bound}
\label{sec:subgradient_bound}
\begin{lemma}
For any $(\vec r_t, \vec l_t) \in \advSet$, the subgradients $\vec{g}_t$ of the gain function in Eq.~\eqref{eq:gain-compact} at point $\allocVecFrac_t \in \allocSetFrac$ are bounded under the norm $\norm{\,\cdot\,}_{l_\infty (\frac{1}{\vec{s}})}$ by $\sigma = \frac{R L_{\max} \Delta_C}{s_{\min}}$, where $s_{\min} \triangleq \min\{s^v_m: \forall (v,m) \in \vertices \times \modelSet\}$, $L_{\max} \triangleq \max\{L^v_m: \forall (v,m) \in \vertices \times \modelSet\}$,  $R = |\requestSet|$, and $\Delta_C \triangleq \max \left\{\left(\sum_{m \in \modelSet} \repo^{\nu(\requestPathVec)}_{m'} C^{\nu(\requestPathVec)}_{\requestPathVec,m'}\right) - C^{v}_{\requestPathVec,m}: \forall (i, \requestPathVec) \in \requestSet, (v,m) \in \requestPathVec \times \modelSet \right\}$ is the maximum serving cost difference between serving at a repository node $\nu(\requestPathVec)$ and at any other node $v \in \requestPathVec$. The norm  $\norm{\,\cdot\,}_{l_\infty (\frac{1}{\vec{s}})}$ is defined as

\begin{align}
  \textstyle \norm{\vec{y}}_{l_\infty(\frac{1}{\vec{s}})} &\triangleq \max \left\{ \frac{|y^v_m|}{{s}^v_m} : (v,m) \in \vertices \times \modelSet\right\}.
  \label{eq:dual_norm}
\end{align}

\label{lemma:subgradient_bound}
\end{lemma}
\begin{proof}
  We have for any $t \in [T]$
 \begin{align}
    \norm{\vec{g}_t}_{l_\infty(\frac{1}{\vec{s}})}  &=  \max\left\{\frac{|g^v_{t,m}|}{s^v_m} , \forall (v,m) \in \vertices \times \modelSet\right\} \leq  \max\left\{\frac{|g^v_{t,m}|}{s_{\min}}, \forall (v,m) \in \vertices \times \modelSet\right\}\\
    &\stackrel{\eqref{eq:subgradient_expression}}{\leq} \frac{ L_{\max} }{s_{\min}} \max\left\{ \sum_{\request \in \requestSet}
     \left(
     \smallCost^{\modelsNo^*_{\request}(\allocVecFrac_t)}_{\request} 
     - C^v_{\requestPathVec, m}
     \right)
     \cdot\mathds{1}_{\{\kappa_{\request} (v, m) < \modelsNo^*_{\request}(\allocVecFrac_t)\}} ,\forall  (v,m) \in \vertices \times \modelSet\right\} \\
     &\stackrel{\eqref{eq:several-definitions}}{\leq} \frac{L_{\max} R}{s_{\min}} \max\{ \smallCost_\request^{K_\request}- \smallCost_\request^{1}, \forall \request \in \requestSet\} \leq \frac{L_{\max} R \Delta_C}{s_{\min}} = \sigma.
     \end{align}
\end{proof}
% $\Delta_C$ is an upper bound on cost difference between serving at a repository node and at any other node. When considering the cost model in Eq.~\eqref{eq:serving_cost2}, it is given by the weight of the heaviest path, the largest accuracy difference, and the largest inference delay difference between the repository and the source node.
\subsection{Dual Norm}
\begin{lemma}
$\norm{\,\cdot\,}_{l_\infty(\frac{1}{\vec{s}})}$ is the dual norm of $\norm{\,\cdot\,}_{l_1(\vec{s})}$  defined in \eqref{eq:dual_norm} and \eqref{eq:primary_norm}, respectively. 

\label{lemma:dual_norm}
\end{lemma}
\begin{proof}

The dual norm  $\norm{\,\cdot\,}_*$ of $\norm{\,\cdot\,}_{l_1(\vec{s})}$  is defined as (e.g., \cite{bubeck2015convexbook})
\begin{align}
    \norm{\vec z}_* \triangleq \sup_{\vec y \in \mathbb{R}^{\vertices\times\modelSet}} \left\{ \vec z^T \vec y: \norm{\vec y }_{l_1(\vec{s})} \leq 1\right\}, \forall \vec z \in  \mathbb{R}^{\vertices\times\modelSet}.
    \label{eq:supremum}
\end{align}

We thus need to show that 
$\norm{\vec z}_{l_\infty(\frac{1}{\vec{s}})}
=
\sup_{\vec y \in \mathbb{R}^{\vertices\times\modelSet}} \left\{ \vec z^T \vec y: \norm{\vec y }_{l_1(\vec{s})} \leq 1\right\}, \forall \vec z \in  \mathbb{R}^{\vertices\times\modelSet}.
$

Consider any two vectors $\vec{y}$ and $\vec{z}$  in $\mathbb{R}^{\vertices\times\modelSet}$. We have 
\begin{align}
 \vec z^T  \vec y  &=   \sum_{(v,m) \in \vertices \times \modelSet} y^v_m z^v_m= 
 \sum_{(v,m) \in \vertices \times \modelSet} (s^v_m y^v_m) \left(\frac{z^v_m}{s^v_m} \right)
 \le 
 \sum_{(v,m) \in \vertices \times \modelSet} (s^v_m \cdot |y^v_m|) \left(\frac{|z^v_m|}{s^v_m} \right)
 \\ 
  &\leq  \textstyle \left(\sum_{(v,m) \in \vertices \times \modelSet} s^v_m |y^v_m|\right) \max\left\{\frac{|z^v_m|}{s^v_m}:(v,m) \in \vertices \times \modelSet \right\}  = \norm{\vec{y}}_{l_1(\vec{s})} \norm{\vec{z}}_{l_\infty(\frac{1}{\vec{s}})}.
\end{align}
Observe that
%By taking $\norm{\vec{y}}_{l_1(\vec{s})} \leq 1$ we get 
\begin{align}
    \label{eq:aux}
    \vec y^T \vec z \leq \norm{\vec{z}}_{l_\infty(\frac{1}{\vec{s}})}, \,\,\,\, \forall \vec{y}: \norm{\vec{y}}_{l_1(\vec{s})} \leq 1.
\end{align}
Let $(v_*,m_*) = \underset{_{(v,m) \in \vertices \times \modelSet}}{\argmax} \left\{\frac{|z^v_m|}{s^v_m}\right\}$. The equality is achieved in~\eqref{eq:aux} when $\vec y_* = \left[\frac{\mathrm{sign}(z^v_m)}{s^v_m} \mathds{1}_{\{(v,m)=(v_*, m_*)\}}\right]_{(v,m) \in \vertices \times \modelSet}$. Note that $\norm{\vec y_* }_{l_1(\vec{s})}=1 \leq 1$, then the supremum in~\eqref{eq:supremum} is attained for $\vec y = \vec y_*$ and  has value $\norm{\vec{z}}_{l_\infty(\frac{1}{\vec{s}})}$; therefore, $\norm{\,\cdot\,}_{l_\infty(\frac{1}{\vec{s}})}$ is the dual norm of $\norm{\,\cdot\,}_{l_1(\vec{s})}$.

\end{proof}
\subsection{Bregman Divergence Bound}
\label{sec:bregman_divergence_bound}

\begin{lemma}
The value of the Bregman divergence $D_\Phi(\allocVecFrac, \allocVecFrac_1)$ in Eq.~\eqref{def:bregman_global_mirrormap} associated with the mirror map $\Phi(\allocVecFrac) = \sum_{v \in \vertices} \Phi^v(\allocVecFrac^v) =  \sum_{v\in\vertices} \sum_{m\in\modelSet}s^v_{m}\allocFrac^v_{m} \log(\allocFrac^v_{m})$ is upper bounded by a constant \begin{align}
    D_{\max} \triangleq  \sum_{v \in \vertices} \min\{\capacity^v,\norm{\vec{s}^v}_1\} \log\left(\frac{\norm{\vec{s}^v}_1}{\min\{\capacity^v,\norm{\vec{s}^v}_1\}}\right)  \geq D_\Phi(\allocVecFrac, \allocVecFrac_1).
\end{align}
where  $y^v_{1,m} = \frac{\min\{\capacity^v,\norm{\vec{s}^v}_1\}}{\norm{\vec{s}^v}_1},\forall (v,m) \in \vertices \times \modelSet$ and $\vec s^v = [s^v_m]_{m\in\modelSet}$ for every $v \in \vertices$.
\label{lemma:bregman_bound}
\end{lemma}
\begin{proof}
 We prove that $\allocVecFrac^v_1$ is the minimizer of $\Phi^v$ over $\allocSetFrac^v$. As $\Phi^v$ is convex over $\allocSetFrac^v$ and differentiable in $\allocVecFrac^v_1$,  $\allocVecFrac^v_1$ is a minimizer if and only if ${\nabla \Phi^v(\allocVecFrac^v_1)}^T (\allocVecFrac^v_1 - \allocVecFrac) \leq 0, \forall \allocVecFrac^v \in \allocSetFrac^v$~\cite[Proposition 1.3]{bubeck2015convexbook} (first order optimality condition). Note that from the definition of $\allocSetFrac^v$ (see Sec.~\ref{sec:algorithm}) we have for any $\allocVecFrac^v \in \allocSetFrac^v$
\begin{align}
    \sum_{m \in \modelSet} s^v_m y^v_m =
    % \capacity^v
    % = %  See definition of Y^v I removed the added part.
    \min\{\capacity^v, \norm{\vec s^v}_1\}.
    \label{eq:def2_repeated}
\end{align}
% where the last equality is due to the fact that $\norm{\vec s^v}_1=\sum_{m \in \modelSet} s^v_m\ge \sum_{m \in \modelSet} s^v_m y^v_m$, since $y^v_m\le 1$. 
% All of this is from the definition. (sorry for commenting the text).

Let $c =  \frac{\min\{\capacity^v,\norm{\vec{s}^v}_1\}}{\norm{\vec{s}^v}_1}$, we get
\begin{align}
    {\nabla \Phi^v(\allocVecFrac^v_1)}^T (\allocVecFrac^v_1 - \allocVecFrac) &= \sum_{m \in \modelSet} s^v_m (\log(c) + 1) (c - y^v_m) = (\log(c) + 1) (c \norm{\vec s^v}_1 - \sum_{m \in \modelSet} s^v_m y^v_m)\\
    &\stackrel{\eqref{eq:def2_repeated}}{=} (\log(c)  + 1) \left( c\norm{\vec s^v}_1 - \sum_{m \in \modelSet} s^v_m y^v_m \right) = (\log(c)  + 1) \left(\min\{\capacity^v,\norm{\vec{s}^v}_1\} - \min\{\capacity^v,\norm{\vec{s}^v}_1\}\right)\\
    &= 0.
\end{align}
We confirmed that $\allocVecFrac^v_1$ is a minimizer of $\Phi^v$ over $\allocSetFrac^v$. We have $\Phi^v(\allocVecFrac^v) \leq 0, \forall \allocVecFrac^v \in \mathcal{Y}^v$, and using the first order optimality condition we obtain
\begin{align}
    D_{\Phi^v}(\allocVecFrac^v, \allocVecFrac^v_1) &\stackrel{\eqref{def:bregman_local_mirrormap}}{=} \Phi^v(\allocVecFrac^v) - \Phi^v(\allocVecFrac^v_1) + {\nabla \Phi^v(\allocVecFrac^v_1)}^T (\allocVecFrac^v_1 - \allocVecFrac^v)\leq \Phi^v(\allocVecFrac^v) - \Phi^v(\allocVecFrac^v_1) \leq  -\Phi^v(\allocVecFrac^v_1) \\
    &= \min\{\capacity^v,\norm{\vec{s}^v}_1\} \log\left(\frac{\norm{\vec{s}^v}_1}{ \min\{\capacity^v,\norm{\vec{s}^v}_1\}}\right).
\end{align}
Thus, we obtain
\begin{align}
\sum_{v\in\vertices} D_{\Phi^v}(\allocVecFrac^v, \allocVecFrac^v_1) \stackrel{\eqref{eq:sum_divergences}}{=} D_\Phi(\allocVecFrac, \allocVecFrac_1) \leq \sum_{v \in \vertices} \min\{\capacity^v,\norm{\vec{s}^v}_1\} \log\left(\frac{\norm{\vec{s}^v}_1}{\min\{\capacity^v,\norm{\vec{s}^v}_1\}}\right)    .
\end{align}
\end{proof}
% \subsection{The Overall Regret Constant}
% The regret constant $A$ is given by:
% \begin{align}
%     A =   \psi   &\frac{R L_{\max} \Delta_C}{s_{\min} } \sqrt{s_{\max} |\vertices| U }\sqrt{2 \sum_{v \in \vertices} \min\{\capacity^v,\norm{\vec{s}^v}_1\} \log\left(\frac{\norm{\vec{s}^v}_1}{\min\{\capacity^v,\norm{\vec{s}^v}_1\}}\right)  }
%     \label{eq:overall_regret_constant}
% \end{align}
\subsection{Bounds on the Gain Function}
\label{appendix:bounds}
%%%%%%%%%%%%%%%%%%%%%%%%%%%%%%%%%%%%%%%%

Upper and lower bounds on the gain function in Eq.~\eqref{eq:gain-compact} will be established using the following bounding function

\begin{align}
     \El(\requestBatchVec_t,\loadVec_t, &\allocVecFrac) \triangleq
     \sum_{\request \in \supp{\vec r^t}}\sum_{k=1}^{\modelsNo_{\request}-1}   \left(\smallCost^{k+1}_{\request} - \smallCost^{k}_{\request}\right)
     \requestBatch^t_\request 
     \left(1 - \prod_{k'=1}^{k} \left(1 - \auxalloc_{
    \request}^{k'}(\loadVec_t, \allocVecFrac) / r^t_\rho \right) \right)\mathds{1}_{\{\Auxalloc^k_\request(\requestBatchVec_t,\loadVec_t, \repoVec)  = 0 \}}, \forall \allocVecFrac\in \allocSet \cup \allocSetFrac,
    \label{eq:gain-surrogate}
\end{align}
% \todo{aa: $\cdot$ added to increase readibility. Replaced $\requestBatch_{t,\request}$ with $\requestBatch^t_{\request}$ everywhere}
%  g( \tfrac{\auxalloc^{k'}_{\request}(\loadVec_t, \allocVecFrac)}{\requestBatch^t_\request})
where 
\begin{align}
\supp{\vec r^t} \triangleq \left\{ \rho \in \requestSet:  r^t_\rho \neq  0  \right\}\label{eq:supp_def}
\end{align}
is the set of request types for which there is a non-zero number of requests in the request batch $\vec r^t$.  
%
% Observe that $\Auxalloc^k_\request(\requestBatchVec_t,\loadVec_t, \repoVec)$ is zero when the model at  the repository mode has a rank higher than $k$. In other words, the model at the repository is not among the $k$ best models. Therefore, the indicator function at the end of~\eqref{eq:gain-surrogate} keeps in the sum only the low-enough ranks such that no repository model is among the best ones.
% I think it may just add confusion to the reader. It is just a function used to sandwich the value of the gain function, so we don't need to provide interpretations on it. 

\begin{lemma}
The gain function in Eq.~\eqref{eq:gain-compact} can be equivalently expressed as
\begin{align}
     \systemGain(\requestBatchVec_t,\loadVec_t, \allocVecFrac) =
    \sum_{\request \in \supp{\vec r^t}}\sum_{k=1}^{\modelsNo_{\request}-1}   \left(\smallCost^{k+1}_{\request} - \smallCost^{k}_{\request}\right)  \min\left\{r^t_\rho,   \sum^{k}_{k'=1}\auxalloc^{k'}_{\request}(\loadVec_t, \allocVecFrac) \right\}  \mathds{1}_{\{\Auxalloc^k_\request(\requestBatchVec_t,\loadVec_t, \repoVec)  = 0\}}, \forall \allocVecFrac \in \allocSet \cup \allocSetFrac.
    \label{eq:equivalent_gain_for_proof}
\end{align}
\end{lemma}
\begin{proof}
Remember from the definition in Eq.~\eqref{eq:sum_of_auxvars} that ${\Auxalloc^k_\request(\requestBatchVec_t,\loadVec_t, \allocVecFrac)} = \min\left\{\requestBatch^t_{\request} ,
    \sum^{k}_{k'=1}\auxalloc^{k'}_{\request}(\loadVec_t, \allocVecFrac) \right\} $.
    We observe that $\Auxalloc^k_\request(\requestBatchVec_t,\loadVec_t, \repoVec)$ is not a function of $\allocVecFrac$ and it is equal to $0$, 
%when a repository model's rank is greater than $k$, or $r^t_\rho$ when its rank is less than k.
when there is no repository with model's rank smaller or equal to $k$, and to $r^t_\rho$, otherwise; therefore, $\Auxalloc^k_\request(\requestBatchVec_t,\loadVec_t, \repoVec) \in \{0, r^t_\rho\}$.

When $\Auxalloc^k_\request(\requestBatchVec_t,\loadVec_t, \repoVec)  \neq 0$, and thus 
$\Auxalloc^k_\request(\requestBatchVec_t,\loadVec_t, \repoVec)  = r^t_\rho$, the following holds
\begin{align}
   {\Auxalloc^k_\request(\requestBatchVec_t,\loadVec_t, \allocVecFrac)} - {\Auxalloc^k_\request(\requestBatchVec_t,\loadVec_t, \repoVec)} 
   &=  \min\left\{\requestBatch^t_{\request} ,
    \sum^{k}_{k'=1}\auxalloc^{k'}_{\request}(\loadVec_t, \allocVecFrac) \right\}  -{\Auxalloc^k_\request(\requestBatchVec_t,\loadVec_t, \repoVec)}  \\
   &= \min\left\{0,    \sum^{k}_{k'=1}\auxalloc^{k'}_{\request}(\loadVec_t, \allocVecFrac)  - \Auxalloc^k_\request(\requestBatchVec_t,\loadVec_t, \repoVec)\right\} = 0.
\end{align}
The last equality holds because $\sum^{k}_{k'=1} \auxalloc^{k'}_{\request}(\loadVec_t, \allocVecFrac)  - \Auxalloc^k_\request(\requestBatchVec_t,\loadVec_t, \repoVec) \geq 0, \forall \allocVecFrac \in \allocSet \cup \allocSetFrac$ from Eq.~\eqref{eq:repo_ineq}.
Otherwise, when $\Auxalloc^k_\request(\requestBatchVec_t,\loadVec_t, \repoVec)  = 0$, we have
${\Auxalloc^k_\request(\requestBatchVec_t,\loadVec_t, \allocVecFrac)} - {\Auxalloc^k_\request(\requestBatchVec_t,\loadVec_t, \repoVec)}={\Auxalloc^k_\request(\requestBatchVec_t,\loadVec_t, \allocVecFrac)}
\stackrel{\eqref{eq:sum_of_auxvars}}{=}
\min\left\{r^t_\rho,   \sum^{k}_{k'=1}\auxalloc^{k'}_{\request}(\loadVec_t, \allocVecFrac) \right\}$.

Hence, 
we can succinctly write, for any value of $ \Auxalloc^k_\request(\requestBatchVec_t,\loadVec_t, \repoVec)$, that
%we have for a generic $ \Auxalloc^k_\request(\requestBatchVec_t,\loadVec_t, \repoVec) \in \{0, r^t_\rho\}$
\begin{align}
   {\Auxalloc^k_\request(\requestBatchVec_t,\loadVec_t, \allocVecFrac)} -{\Auxalloc^k_\request(\requestBatchVec_t,\loadVec_t, \repoVec)}=  \min\left\{r^t_\rho,   \sum^{k}_{k'=1}\auxalloc^{k'}_{\request}(\loadVec_t, \allocVecFrac) \right\} \mathds{1}_{\{\Auxalloc^k_\request(\requestBatchVec_t,\loadVec_t, \repoVec)  = 0\}}.
    \label{eq:part_aa}
\end{align}
%Moreover, if $r^t_\rho = 0$,
Since $\sum^{k}_{k'=1}\auxalloc^{k'}_{\request}(\loadVec_t, \allocVecFrac)$ is non negative, the previous formula implies that
%then we have 
\begin{align}
    {\Auxalloc^k_\request(\requestBatchVec_t,\loadVec_t, \allocVecFrac)} -{\Auxalloc^k_\request(\requestBatchVec_t,\loadVec_t, \repoVec)}=0, \ \ \ \ \text{if  }r^t_\rho = 0.
    \label{eq:part_bb}
\end{align}
%\begin{align}
%     {\Auxalloc^k_\request(\requestBatchVec_t,\loadVec_t, \allocVecFrac)} -{\Auxalloc^k_\request(\requestBatchVec_t,\loadVec_t, \repoVec)}&=  \min\left\{r^t_\rho,   \sum^{k}_{k'=1}\auxalloc^{k'}_{\request}(\loadVec_t, \allocVecFrac) \right\} \mathds{1}_{\{\Auxalloc^k_\request(\requestBatchVec_t,\loadVec_t, \repoVec)  = 0\}}\\
%     &=   \min\left\{0,   \sum^{k}_{k'=1}\auxalloc^{k'}_{\request}(\loadVec_t, \allocVecFrac) \right\} \mathds{1}_{\{\Auxalloc^k_\request(\requestBatchVec_t,\loadVec_t, \repoVec)  = 0\}}= 0.
%     \label{eq:part_bb}
%\end{align}
%The last equality is obtained considering that $\sum^{k}_{k'=1}\auxalloc^{k'}_{\request}(\loadVec_t, \allocVecFrac) \geq 0$. 
By combining Eq.~\eqref{eq:part_aa} and \eqref{eq:part_bb} that 
\begin{align}
       {\Auxalloc^k_\request(\requestBatchVec_t,\loadVec_t, \allocVecFrac)} -{\Auxalloc^k_\request(\requestBatchVec_t,\loadVec_t, \repoVec)}= \min\left\{r^t_\rho,   \sum^{k}_{k'=1}\auxalloc^{k'}_{\request}(\loadVec_t, \allocVecFrac) \right\}  \mathds{1}_{\{\Auxalloc^k_\request(\requestBatchVec_t,\loadVec_t, \repoVec)  = 0 \,\,\land\,\, r^t_\rho \neq 0\}}.
   \label{eq:shortcut_a}
\end{align}
Hence, applying the above equalities on the gain expression we get 
\begin{align}
% \nonumber
    \systemGain(\requestBatchVec_t,\loadVec_t, \allocVecFrac)&\stackrel{\eqref{eq:gain-compact}}{=}
    \sum_{\request \in \requestSet}\sum_{k=1}^{\modelsNo_{\request}-1}   \left(\smallCost^{k+1}_{\request} - \smallCost^{k}_{\request}\right) \left( {\Auxalloc^k_\request(\requestBatchVec_t,\loadVec_t, \allocVecFrac)} -{\Auxalloc^k_\request(\requestBatchVec_t,\loadVec_t, \repoVec)}\right)\\
   &\stackrel{\eqref{eq:shortcut_a}}{=} \sum_{\request \in \requestSet}\sum_{k=1}^{\modelsNo_{\request}-1}   \left(\smallCost^{k+1}_{\request} - \smallCost^{k}_{\request}\right) \min\left\{r^t_\rho,   \sum^{k}_{k'=1}\auxalloc^{k'}_{\request}(\loadVec_t, \allocVecFrac) \right\}  \mathds{1}_{\{\Auxalloc^k_\request(\requestBatchVec_t,\loadVec_t, \repoVec)  = 0 \,\,\land\,\, r^t_\rho \neq 0\}}\\
    &\stackrel{\eqref{eq:supp_def}}{=}
    \sum_{\request \in \supp{\vec r^t}}\sum_{k=1}^{\modelsNo_{\request}-1}   \left(\smallCost^{k+1}_{\request} - \smallCost^{k}_{\request}\right)  \min\left\{r^t_\rho,   \sum^{k}_{k'=1}\auxalloc^{k'}_{\request}(\loadVec_t, \allocVecFrac) \right\}  \mathds{1}_{\{\Auxalloc^k_\request(\requestBatchVec_t,\loadVec_t, \repoVec)  = 0\}}.
\end{align}
\end{proof}

\begin{lemma}
Consider $n \in \naturals$, $\allocVecFrac \in [0,1]^n$, $\vec{q} \in \naturals^n$, and $c \in \naturals$. We assume that $q_i \leq c, \forall i \in [n]$. The following holds
\begin{align}
  \min\left\{c , \sum_{i \in [n]} \allocFrac_i q_i \right\} \geq c - c \prod_{i \in [n]} (1 - \allocFrac_i  q_i /c).
  \label{eq:lower_bound_piece}
\end{align}
\label{lemma:util_bounds_lower}
\end{lemma}
\begin{proof}
We define $a_n \triangleq c - c \prod_{i \in [n]} (1 - \allocFrac_i  q_i/c)$ and $b_n \triangleq \min\left\{c, \sum_{i \in [n]} \allocFrac_i  q_i\right\}$. 

{We first show by induction that, if $a_n \leq b_n$, then this inequality holds also for $n+1$.}

\noindent\emph{Base case ($n=1$).}
\begin{align}
    a_1 &= c - c + y_1 q_1 = y_1 q_1 = \min\{c, q_1 y_1\} = b_1.
\end{align}

\noindent\emph{Induction step.}

\begin{align}
    a_{n+1} &= c - c \prod_{i \in [n+1]} (1 - \allocFrac_i  q_i/c )
    \label{eq:an+1}
    \\
    &=c - c \prod_{i \in [n]} (1 - \allocFrac_i  q_i/c ) (1 - \allocFrac_{n+1} q_{n+1}/c) \\
    &= c - c \prod_{i \in [n]} (1 - \allocFrac_i q_i/c) +  (c \allocFrac_{n+1} q_{n+1}/c) \prod_{i \in [n]} (1 - \allocFrac_i q_i/c)\\
    &= a_n +  \allocFrac_{n+1} q_{n+1}\prod_{i \in [n]} (1 - \allocFrac_i q_i/c) \\
    &\leq a_n +  \allocFrac_{n+1} q_{n+1}.
  \end{align}
The last inequality holds since by construction $q_i\le c$ and thus $0\le y_i q_i /c \le 1$, and $0 \leq \prod_{i \in [n]} \left(1- {y_i q_i}/{c}\right)\leq 1$. For the same reason, $0\le\prod_{i \in [n+1]} (1 - \allocFrac_i  q_i/c )\le 1$ and thus, by~\eqref{eq:an+1}, we have
%. Note that 
$a_{n+1} \leq c$. Moreover, note that if $b_n = c$ then $b_{n+1} = c$. Therefore:

\begin{align}
    a_{n+1} &\leq \min\left\{c, a_n +  \allocFrac_{n+1}  q_{n+1}\right\} \leq \min\left\{c, b_n +  \allocFrac_{n+1}  q_{n+1} \right\}\\
    &= \begin{cases}
    \min\left\{c, \sum^{n+1}_{i=1} \allocFrac_i q_i\right\} = b_{n+1}, &\text{ if}\quad b_n \leq c,  \\
    \min\left\{c, c + \allocFrac_{n+1}  q_{n+1} \right\} = c = b_{n+1}, &\text{ if}\quad b_n = c,
    \end{cases}
\end{align}
and the proof by induction is completed.

\end{proof}

\begin{lemma}
Consider $\allocVecFrac \in [0,1]^n$, $\vec{q} \in \naturals^n$, $c \in \naturals$ and $n \in \naturals$. We assume that $q_i \leq c, \forall i \in [n]$. The following holds
\begin{align}
    c - c\prod_{i \in [n]}  (1- \allocFrac_i q_i / c)  \geq (1-1/e)   \min\left\{c,  \sum_{i \in [n]} \allocFrac_i  q_i\right\}.
     \label{eq:upper_bound_piece}
\end{align}

\label{lemma:util_bounds_upper}
\end{lemma}
\begin{proof}
Our proof follows the same lines of the proof of~\cite[Lemma 3.1]{goemans1994new}. We use the  arithmetic/geometric mean inequality~{\cite{wilf1963some}} on the non-negative variables $1 - y_i q_i/c, i \in [n]$ to obtain:
\begin{align}
   \frac{1}{n} \sum_{i \in [n]}  (1 - \allocFrac_i  q_i / c) \geq\left( \prod_{i \in [n]}  (1- \allocFrac_i q_i/c) \right)^{\frac{1}{n}}.
\end{align}
We reformulate the above as:
\begin{align}
 1 - \prod_{i \in [n]}  (1- \allocFrac_i q_i/c)  &\geq  1 - \left( {\frac{1}{n}} \sum_{i \in [n]} \left(1 - \allocFrac_i q_i/c\right) \right)^{n} =    1 - 
 \left(1 -  {\frac{1}{n}} \sum_{i \in [n]}  \allocFrac_i q_i/c\right)^{n}
 \\
 &\geq  1 - \left(1 -  {\frac{1}{n}} \min\left\{1,  \sum_{i \in [n]} \allocFrac_i q_i/c \right\}\right)^{n}.
 \label{eq:ineq-1}
\end{align}
{To obtain the last inequality, consider that, for any number $z$, we have $z\ge\min\{1,z\}$, and thus 
%that 
$\sum_{i \in [n]}  \allocFrac_i q_i/c \geq \min\left\{1, \sum_{i \in [n]}  \allocFrac_i q_i/c \right\}$.
}

The function $f(z) = 1 - (1-z/n)^{n}$ is concave for $z \in [0,1]$, then, for $z \in [0,1]$, $f(z) \ge f(0) + z \frac{f(1)-f(0)}{1-0}= z f(1)$, as $f(0)=0$.
%~\cite[Lemma 3.1]{goemans1994new}, where  $f(0) = 0$. Therefore, $f(z) \geq z f(1)$. 
Setting $z=\min\left\{1,  \sum_{i \in [n]}  \allocFrac_i  q_i/c\right\}$
%Replacing $z$ with $\min\left\{1,  \sum_{i \in [n]}  \allocFrac_i  q_i/c\right\}$
, we obtain the following:
\begin{align}
  1 - \prod_{i \in [n]}  (1- \allocFrac_i q_i / c)   
  \stackrel{\eqref{eq:ineq-1}}{\geq}
  1 - \left(1 -  {\frac{1}{n}} z\right)^{n}
  \geq \left(1 - (1 - 1/n)^{n}\right)  z
  \geq (1-1/e)  z.
\end{align}
%\begin{align}
%  1 - \prod_{i \in [n]}  (1- \allocFrac_i q_i / c)   &\geq 
%  1 - \left(1 -  {\frac{1}{n}} \min\left\{1,  \sum_{i \in [n]}  \allocFrac_i  q_i/c\right\}\right)^{n}\\
% & \geq \left(1 - (1 - 1/n)^{n}\right)  \min\left\{1,  \sum_{i \in [n]} \allocFrac_i  q_i/c\right\}\\
% &\geq (1-1/e)  \min\left\{1,  \sum_{i \in [n]} \allocFrac_i  q_i/c \right\}.
%\end{align}
The last inequality is obtained since $1 - (1-1/n) ^{n}$ decreases in $n$, and it is  lower bounded by $1-1/e$.  By multiplying both sides of the above inequality by $c \in \naturals$, and replacing $z$ with its value we conclude the proof.
% \begin{align}
%  c\left( 1 - \prod_{i \in [n]}  (1- \allocFrac_i q_i / c)  \right) \geq (1-1/e)  c \min\left\{1,  \sum_{i \in [n]} \allocFrac_i  q_i/c \right\}.
% \end{align}
% This gives 
% \begin{align}
% c - c\prod_{i \in [n]}  (1- \allocFrac_i q_i / c)  \geq (1-1/e)   \min\left\{c,  \sum_{i \in [n]} \allocFrac_i  q_i\right\}.
% % \label{eq:upper_bound_piece}
% \end{align}
\end{proof}

\begin{lemma}
\label{lemma:gain_lower_upper_bound}
for any request batch $\vec r_t$ and potential available capacity $\vec l_t$ such that $(\vec r_t, \vec l_t) \in \mathcal{A}$, the allocation gain $\systemGain(\requestBatchVec_t,\loadVec_t, \allocVecFrac)$ has the following lower  and upper bounds 
\begin{align}
\label{eq:gain_lower_upper_bound}
   \left(1-\tfrac{1}{e}\right)^{-1}\El(\requestBatchVec_t,\loadVec_t, \allocVecFrac) \geq    G(\requestBatchVec_t,\loadVec_t, \allocVecFrac) 
     &\geq \El(\requestBatchVec_t,\loadVec_t, \allocVecFrac), 
     &\forall \allocVecFrac\in\allocSet \cup\allocSetFrac.
\end{align}
\end{lemma}
\begin{proof}

We have the following
\begin{align}
% \nonumber
    \systemGain(\requestBatchVec_t,\loadVec_t, \allocVecFrac)&\stackrel{\eqref{eq:equivalent_gain_for_proof}}{=} 
    \sum_{\request \in \supp{\vec r^t}}\sum_{k=1}^{\modelsNo_{\request}-1}   \left(\smallCost^{k+1}_{\request} - \smallCost^{k}_{\request}\right)  \min\left\{r^t_\rho,   \sum^{k}_{k'=1}\auxalloc^{k'}_{\request}(\loadVec_t, \allocVecFrac) \right\}  \mathds{1}_{\{\Auxalloc^k_\request(\requestBatchVec_t,\loadVec_t, \repoVec)  = 0\}}\\
    &\stackrel{\eqref{eq:lower_bound_piece}}{\geq} \sum_{\request \in \supp{\vec r^t}}\sum_{k=1}^{\modelsNo_{\request}-1}   \left(\smallCost^{k+1}_{\request} - \smallCost^{k}_{\request}\right) r^t_\rho \left(1 - \prod_{k'=1}^{k} \left(1 - \auxalloc_{
    \request, k'}(\loadVec_t, \allocVecFrac) / r^t_\rho \right) \right) \mathds{1}_{\{\Auxalloc^k_\request(\requestBatchVec_t,\loadVec_t, \repoVec)  = 0\}}\label{eq:ineq_lo}  \\
    % &\stackrel{\eqref{eq:shortcut_b}}{=}\sum_{\request \in \supp{\vec r^t}}\sum_{k=1}^{\modelsNo_{\request}-1}   \left(\smallCost^{k+1}_{\request} - \smallCost^{k}_{\request}\right) \left(r^t_\rho - \Auxalloc^k_\request(\requestBatchVec_t,\loadVec_t, \repoVec)\right) \left(1 - \prod_{k'=1}^{k} \left(1 - \auxalloc_{
    % \request, k'}(\loadVec_t, \allocVecFrac) / r^t_\rho \right) \right) \\
    &\stackrel{\eqref{eq:gain-surrogate}}{=} \El(\requestBatchVec_t,\loadVec_t, \allocVecFrac), \forall \allocVecFrac\in\allocSet \cup\allocSetFrac, 
\end{align}
and 
\begin{align}
% \nonumber
   (1- 1/e) \systemGain(\requestBatchVec_t,\loadVec_t, \allocVecFrac)
   &\stackrel{\eqref{eq:equivalent_gain_for_proof}}{=} (1-1/e)
    \sum_{\request \in \supp{\vec r^t}}\sum_{k=1}^{\modelsNo_{\request}-1}   \left(\smallCost^{k+1}_{\request} - \smallCost^{k}_{\request}\right)  \min\left\{r^t_\rho,   \sum^{k}_{k'=1}\auxalloc^{k'}_{\request}(\loadVec_t, \allocVecFrac) \right\}  \mathds{1}_{\{\Auxalloc^k_\request(\requestBatchVec_t,\loadVec_t, \repoVec)  = 0\}}\\
    &\stackrel{\eqref{eq:upper_bound_piece}}{\leq} \sum_{\request \in \supp{\vec r^t}}\sum_{k=1}^{\modelsNo_{\request}-1}   \left(\smallCost^{k+1}_{\request} - \smallCost^{k}_{\request}\right) r^t_\rho \left(1 - \prod_{k'=1}^{k} \left(1 - \auxalloc_{
    \request, k'}(\loadVec_t, \allocVecFrac) / r^t_\rho \right) \right) \mathds{1}_{\{\Auxalloc^k_\request(\requestBatchVec_t,\loadVec_t, \repoVec)  = 0\}} \label{eq:ineq_up} \\
    % &\stackrel{\eqref{eq:shortcut_b}}{=} \sum_{\request \in \supp{\vec r^t}}\sum_{k=1}^{\modelsNo_{\request}-1}   \left(\smallCost^{k+1}_{\request} - \smallCost^{k}_{\request}\right) \left(r^t_\rho - \Auxalloc^k_\request(\requestBatchVec_t,\loadVec_t, \repoVec)\right) \left(1 - \prod_{k'=1}^{k} \left(1 - \auxalloc_{
    % \request, k'}(\loadVec_t, \allocVecFrac) / r^t_\rho \right) \right) \\
    &\stackrel{\eqref{eq:gain-surrogate}}{=} \El(\requestBatchVec_t,\loadVec_t, \allocVecFrac), \forall \allocVecFrac\in\allocSet \cup\allocSetFrac.
\end{align}

Inequalities in Eq.~\eqref{eq:ineq_lo} and Eq.~\eqref{eq:ineq_up} follow from Eq.~\eqref{eq:lower_bound_piece} and Eq.~\eqref{eq:upper_bound_piece}, respectively,  by replacing $c = r^t_\rho$, $q_{k'} = \auxload^{k'}_{\request}(\loadVec_t)$, $x_{k'} = \auxalloc_{
    \request, k'}(\loadVec_t, \allocVecFrac) / \auxload^{k'}_{\request}(\loadVec_t)$, and  $n = k$.
\end{proof}

\begin{lemma}
\label{lemma:depround}
Let the allocation $\vec x^v$ be the random output of { \DepRound} on node $v\in\vertices$ given the fractional allocation $\vec y^v \in \allocSetFrac^v$.
%, $\forall m \in \modelSet$. 
For any subset of the model catalog $S \subset \modelSet$ and any number $c_m \in [0,1], \forall m \in S$, { \DepRound} satisfies the following:
\begin{align}
      \mathbb{E}\left[ \prod_{m \in S} (1 - \alloc^v_{m} c_m)\right] \leq \prod_{m \in S} (1 - \allocFrac^v_{m} c_m).
     \label{eq:sampling_upperbound}
 \end{align}
\end{lemma}
\begin{proof}

\DepRound{} uses a subroutine \Simplify, which, given input variables $y_m, y_{m'} \in (0,1)$, outputs $x_m, x_{m'} \in [0,1]$ with at least one of them being integral ($0$ or $1$).  Note that the input to \Simplify{} is never integral since it is only called on fractional and yet unrounded variables. 
The property (B3) in~\cite[Lemma 2.1]{byrka2014improved} implies that the output variables $x_m$ and $x_{m'}$ satisfy the following inequality:
%obtained from $y_m$ and $y_{m'}$ returned by the \Simplify{} (elementary rounding subroutine of \DepRound) satisfy for $x_m, x_{m'} \in \{0,1\}$ and $y_m, y_{m'} \in [0,1]$  $\forall m,m' \in S$:
\begin{align}
    \mathbb{E}[x_m x_{m'}] &\leq y_m y_{m'}.
\end{align}
We have for any $c_m, c_{m'} \in [0,1]$:
\begin{align}
    \mathbb{E}[(1-x_m c_m) (1 - x_{m'} c_{m'})] &=\mathbb{E}[1 - x_{m} c_m - x_{m'}c_{m'} + x_{m} x_{m'} c_{m} c_{m'}] \label{eq:eq-1} \\
    &= 1 - y_{m} c_{m} - y_{m'} c_{m'} + \mathbb{E}[x_{m} x_{m'}] c_{m} c_{m'}  \label{eq:eq-2} \\
     &\leq1 - y_{m} c_{m} - y_{m'} c_{m'} +  y_{m} y_{m'} c_{m} c_{m'}   \nonumber \\
    &=(1-y_{m} c_{m}) (1 - y_{m'} c_{m'}) .\label{eq:expectation_dependent_vars}
\end{align}
where the second equality is obtained recalling that, by construction, $\mathbb{E}[\alloc^v_m]=y_m,\forall m\in\mathcal{M}$ (Sec.~\ref{subsec:rounding}).

Thus, the two fractional inputs $y_{m}$ and $y_{m'}$ to the \Simplify{} subroutine, return $x_{m}$ and $x_{m'}$ satisfying the property \eqref{eq:expectation_dependent_vars}. By induction as in the proof in~\cite[Lemma 2.2]{byrka2014improved}, we obtain for any $S \subset \modelSet$:
\begin{align}
  \mathbb{E}\left[ \prod_{m \in S} (1 - x_{m} c_{m})\right] \leq \prod_{m \in S} (1 - y_{m}  c_{m}).
\end{align}
Note that the above property is satisfied with equality if the components of  $\vec x \in \{0,1\}^{|\modelSet|}$ are sampled independently with  $\mathbb{E}[x_{m}] = y_{m}$. 
%Independent sampling satisfies the above property with equality for $x \in \{0,1\}^{|\modelSet|}$, since each $\mathbb{E}[x_{m}] = y_{m}$ and the random variables $x_{m}, x_{m'}$ are independent for $m \neq m' $.
\end{proof}

\begin{lemma}
\label{lemma:el_last}
Let the allocation $\vec x^v$ be the random output of { \DepRound} on node $v\in\vertices$ given the fractional allocation $\vec y^v \in \allocSetFrac^v$.  The following holds
\begin{align}
    \mathbb{E}\left[\El(\requestBatchVec_t,\loadVec_t, \allocVec_t)\right] \geq \El(\requestBatchVec_t,\loadVec_t, \allocVecFrac_t).
    \label{eq:El_lowerbound}
\end{align}
\end{lemma}
\begin{proof}
{
For any $k'$, assume $m\in\modelSet$ and $v\in\vertices$ are such that $\kappa_{\request} (v, m)=k'$ (Sec.~\ref{subsec:serving}). }Since $\auxalloc^{k'}_{  \request}(\loadVec_t, \allocVecFrac) = \allocFrac^{v}_{m} \load^{t,v}_{\request,m}$ for a given $(m,v) \in \modelSet \times \vertices$, then $\auxalloc_{\request, k'}(\loadVec_t, \allocVecFrac) / r^t_\rho$ can be written as:
\begin{align}
    \auxalloc_{
    \request, k'}(\loadVec_t, \allocVecFrac) / r^t_\rho = \frac{\load^{t,v}_{\request,m}}{r^t_\rho} y_m^v.
    \label{eq:simplified_surrogate}
\end{align}
where $\frac{\load^{t,v}_{\request,m}}{r^t_\rho}$ is a constant in $[0,1]$ ($\load^{t,v}_{\request,m} \leq \min\{r^t_\rho, L_m^v\}$ - see Sec.~\ref{subsec:load}) {that scales} variable $\allocFrac^{v}_{m}$
% We have the following
% \begin{align}
%     \left(r^t_\rho - \Auxalloc^k_\request(\requestBatchVec_t,\loadVec_t, \repoVec)\right)  = \mathds{1}_{\{\Auxalloc^k_\request(\requestBatchVec_t,\loadVec_t, \repoVec) = 0\}} r^t_\rho.
%     \label{eq:shortcut_b}
% \end{align}
% Recall from Eq.~\eqref{eq:simplified_surrogate} that $\auxalloc_{
%     \request, k'}(\loadVec_t, \allocVecFrac_t)=l^{t,v}_{\rho,m} \allocFrac^v_{t,m}$ for a given $(v,m)$ such that $\kappa_\rho(v,m) = k'$, and 
%     $l^{t,v}_{\rho,m} \leq r^t_\rho$
;therefore, by applying Lemma~\ref{lemma:depround}, we obtain the following upper bound on the bounding function. Consider for all $v\in \vertices$ and $t \in [T]$ that $\vec{x}^v_t$ is the random allocation obtained by running \DepRound{} on the fractional allocation $\vec{y}^v_t$, then

\begin{align}
    \mathbb{E}\left[\El(\requestBatchVec_t,\loadVec_t, \allocVec_t)\right] & \stackrel{\eqref{eq:gain-surrogate}}{=}
     \mathbb{E} \left[\sum_{\request \in \supp{\vec r^t}}\sum_{k=1}^{\modelsNo_{\request}-1}   \left(\smallCost^{k+1}_{\request} - \smallCost^{k}_{\request}\right)
     \requestBatch^t_\request 
     \left(1 - \prod_{k'=1}^{k} \left(1 - \auxalloc_{
    \request, k'}(\loadVec_t, \allocVec_t) / r^t_\rho \right) \right) \mathds{1}_{\{\Auxalloc^k_\request(\requestBatchVec_t,\loadVec_t, \repoVec)  = 0\}}\right] \\
    &= \sum_{\request \in \supp{\vec r^t}}\sum_{k=1}^{\modelsNo_{\request}-1}   \left(\smallCost^{k+1}_{\request} - \smallCost^{k}_{\request}\right)
     \requestBatch^t_\request 
     \left(1 - \mathbb{E} \left[\prod_{k'=1}^{k} \left(1 - \auxalloc_{
    \request, k'}(\loadVec_t, \allocVec_t) / r^t_\rho \right)\right] \right) \mathds{1}_{\{\Auxalloc^k_\request(\requestBatchVec_t,\loadVec_t, \repoVec)  = 0\}} \\
    &\geq \sum_{\request \in \supp{\vec r^t}}\sum_{k=1}^{\modelsNo_{\request}-1}   \left(\smallCost^{k+1}_{\request} - \smallCost^{k}_{\request}\right)
     \requestBatch^t_\request 
     \left(1 - \prod_{k'=1}^{k} \left(1 - \auxalloc_{
    \request, k'}(\loadVec_t, \allocVecFrac_t) / r^t_\rho \right)\right) \mathds{1}_{\{\Auxalloc^k_\request(\requestBatchVec_t,\loadVec_t, \repoVec)  = 0\}}\\
    &= \El(\requestBatchVec_t,\loadVec_t, \allocVecFrac_t).
\end{align}
The equality is obtained using the linearity of the expectation, and the inequality is obtained by applying directly Lemma~\ref{lemma:depround}.
\end{proof}
\section{Proof of Theorem \ref{th:regret_bound}}
\label{appendix:regret_proof}

\begin{proof}To prove the $\psi$-regret guarantee: \emph{(i)} we first  establish an upper bound on the regret of the \AlgoName{} policy over its fractional allocations domain $\allocSetFrac$ against a fractional optimum, then \emph{(ii)} we 
%then 
use it to derive a corresponding $\psi$-regret guarantee over the integral allocations domain $\allocSet$. % in expectation.
%guarantee is transferred to $\psi$-regret guarantee over the integral allocations domain $\allocSet$ in expectation.

\noindent\textbf{Fractional domain regret guarantee.} To establish the regret guarantee of running Algorithm~\ref{algo:idn} at the level of each computing node $v \in \vertices$, we showed that the following properties hold:
\begin{enumerate}
    \item The function $G$ is concave over its domain $\allocSetFrac$ (Lemma~\ref{lemma:concavity}).
    \item  The mirror map $\Phi:\mathcal{D} \to \mathbb{R}$ is $\theta$-strongly convex w.r.t. the norm $\norm{\,\cdot\,}_{l_1(\vec{s})}$ over $\allocSetFrac \cap \mathcal{D}$, where $\theta$ is equal to Eq.~\eqref{eq:theta} (Lemma~\ref{lemma:strong_convexity}).
    \item The gain function $G:\mathcal{Y} \to \mathbb{R}$ is $\sigma$-Lipchitz w.r.t $\norm{\,\cdot\,}_{l_1(\vec s)}$: the subgradients are bounded under the norm $\norm{\cdot}_{l_\infty (\frac{1}{\vec{s}})}$  by $\sigma$, i.e., the subgradient of $G (\vec r_t, \vec l_t, \vec y)$ at point $\vec y_t \in \allocSetFrac$ is upper bounded ($\norm{\vec{g}_t}_{l_\infty (\frac{1}{\vec{s}})} \leq \sigma$) for any $(\vec r_t, \vec l_t) \in \advSet$  (Lemma~\ref{lemma:subgradient_bound}).
    
    \item $\norm{\,\cdot\,}_{l_\infty(\frac{1}{\vec{s}})}$ is the dual norm of $\norm{\,\cdot\,}_{l_1(\vec{s})}$ (Lemma~\ref{lemma:dual_norm}).
    \item The Bregman divergence $D_\Phi(\allocVecFrac_*, \allocVecFrac_1)$ in Eq.~\eqref{def:bregman_global_mirrormap} is upper bounded by a constant $D_{\max}$ where $\allocVecFrac_* = {\arg\max}_{\allocVecFrac \in \allocSetFrac} \sum^T_{t=1} \systemGain(\requestBatchVec_t,\loadVec_t, \allocVecFrac)${ and $\allocVecFrac_1 = \underset{\allocVecFrac  \in \allocSetFrac \cap \mathcal{D}}{\arg\min} \,\Phi(\allocVecFrac)$ is the initial allocation} (Lemma~\ref{lemma:bregman_bound}).
\end{enumerate}

% By taking $L_{\max} \triangleq \max\{L^v_m, \forall (v,m) \in \vertices \times \modelSet\}$, then the set of all possible potential capacities\footnote{To simplify the analysis w.l.g. we allowed the adversary to select potential capacities from the larger set $\loadSet$. The constraints $\sum_{\request \in \requestSet^i} l^v_{\request, m } \leq  L^v_{m}$ are ignored, and replaced by the simplified constraints $l^v_{\request, m } \leq L_{\max}$.}  denoted by $\loadSet$ can be relaxed as
% \begin{align}
%     \mathcal{L} \triangleq \left\{ \vec{l} \in \naturals^{\bigcup_{i \in \taskcatalog} \requestSet_i \times \modelSet_i \times \vertices} :  l^v_{\request, m } \leq L_{\max} \right\}.
% \end{align}

% Furthermore, we assume that the system receives a finite number of requests for any time slot (see Eq.~\eqref{eq:repo_feasibility_constraint}); therefore, the set of all possible request batches is defined as 

% \begin{align}
%     \mathcal{B} \triangleq \left\{ \vec{r} \in \naturals^\requestSet: \vec{r}~\mathrm{satisfies~Eq.~\eqref{eq:repo_feasibility_constraint}}\right\}.
% \end{align}

Because of properties 1--5 above, the following bound holds for the regret of \AlgoName{} over its fractional domain $\allocSetFrac$ (vector field point of view of Mirror Descent in \cite[Sec.~4.2]{bubeck2015convexbook} combined with~\cite[Theorem 4.2]{bubeck2015convexbook}):
\begin{align}
\mathrm{Regret}_{T, \allocSetFrac} &= \underset{
{{\{(\requestBatchVec_t, \loadVec_t) \}_{t=1}^{T} \in \advSet^T}}
}{\sup}\hspace{-0.1em}
    \left\{ \sum^T_{t=1}    \systemGain(\requestBatchVec_t, \loadVec_t, \allocVecFrac_*){-}\sum^T_{t=1}    \systemGain(\requestBatchVec_t, \loadVec_t, \allocVecFrac_t) \right\}\hspace{-0.1em}
    \\
    &\leq \frac{D_\Phi(\allocVecFrac_*, \allocVecFrac_1)}{\eta} + \frac{\eta}{2\theta} \sum^T_{t=1} \norm{\vec{g}_t}^2_{l_\infty (\frac{1}{\vec{s}})} \leq \frac{D_{\max}}{\eta} + \frac{\eta \sigma^2 T}{2\theta}.
\end{align}
where $\eta$ is the learning rate of INFIDA (Algorithm~\ref{algo:idn}, line~\ref{ln:update}).
By selecting the learning rate $\eta=\frac{1}{\sigma}\sqrt{\frac{2\theta D_{\max}}{T}}$ giving the tightest upper bound we obtain 
\begin{align}
    \mathrm{Regret}_{T, \allocSetFrac} \leq \sigma \sqrt{\frac{2D_{\max}}{\theta} T}.
    \label{eq:fractional_regret_bound}
\end{align}

\noindent\textbf{Integral domain regret guarantee.} Note that, by restricting the maximization to the subset of integral allocations $\allocVec\in\allocSet$, the optimal allocation $\allocVec_* = {\arg\max}_{\allocVec \in \allocSet} \sum^T_{t=1} \systemGain(\requestBatchVec_t,\loadVec_t, \allocVec)$ can only lead to a lower gain, i.e.,
\begin{align}
   \sum^T_{t=1} \systemGain(\requestBatchVec_t,\loadVec_t,\allocVec_*)
  \leq
  \sum^T_{t=1} \systemGain(\requestBatchVec_t,\loadVec_t, \allocVecFrac_*). 
  \label{eq:optimality_in_convx}
\end{align}

By taking $\psi = 1 - \frac{1}{e}$, and using the bounding function $\El$ defined in Eq.~\eqref{eq:gain-surrogate}, with the expectation taken over the random choices of the policy (\DepRound{} at line 8 in Algorithm~\ref{algo:idn}) we obtain
\begin{align}
    \mathbb{E} \left[\sum^T_{t=1} \systemGain(\requestBatchVec_t,\loadVec_t,\allocVec_t)\right] 
    &\stackrel{\eqref{eq:gain_lower_upper_bound}}{\geq} \mathbb{E} \left[\sum^T_{t=1} \El(\requestBatchVec_t,\loadVec_t,\allocVec_t)\right] \nonumber 
    \stackrel{\eqref{eq:El_lowerbound}}{\geq} \sum^T_{t=1} \El(\requestBatchVec_t,\loadVec_t,\allocVecFrac_t)
    \stackrel{\eqref{eq:gain_lower_upper_bound}}{\geq} \psi \sum^T_{t=1} \systemGain(\requestBatchVec_t,\loadVec_t,\allocVecFrac_t) \nonumber \\
    &\stackrel{\eqref{eq:fractional_regret_bound}}{\geq} \psi \sum^T_{t=1} \systemGain(\requestBatchVec_t,\loadVec_t,\allocVecFrac_*) - \psi \sigma \sqrt{\frac{2D_{\max}}{\theta} T}
    \stackrel{\eqref{eq:optimality_in_convx}}{\geq} \psi \sum^T_{t=1} \systemGain(\requestBatchVec_t,\loadVec_t,\allocVec_*) - \psi \sigma \sqrt{\frac{2D_{\max}}{\theta} T}.
    \label{eq:proof_main_theorem}
\end{align}
Thus, we have
\begin{align}
   \psi \sum^T_{t=1} \systemGain(\requestBatchVec_t,\loadVec_t,\allocVec_*)  - \mathbb{E} \left[\sum^T_{t=1} \systemGain(\requestBatchVec_t,\loadVec_t,\allocVec_t)\right] \leq \psi \sigma \sqrt{\frac{2D_{\max}}{\theta} T}.
\end{align}
The above inequality holds for any sequence $\{(\requestBatchVec_t, \loadVec_t) \}_{t=1}^{T} \in \advSet^T$. Thus, the $\psi$-regret is given by
\begin{align}
    \psi \text{-} \mathrm{Regret}_{T, \allocSet} \stackrel{\eqref{eq:regret-definition}}{=}  \underset{
{{\{\requestBatchVec_t, \loadVec_t \}_{t=1}^{T} \in \advSet^T}}
}{\sup}\hspace{-0.1em} \left\{\psi \sum^T_{t=1} \systemGain(\requestBatchVec_t,\loadVec_t,\allocVec_*) - \mathbb{E} \left[\sum^T_{t=1} \systemGain(\requestBatchVec_t,\loadVec_t,\allocVec_t)\right] \right\} \leq A \sqrt{T},
\end{align}
where
$$A =\psi \sigma \sqrt{\frac{2D_{\max}}{\theta}} =  \psi   \frac{R L_{\max} \Delta_C}{s_{\min} } \sqrt{s_{\max} |\vertices| |\modelSet| }\sqrt{2 \sum_{v \in \vertices} \min\{\capacity^v,\norm{\vec{s}^v}_1\} \log\left(\frac{\norm{\vec{s}^v}_1}{\min\{\capacity^v,\norm{\vec{s}^v}_1\}}\right)  }, $$
using the upper bounds on $\theta$, $\sigma$, and $D_{\max}$ determined in Lemmas \ref{lemma:strong_convexity}, \ref{lemma:subgradient_bound}, and \ref{lemma:bregman_bound}, respectively.

This proves Theorem~\ref{th:regret_bound}.
\end{proof}

\section{Proof of Proposition~\ref{corollary:offline_solution}}
\label{proof:offline_solution}
\begin{proof}
Let $\bar{\allocVecFrac}$ be the average fractional allocation $\bar{\allocVecFrac} =\frac{1}{\tilde T} \sum^{\tilde T}_{t=1} \allocVecFrac_t$ of \textnormal{\AlgoName{}}, and $\bar{\allocVec}$ the random state sampled from $\bar{\allocVecFrac}$ using \textnormal{\DepRound{}}. We take $G_T(\allocVecFrac) = \frac{1}{{T}} \sum^{{T}}_{{t}=1} \systemGain(\requestBatchVec_{{t}},\loadVec_{{t}},{\allocVecFrac}), \forall \allocVecFrac \in \allocSetFrac$. We have
\begin{align}
\mathbb{E} \left[ G_T(\bar{\allocVec})\right] 
    &\stackrel{\eqref{eq:gain_lower_upper_bound}}{\geq} \mathbb{E} \left[\frac{1}{{T}}\sum^{{T}}_{{t}=1} \El(\requestBatchVec_{{t}},\loadVec_{{t}},\bar{\allocVec})\right]  
    \stackrel{\eqref{eq:El_lowerbound}}{\geq} \frac{1}{{T}}\sum^{{T}}_{{t}=1} \El(\requestBatchVec_{{t}},\loadVec_{{t}},\bar{\allocVecFrac})
    \stackrel{\eqref{eq:gain_lower_upper_bound}}{\geq} \psi G_T(\bar{\allocVecFrac}).
    \label{eq:p1_corollary}
\end{align}
Using Jensen's inequality we get 
\begin{align}
  \psi G_T(\bar{\allocVecFrac})  \geq  \psi \frac{1}{\tilde T} \sum^{\tilde T}_{t=1} G_T(\allocVecFrac_t).
    \label{eq:p2_corollary}
\end{align}
\end{proof}
It straightforward to check that $G_T$ satisfies the same properties 1 (concavity) and 3 (subgradient boundedness) as $G$ and the remaining properties are preserved under the same mirror map and convex decision set. With properties 1--5 satisfied, we can apply \cite[Theorem 4.2]{bubeck2015convexbook} to obtain 
\begin{align}
  \sum^{\tilde T}_{t=1}  G_T(\allocVecFrac_*)  -   \sum^{\tilde T}_{t=1} G_T(\allocVecFrac_t) =  \tilde T G_T(\allocVecFrac_*)  -   \sum^{\tilde T}_{t=1} G_T(\allocVecFrac_t)  \leq   \sigma \sqrt{\frac{2D_{\max}}{\theta} \tilde T}.
\end{align}
Dividing both sides of the above inequality by $\tilde T$ gives
\begin{align}
  \frac{1}{\tilde T} \sum^{\tilde T}_{t=1} G_T(\allocVecFrac_t)  \geq G_T(\allocVecFrac_*)-\sigma \sqrt{\frac{2D_{\max}}{\theta \tilde T} }.
\end{align}
Using the same argument to obtain Eq.~\eqref{eq:optimality_in_convx}, i.e., restricting the maximization to the integral domain gives a lower value, we get 
\begin{align}
  \frac{1}{\tilde T} \sum^{\tilde T}_{t=1} G_T(\allocVecFrac_t)  \geq G_T(\allocVec_*)-\sigma \sqrt{\frac{2D_{\max}}{\theta \tilde T} }.
  \label{eq:new_part_bound_corollary}
\end{align}
Using Eq.~\eqref{eq:p1_corollary}, and Eq.~\eqref{eq:new_part_bound_corollary} we obtain
\begin{align}
   \mathbb{E} \left[ G_T(\bar{\allocVec})\right]  \geq  \psi G_T(\allocVec_*) - \psi \sigma \sqrt{\frac{2D}{\theta \tilde T} }.
\end{align}
Thus, $\forall \epsilon > 0$ and over a sufficiently large running time $\tilde T$ for \AlgoName{}, $\bar{\allocVec}$ satisfies %the following:
\begin{align}
\mathbb{E} \left[ G_T(\bar{\allocVec}) \right] \geq \left(1-\frac{1}{e} - \epsilon\right) G_T({\allocVec}_*) .
\end{align}

\fi

\end{document}